\documentclass[10pt]{IEEEtran}
\usepackage[latin9]{inputenc}
\usepackage{array}
\usepackage{verbatim}
\usepackage{float}
\usepackage{amsmath}
\usepackage{amssymb}
\usepackage{graphicx}
\usepackage{esint,hyperref}

\makeatletter

\providecommand{\tabularnewline}{\\}

 \usepackage{color}

\usepackage{cite}\usepackage{amsthm}\usepackage{dsfont}\usepackage{array}\usepackage{mathrsfs}\usepackage{comment}

\newcommand{\cA}{{\mathcal A}}

\newcommand{\bA}{{\boldsymbol A}}
\newcommand{\ba}{{\boldsymbol a}}

\newcommand{\bfeta}{{\boldsymbol \eta}}
\newcommand{\bSigma}{{\boldsymbol \Sigma}}

\newcommand{\bx}{{\boldsymbol x}}
\newcommand{\by}{{\boldsymbol y}}

\newcommand{\bH}{{\boldsymbol H}}

\newcommand{\bX}{{\boldsymbol X}}

 \DeclareMathOperator{\trace}{Tr}

 \DeclareMathOperator{\diag}{diag}
  
  \DeclareMathOperator{\argmin}{argmin}  

\title{Exact and Stable Covariance Estimation\\ from Quadratic Sampling via Convex Programming}

\makeatother

\begin{document}
\theoremstyle{plain}\newtheorem{lemma}{\textbf{Lemma}}\newtheorem{theorem}{\textbf{Theorem}}\newtheorem{fact}{\textbf{Fact}}\newtheorem{corollary}{\textbf{Corollary}}\newtheorem{assumtion}{\textbf{Assumption}}\newtheorem{example}{\textbf{Example}}\newtheorem{definition}{\textbf{Definition}}\newtheorem{proposition}{\textbf{Proposition}}

\theoremstyle{definition}

\theoremstyle{remark}\newtheorem{remark}{\textbf{Remark}}

\bibliographystyle{IEEEtran}

\global\long\def\footnotemark{}

%\date{October 2013; Revised: March 2015}
%\date{The date}

\author{Yuxin Chen, Yuejie Chi, and Andrea J. Goldsmith %
\thanks{Y. Chen is with the Department of Statistics, Stanford
University, Stanford, CA 94305, USA (email: yxchen@stanford.edu). %
} %
\thanks{Y. Chi is with the Department of Electrical and Computer Engineering,
The Ohio State University, Columbus, OH 43210, USA (email: chi.97@osu.edu). %
} %
\thanks{A. J. Goldsmith is with the Department of Electrical Engineering,
Stanford University, Stanford, CA 94305, USA (email: andrea@ee.stanford.edu). %
} %
\thanks{This paper was presented in part at the International Conference on
Acoustics, Speech, and Signal Processing (ICASSP) \cite{chen2014icassp}
and the International Symposium on Information Theory (ISIT) \cite{chen2014isit}.%
}
\thanks{
Copyright (c) 2014 IEEE. Personal use of this material is permitted.  However, permission to use this material for any other purposes must be obtained from the IEEE by sending a request to pubs-permissions@ieee.org.}
}
\maketitle

\begin{abstract}
Statistical inference and information processing of high-dimensional
data often require efficient and accurate estimation of their second-order
statistics. With rapidly changing data, limited processing power and
storage at the acquisition devices, it is desirable to extract
the covariance structure from a single pass over the data and
a small number of stored measurements. In this paper, we explore a quadratic
(or rank-one) measurement model which imposes minimal memory requirements
and low computational complexity during the sampling process, and
is shown to be optimal in preserving various low-dimensional covariance structures.
Specifically, four popular structural assumptions of covariance matrices,
namely low rank, Toeplitz low rank, sparsity, jointly rank-one and
sparse structure, are investigated, while recovery is achieved via convex
relaxation paradigms for the respective structure.

The proposed quadratic sampling framework has a variety of potential
applications including streaming data processing, high-frequency wireless
communication, phase space tomography and phase retrieval in optics, and non-coherent
subspace detection. Our method admits universally accurate covariance
estimation in the absence of noise, as soon as the number of measurements
exceeds the information theoretic limits. We also demonstrate the
robustness of this approach against noise and imperfect structural
assumptions. Our analysis is established upon a novel notion called
the mixed-norm restricted isometry property (RIP-$\ell_{2}/\ell_{1}$),
as well as the conventional RIP-$\ell_{2}/\ell_{2}$ for near-isotropic
and bounded measurements. In addition, our results improve upon the
best-known phase retrieval (including both dense and sparse signals)
guarantees using PhaseLift with a significantly simpler approach. 
\end{abstract}
%\begin{IEEEkeywords}Quadratic measurements, data
%stream, covariance sketching, energy measurements, phase retrieval, phase tomography,
%RIP-$\ell_{2}/\ell_{1}$, nearly-isotropic operators, Toeplitz, low
%rank, sparsity \end{IEEEkeywords}


\section{Introduction}

Accurate estimation of second-order statistics of stochastic processes
and data streams is of ever-growing importance to various applications
that exhibit high dimensionality. Covariance estimation is the cornerstone
of modern statistical analysis and information processing, as the
covariance matrix constitutes the sufficient statistics to many signal
processing tasks, and is particularly crucial for extracting reduced-dimension
representation of the objects of interest. For signals and data streams
of high dimensionality, there might be limited memory and computation
power available at the data acquisition devices to process the rapidly
changing input, which requires the covariance estimation task to be
performed with a single pass over the data stream, minimal storage,
and low computational complexity. This is not possible unless appropriate
structural assumptions are incorporated into the high-dimensional
problems. Fortunately, a broad class of high-dimensional signals indeed
possesses low-dimensional structures, and the intrinsic dimension
of the covariance matrix is often far smaller than the ambient dimension.
For different types of data, the covariance matrix may exhibit different
structures; four of the most widely considered structures are listed
below. 
\begin{itemize}
\itemsep0.3em
\item \textit{Low Rank:} The covariance matrix is (approximately) low-rank,
which occurs when a small number of components accounts for most of
the variability in the data. Low-rank covariance matrices arise in
applications including traffic data monitoring, array signal processing,
collaborative filtering, and metric learning. 
\item \textit{Stationarity and Low Rank:} The covariance matrix is simultaneously
low-rank and Toeplitz, which arises when the random process is generated
by a few spectral spikes. Recovery of the stationary covariance matrix,
often equivalent to spectral estimation, is crucial in many tasks
in wireless communications (e.g. detecting spectral holes in cognitive
radio networks), and array signal processing (e.g. direction-of-arrival
analysis \cite{RoyKailathESPIRIT1989}). 
\item \textit{Sparsity:} The covariance matrix can be approximated in a
sparse form \cite{karoui2008operator}. This arises when a large number
of variables have small pairwise correlation, or when several variables
are mutually exclusive. Sparse covariance matrices arise in finance,
biology and spectrum estimation. %A partial list of applications includes cyclic spectrum sensing \cite{}, finance, and 

\item \textit{Joint Sparsity and Rank-One:} The covariance matrix can be
approximated by a jointly sparse and rank-one matrix. This has received
much attention in recent development of sparse PCA, and is closely
related to sparse signal recovery from magnitude measurements (called
\emph{sparse phase retrieval}). 
\end{itemize}
%This gives rise to the natural question as to how to effectively exploit the information structure and obtain faithful covariance estimation, provided that we are only allowed to obtain low-rate and noisy samples.  

In this paper, we wish to reconstruct an unknown covariance matrix
$\bSigma\in\mathbb{R}^{n\times n}$ with the above structure from
a small number of rank-one measurements. In particular, we explore
 sampling methods of the form 
\begin{align}
y_{i} & =\ba_{i}^{\top}\bSigma\ba_{i}+\eta_{i},\quad i=1,\ldots,m,\label{measurements}
\end{align}
where $\boldsymbol{y}:=\left\{ y_{i}\right\} _{i=1}^{m}$ denotes
the measurements, $\ba_{i}\in\mathbb{R}^{n}$ represents the sensing
vector, $\bfeta:=\left\{ \eta_{i}\right\} _{i=1}^{m}$ stands for
the noise term, and $m$ is the number of measurements. The noise-free
measurements $\ba_{i}^{\top}\bSigma\ba_{i}$'s are henceforth referred to as {\em quadratic measurements} (or
\emph{rank-one measurements}). In practice, the number of measurements
one can obtain is constrained by the storage requirement in data acquisition,
which could be much smaller than the ambient dimension of $\bSigma$.
This sampling scheme finds applications in a wide spectrum of practical
scenarios, admits optimal covariance estimation with tractable algorithms,
and brings in computational and storage advantages in comparison with
other types of measurements, as detailed in the rest of the paper.

\subsection{Motivation}

The quadratic measurements in the form of \eqref{measurements} are
motivated by several application scenarios listed below, which illustrate
the practicability and benefits of the proposed quadratic measurement
scheme.

\subsubsection{Covariance Sketching for Data Streams}

\label{sketching_stream}

A high-dimensional data stream model represents \emph{real-time} data
that arrives sequentially at a high rate, where each data instance
is itself high-dimensional. In many resource-constrained applications,
the available memory and processing power at the data acquisition
devices are severely limited compared with the volume and rate of
the data \cite{muthukrishnan2005data}. Therefore it is desirable
to extract the covariance matrix of the data instances from inputs
on the fly without storing the whole stream. Interestingly, the quadratic
measurement strategy can be leveraged as an effective data stream
processing method to extract the covariance information from real-time
data, with limited memory and low computational complexity.

Specifically, consider an input stream $\{\boldsymbol{x}_{t}\}_{t=1}^{\infty}$
that arrives sequentially, where each $\boldsymbol{x}_{t}\in\mathbb{R}^{n}$
is a high-dimensional data instance generated at time $t$. The goal
is to estimate the covariance matrix $\boldsymbol{\Sigma}=\mathbb{E}[\boldsymbol{x}_{t}\boldsymbol{x}_{t}^{\top}]\in\mathbb{R}^{n\times n}$.
The prohibitively high rate at which data is generated forces covariance
extraction to function with as small a memory as possible. The scenario
we consider is quite general, and we only impose that the covariance
of a random substream of the original data stream converges to the
true covariance $\bSigma$. No prior information on the correlation
statistics across consecutive instances is assumed to be known {\em
a priori} (e.g. they are not necessarily independently drawn), and
hence it is not feasible to exploit these statistics to enable lower
sample complexity.

We propose to pool the data stream $\{\boldsymbol{x}_{t}\}_{t=1}^{\infty}$
into a small set of measurements in an easy-to-adapt fashion with
a collection of sketching vectors $\{\ba_{i}\}_{i=1}^{m}$. Our covariance
sketching method, termed \emph{quadratic sketching}, is outlined as
follows: 
\begin{enumerate}
\itemsep0.3em
\item At each time $t$, we randomly choose a sketching vector indexed by
$\ell_{t}\in\{1,\ldots,m\}$, and obtain a single nonnegative \emph{quadratic
sketch} $(\boldsymbol{a}_{\ell_{t}}^{\top}\boldsymbol{x}_{t})^{2}$. 
\item All sketches employing the same sketching vector $\ba_{i}$ are aggregated
and normalized, which converge \emph{rapidly} to a {\em measurement}%
\footnote{Note that we might only be able to obtain measurements for empirical
covariance matrices instead of $\bSigma$, but this inaccuracy can
be absorbed into the noise term $\bfeta$. In fact, for stationary
data streams, $y_{i}$ converges rapidly to $\ba_{i}^{\top}\bSigma\ba_{i}$
with {\em a few} instances $\boldsymbol{x}_{t}$. %
} 
\begin{align}
y_{i} & =\mathbb{E}\left[(\boldsymbol{a}_{i}^{\top}\boldsymbol{x}_{t})^{2}\right]+\bfeta_{i}=\ba_{i}^{\top}\mathbb{E}\left[\bx_{t}\bx_{t}^{\top}\right]\ba_{i}+\bfeta_{i}\nonumber \\
 & =\boldsymbol{a}_{i}^{\top}\boldsymbol{\Sigma}\boldsymbol{a}_{i}+\bfeta_{i},\quad\quad\quad i=1,\ldots,m,\label{eq:DataStream}
\end{align}
where $\bfeta:=\left\{ \eta_{i}\right\} _{i=1}^{m}$ denotes the error
term. 
\end{enumerate}
There are several benefits of this covariance sketching method. First,
the storage complexity $m$, as will be shown, can be much smaller
than the ambient dimension of $\bSigma$. The computational cost for
sketching each instance is linear with respect to the dimension of
the instance in the data stream. Unlike the uncompressed sketching
methods where each instance one measures usually affects many stored
measurements, our scheme allows each aggregate quadratic sketch to
be composed by completely different instances, which allows sketching
to be performed in a distributed and asynchronous manner. This arises
since each randomized sketch is a compressive snapshot of the second-order
statistics, while each uncompressed measurement itself is unable to
capture the correlation information. As we will demonstrate later,
this sketching scheme allows optimal covariance estimation with information
theoretically minimal memory complexity at the data acquisition stage.
One motivating application for this covariance sketching method is
covariance estimation of ultra-wideband random processes, as is further
elaborated in Section~\ref{sec:noncoherentsignal}.

\subsubsection{Noncoherent Energy Measurements in Communications and Signal Processing}

\label{sec:noncoherentsignal} When communication takes place in the
high-frequency regime, \emph{empirical} energy measurements are often
more accurate and cheaper to obtain than phase measurements. For instance,
energy measurements will be more reliable when communication systems
are operating with extremely high carrier frequencies (e.g.
60GHz communication systems \cite{daniels200760}). % and also result in popular noncoherent detection methods which do not require prior information on the transmitted signal. 

\begin{itemize}
\itemsep0.3em
\item {\emph{Spectrum Estimation of Stochastic Processes from Energy Measurements:}
} Many wireless communication systems operating in stochastic environments
rely on reliable estimation of the spectral characteristics of random
processes \cite{jung2013compressive}, such as recovering the power
spectrum of the ultra-wideband random process characterizing the spectrum
occupancy in cognitive radio \cite{Leus2011,ariananda2012compressive}.
Moreover, optimal signal transmissions are often based on the Karhunen\textendash{}Loeve
decomposition of a random process, which requires accurate covariance
information \cite{gallager1968information}. If one employs a sensing
vector $\boldsymbol{a}_{i}$, which is implementable using random
demodulators \cite{tropp2010beyond}, and observes the average energy
measurements over $N$ instances $\left\{ \boldsymbol{x}_{t}\right\} _{1\leq t\leq N}$,
then the energy measurements read
\begin{equation}
y_{i}=\frac{1}{N}\sum_{t=1}^{N}\left|\ba_{i}^{\top}\bx_{t}\right|^{2}=\ba_{i}^{\top}\bSigma_{N}\ba_{i},\quad i=1,\ldots,m\label{sample-measurement}
\end{equation}
where $\bSigma_{N}:=\frac{1}{N}\sum_{t=1}^{N}\bx_{t}\bx_{t}^{\top}$
denotes the sample covariance matrix, leading to the quadratic-form
observations. 
\item {\emph{Noncoherent Subspace Detection from Energy Measurements:}}
Matched subspace detection \cite{scharf1994matched} spans many applications
in wireless communication, radar, and pattern recognition when the
transmitted signal is encoded by the subspaces. The problem can also
be cast as recovering the principal subspace of a dataset $\{\bx_{t}\}_{t=1}^{N}$,
with an energy detector obtaining $m$ measurements in the form of
\eqref{sample-measurement}. Thus, the noncoherent subspace detection
is subsumed by the formulation (\ref{measurements}). 
\end{itemize}

\subsubsection{Phaseless Measurements in Physics}

Optical imaging devices are incapable of acquiring phase measurements
due to ultra-high frequencies associated with light. In many applications,
measurements taking the form of \eqref{measurements} arise naturally. 
\begin{itemize}
\itemsep0.3em
\item {\emph{Compressive Phase Space Tomography:}} Phase Space Tomography
\cite{raymer1994complex} is an appealing method to measure the correlation
function of a wave field in physics. However, tomography becomes challenging
when the dimensionality of the correlation matrix becomes large. Recently,
it was proposed experimentally in \cite{tian2012experimental} to
recover an approximately low-rank correlation matrix, which often
holds in physics, by only taking a small number of measurements in
the form of \eqref{measurements}. 

\item {\emph{Phase Retrieval:}} Due to the physical constraints, one can
only measure amplitudes of the Fourier coefficients of an optical
object. This gives rise to the problem of recovering a signal $\boldsymbol{x}\in\mathbb{R}^{n}$
from magnitude measurements, which is often referred to as phase retrieval.
Several convex (e.g. \cite{candes2012phaselift,shechtman2011sparsity,waldspurger2012phase})
and nonconvex algorithms (e.g. \cite{candes2014wirtinger,shechtman2013gespar,netrapalli2013phase,Schniter2015})
have been proposed that enable exact phase retrieval (i.e. recovers
$\boldsymbol{x}\cdot\boldsymbol{x}^{\top}$) from random magnitude
measurements. If we set $\bSigma:=\boldsymbol{x}\boldsymbol{x}^{\top}$,
then our problem formulation (\ref{measurements}) subsumes phase
retrieval as a special case in the low-rank setting. 
%\item {\emph{Line Spectrum Estimation from Phaseless Measurements:}} If the signal of interest $\boldsymbol{x}\in\mathbb{C}^n$ is further a spectrally-sparse object that only contains a small number of spectrum lines in the phase retrieval problem, the lifted matrix $\bSigma$ will be simultaneously Toeplitz and rank-one.
\end{itemize}

Apart from the preceding applications, we are aware that this rank-one measurement model naturally arises in the mixture of linear regression problem \cite{chen2014convex}. All in all, all of these applications require structured matrix recovery
 from \emph{a small number }of rank-one measurements (\ref{measurements}).
The aim of this paper to develop tractable recovery algorithms that enjoy near-optimal performance guarantees.

\subsection{Contributions}

Our main contributions are three fold. First, we have developed convex
optimization algorithms for covariance estimation from a set of quadratic
measurements as given in \eqref{measurements} for a variety of structural
assumptions including low-rank, Toeplitz low-rank, sparse, and sparse
rank-one covariance matrices. The proposed algorithms exploit the
presumed low-dimensional structures using convex relaxation tailored
for respective structures. For a large class of sub-Gaussian sensing
vectors, we derive theoretical performance guarantees (Theorems~\ref{thm:ApproxLR}
-- \ref{thm:SparsePR}) from the following aspects: 
\begin{enumerate}
\itemsep0.3em
\item \textbf{Exact and universal recovery}: once the sensing vectors are
selected, then with high probability, all covariance matrices satisfying
the presumed structure can be recovered; 
\item \textbf{Stable recovery}: the proposed algorithms allow reconstruction
of the true covariance matrix to within high accuracy even under imperfect
structural assumptions; additionally, if the measurements are corrupted
by noise, possibly adversarial, the estimate deviates from the true
covariance matrix by at most a constant multiple of the noise level; 
\item \textbf{Near-minimal measurements}: the proposed algorithms succeed
as soon as the number of measurements is slightly above the information
theoretic limits for most of the respective structure. For the special
case of (sparse) rank-one matrices, our result recovers and strengthens
the best-known reconstruction guarantees of (sparse) phase retrieval
using PhaseLift \cite{candes2012phaselift,candes2012solving,li2012sparse}
with a much simpler proof technique. 
\end{enumerate}
Secondly, to obtain some of the above theoretical guarantees (Theorems~\ref{thm:ApproxLR},
\ref{thm:ApproxSP}, and \ref{thm:SparsePR}), we have introduced
a novel mixed-norm restricted isometry property, denoted by RIP-$\ell_{2}/\ell_{1}$.
An operator is said to satisfy the RIP-$\ell_{2}/\ell_{1}$ if the
strength of the signal class of interest before and after measurements
are preserved when measured in the $\ell_{2}$ norm and in the $\ell_{1}$
norm, respectively. While the conventional RIP-$\ell_{2}/\ell_{2}$
does not hold for the quadratic sensing model for general low-rank
structures as pointed out by \cite{candes2012phaselift}, we have
established that the sensing mechanism does satisfy the RIP-$\ell_{2}/\ell_{1}$
after a ``debiasing'' modification, under general low-rank, sparse,
and simultaneously sparse and rank-one structural assumptions. This
seemingly subtle change enables a significantly simpler analytical
approach without resorting to complicated dual construction as in
\cite{candes2012phaselift,candes2012solving,li2012sparse}.

On the other hand, we demonstrate, via the entropy method \cite{rudelson2008sparse},
that linear combinations of the quadratic measurements satisfy RIP-$\ell_{2}/\ell_{2}$
when restricted to \textit{Toeplitz} low-rank covariance matrices.
This leads to near-optimal recovery guarantees for Toeplitz low-rank
covariance matrices (Theorem~\ref{thm:ToeplitzPhaseLift}). Along
the way, we have also established a RIP-$\ell_{2}/\ell_{2}$ for bounded
and near-isometric operators (Theorem~\ref{thm:RIP_Isotropic}),
which strengthens previous work \cite{Gross2011recovering,liu2011universal}
by offering universal and stable recovery guarantees for a broader
class of operators including Fourier-type measurements.

Last but not least, our measurement schemes and algorithms may be
of independent interest to high-dimensional data processing. The measurements
in \eqref{measurements} are rank-one measurements with respect to
the covariance matrix, which are much easier to implement and bear
a smaller computational cost than full-rank measurement matrices with
i.i.d. entries. Moreover, the performance guarantees of the measurement
scheme \eqref{measurements} is universal, which does not require
any additional incoherence conditions on the covariance matrix as
required in the standard matrix completion framework \cite{ExactMC09,Gross2011recovering,chen2013robust}.

\subsection{Related Work}

In most existing work, the covariance matrix is estimated from a collection
of \textit{full} data samples, and fundamental guarantees have been
derived on how many samples are sufficient to approximate the ground
truth \cite{johnstone2006high,karoui2008operator}. In contrast, this
paper is motivated by the success of Compressed Sensing (CS) \cite{Don2006,CandTao06},
which asserts that compression can be achieved at the same time as
sensing without losing information. Efficient algorithms have been
developed to estimate a \textit{deterministic signal} from a much
smaller number of linear measurements that is proportional to the
complexity of the parsimonious signal model. As we will show in this
paper, covariance estimation from compressive measurements can be
highly robust.

When the covariance matrix is assumed to be approximately sparse,
recent work \cite{tian2012cyclic,Leus2011} explored reconstruction
of second-order statistics of a  cyclostationary signal from random
linear measurements, by $\ell_{1}$-minimization without performance
guarantees. Other spectral prior information has been considered as well in \cite{romero2013wideband} for stationary processes. These problem setups are quite different from \eqref{measurements} in the current work.
Another work by Dasarathy et al. \cite{dasarathy2013sketching} proposed estimating an
approximately sparse covariance matrix from measurements of the form
$\boldsymbol{Y}=\boldsymbol{A}\bSigma\boldsymbol{A}^{\top}$, where
$\boldsymbol{A}\in\mathbb{R}^{m\times n}$ denotes the sketching matrix
constructed from expander graphs. Nevertheless, this scheme cannot
accommodate low-rank covariance matrix estimation.

Our covariance estimation method is inspired by recent developments
in phase retrieval \cite{candes2012phaselift,candes2012solving,beck2012sparsity,jaganathan2012recovery,waldspurger2012phase,netrapalli2013phase},
which is tantamount to recovering rank-one covariance matrices from
quadratic measurements. In particular, our recovery algorithm coincides
with \textit{PhaseLift} \cite{candes2012phaselift,candes2012solving}
when applied to low-rank matrices. In \cite{candes2012solving}, it
is shown that PhaseLift succeeds at reconstructing a signal of dimensionality
$n$ from $\Theta(n)$ phaseless \emph{Gaussian} measurements, and
stable recovery has also been established in the presence of noise.
When specializing our result to this case, we have shown that the
same type of theoretical guarantee holds for a much larger class of
\textit{sub-Gaussian} measurements, with a different proof technique
that yields a much simpler proof. Moreover, when the signal is further
assumed to be $k$-sparse, the pioneering work \cite{li2012sparse}
showed that $O(k^{2}\log n)$ \emph{Gaussian} measurements suffice;
this result is extended to accommodate sub-Gaussian measurements and
approximately sparse signals by our framework with a much simpler
proof. More details can be found in Section~\ref{sec:PhaseRetrieval}.

We also put the proposed covariance sketching scheme in Section~\ref{sketching_stream}
into perspective. In a streaming setting, online principal component
analysis (PCA) has been an active area of research for decades \cite{oja1983subspace}
using full data samples, where non-asymptotic convergence guarantees
have only been recently developed \cite{Caramanis2013onlinePCA}.
Inspired by CS, subspace tracking from partial observations of a data
stream \cite{Balzano-2010,chi2012petrels}, which can be regarded
as a variant of incremental PCA \cite{brand2002incremental} in the
presence of missing values, is also closely related. However, existing
subspace tracking algorithms mainly aim to recover the data stream,
which is not necessary if one only cares to extract the second-order
statistics.

Finally, after we posted our work on Arxiv, Cai and Zhang made available
their manuscript \cite{cai2013rop}, an independent work that studies low-rank
matrix recovery under rank-one measurements via the notion of restricted
uniform boundedness. In comparison, our results accommodate a larger
class of covariance structures including Toeplitz low-rank, sparse,
and jointly low-rank and sparse matrices.

\subsection{Organization}

The rest of this paper is organized as follows. We first present the
convex optimization based algorithms in Section~\ref{sec:Formal-Setup},
and establish their theoretical guarantees. The analysis framework
is based upon a novel mixed-norm restricted isometry property as well
as conventional RIP for near-isotropic and bounded measurements, as
elaborated in Sections~\ref{sec:RIP} and \ref{sec:Approximate-Isometry:-RIP-22}.
The proof of main theorems is deferred to the appendices. Numerical
examples are provided in Sections~\ref{sec:numerical}. Finally,
Section \ref{sec:Conclusions-and-Future} concludes the paper with
a summary of our findings and a discussion of future directions.

%Section~\ref{sec:PhaseRetrieval} discusses sparse phase retrieval, where the proposed proof architecture is used to recover and improve upon existing results.

\subsection{Notations}

Before proceeding, we provide a brief summary of useful notations that will be
used throughout this paper. A variety of matrix norms will be discussed;
in particular, we denote by $\left\Vert \boldsymbol{X}\right\Vert $,
$\left\Vert \boldsymbol{X}\right\Vert _{\text{F}}$, and $\left\Vert \boldsymbol{X}\right\Vert _{*}$
the spectral norm, the Frobenius norm, and the nuclear norm (i.e.
sum of all singular values) of $\boldsymbol{X}$, respectively. When
$\bX$ is a positive semidefinite (PSD) matrix, the nuclear norm coincides
with the trace $\|\bX\|_{*}=\trace(\bX)$. We use $\|\bX\|_{1}$
and $\|\bX\|_{0}$ to denote the $\ell_{1}$ norm and support size
of the vectorized $\bX$, respectively. The Euclidean inner product
between $\boldsymbol{X}$ and $\boldsymbol{Y}$ is defined as $\left\langle \boldsymbol{X},\boldsymbol{Y}\right\rangle =\trace(\boldsymbol{X}^{\top}\boldsymbol{Y})$.
We will abuse the notation and let $\boldsymbol{\Sigma}_{r}$ and
$\boldsymbol{\Sigma}_{k}$ stand for the best rank-$r$ approximation
and the best $k$-term approximation of $\boldsymbol{\Sigma}$ respectively,
i.e. 
\[
\boldsymbol{\Sigma}_{r}=\argmin_{\boldsymbol{M}:\text{rank}\left(\boldsymbol{M}\right)=r}\left\Vert \boldsymbol{\Sigma}-\boldsymbol{M}\right\Vert _{\mathrm{F}},
\]
and 
\[
\boldsymbol{\Sigma}_{k}=\argmin_{\boldsymbol{M}:\left\Vert \boldsymbol{M}\right\Vert _{0}=k}\left\Vert \boldsymbol{\Sigma}-\boldsymbol{M}\right\Vert _{\mathrm{F}},
\]
whenever clear from context. Besides, we denote by $\mathcal{T}$
the orthogonal projection operator onto Toeplitz matrices, and $\mathcal{T}^{\perp}$
its orthogonal complement. Some useful notations are summarized in
Table \ref{tab:Summary-of-Notation-Nonuniform}.

\begin{table*}
\caption{\label{tab:Summary-of-Notation-Nonuniform}Summary of Notation and
Parameters}

\vspace{5pt}
 \centering{}%
\begin{tabular}{l>{\raggedright}p{5in}}
\hline 
$\bSigma$, $\bSigma_{r}$, $\bSigma_{\mathrm{c}}$  & true covariance matrix, best rank-$r$ approximation of $\bSigma$,
and $\bSigma_{\mathrm{c}}:=\bSigma-\bSigma_{r}$\tabularnewline
$\bSigma$, $\bSigma_{\Omega_{0}}$, $\bSigma_{\Omega_{0}^{\text{c}}}$  & true covariance matrix, best $k$-sparse approximation of $\bSigma$,
and $\bSigma_{\Omega_{0}^{\text{c}}}:=\bSigma-\bSigma_{\Omega_{0}}$\tabularnewline
$\mathcal{T}$, $\mathcal{T}^{\perp}$  & orthogonal projection operator onto Toeplitz matrices, and its orthogonal
complement. \tabularnewline
$\bfeta,\boldsymbol{y}\in\mathbb{R}^{m}$ ,  & noise, quadratic measurements $\left\{ \boldsymbol{a}_{i}^{\top}\bSigma\boldsymbol{a}_{i}+\eta_{i}\right\} _{1\leq i\leq m}$\tabularnewline
$\boldsymbol{a}_{i}\in\mathbb{R}^{n}$, $\boldsymbol{A}_{i}\in\mathbb{R}^{n\times n}$  & $i$th sensing vector, $i$th sensing matrix $\boldsymbol{A}_{i}:=\boldsymbol{a}_{i}\cdot\boldsymbol{a}_{i}^{\top}$\tabularnewline
$\boldsymbol{B}_{i}\in\mathbb{R}^{n\times n}$  & auxiliary sensing matrix\tabularnewline
%$T$, $\mathcal{P}_{T}$, $\mathcal{P}_{T^{\perp}}$  & tangent space at the point $\bSigma$ if $\bSigma$ is low-rank or
%at the point $\bSigma_{r}$ otherwise, orthogonal projection onto
%$T$, and its complement operator\tabularnewline
$\mathcal{A}_{i}$,$\mathcal{A}$  & linear transformation $\boldsymbol{X}\mapsto\boldsymbol{a}_{i}^{\top}\boldsymbol{X}\boldsymbol{a}_{i}$
, linear mapping $\boldsymbol{X}\mapsto\left\{ \boldsymbol{a}_{i}^{\top}\boldsymbol{X}\boldsymbol{a}_{i}\right\} _{1\leq i\leq m}$\tabularnewline
$\mathcal{B}_{i}$,$\mathcal{B}$  & linear transformation $\boldsymbol{X}\mapsto\left\langle \boldsymbol{B}_{i},\boldsymbol{X}\right\rangle $
, linear mapping $\boldsymbol{X}\mapsto\left\{ \mathcal{B}_{i}\left(\boldsymbol{X}\right)\right\} _{1\leq i\leq m}$\tabularnewline
\hline 
\end{tabular}
\end{table*}

\section{Convex Relaxation and Its Performance Guarantees}

\label{sec:Formal-Setup}

In general, recovering the covariance matrix $\bSigma\in\mathbb{R}^{n\times n}$
from $m<n(n+1)/2$ measurements is ill-posed, unless the sampling
mechanism can effectively exploit the low-dimensional covariance structure.
Random sampling often preserves the information structure from minimal
observations, and allows robust recovery from noisy measurements.

In this paper, we restrict our attention to the following random sampling
model. We assume that the sensing vectors are composed of i.i.d. \emph{sub-Gaussian}
entries. In particular, we assume $\ba_{i}$'s ($1\leq i\leq m$)
are i.i.d. copies of $\boldsymbol{z}=\left[z_{1},\cdots,z_{n}\right]^{\top}$,
where each $z_{i}$ is i.i.d. drawn from a distribution with the following
properties 
\begin{equation}
\mathbb{E}[z_{i}]=0,\quad\mathbb{E}[z_{i}^{2}]=1,\quad\mbox{and}\quad\mu_{4}:=\mathbb{E}\left[z_{i}^{4}\right]>1.\label{sampling}
\end{equation}
%where the sub-Gaussian norm of a random variable $X$ is defined as
%$$ \|X \|_{\psi_2} =  \sup_{p\geq 1} p^{-1/2}(\mathbb{E} |X|^p)^{1/p}.$$
% \|a_i \|_{\psi_2} \leq K, 

\begin{comment}
In order to accommodate noisy measurements, we consider the following
general noise model. 
\begin{itemize}
\item \textbf{Noise Model.} We assume that the noise $\bfeta:=[\eta_{1},\cdots,\eta_{m}]^{\top}$
is bounded in $\ell_{1}$ norm \textit{deterministically}%
\footnote{Extension to stochastic noise models is straightforward, which we
omit in this paper.%
}, so that 
\begin{equation}
\|\bfeta\|_{1}\leq\epsilon,\label{noise}
\end{equation}
where $\epsilon$ is known {\em a priori}. \end{itemize}
\end{comment}

We assume that the error term $\bfeta:=[\eta_{1},\cdots,\eta_{m}]^{\top}$
is bounded in either $\ell_{1}$ norm or $\ell_{2}$ norm as specified
later in the theoretical guarantees. For notational simplicity, let
$\boldsymbol{A}_{i}:=\boldsymbol{a}_{i}\boldsymbol{a}_{i}^{\top}$
represent the equivalent sensing matrix, and hence the measurements
$\boldsymbol{y}:=\left[y_{1},\cdots,y_{m}\right]^{\top}$ obeys $y_{i}:=\left\langle \boldsymbol{A}_{i},\bSigma\right\rangle +\eta_{i}$.
We also define the linear operator $\cA(\boldsymbol{M}):\mathbb{R}^{n\times n}\mapsto\mathbb{R}^{m}$
that maps a matrix $\boldsymbol{M}\in\mathbb{R}^{n\times n}$ to $\{\langle\boldsymbol{M},\bA_{i}\rangle\}_{i=1}^{m}$.
These notations allow us to express the measurements as 
\begin{equation}
\by=\mathcal{A}(\bSigma)+\bfeta.\label{eq:ModelOperator}
\end{equation}

\subsection{Recovery of Low-Rank Covariance Matrices\label{sub:Recovery-of-Low-Rank}}

%Suppose that the true covariance matrix is $\bSigma\in\mathbb{R}^{n\times n}$,
%and we obtain $m$ independent quadratic measurements
%\[
%y_{i}:=\boldsymbol{a}^{(i)*}\bSigma\boldsymbol{a}^{(i)}+\eta_{i}\quad(1\leq i\leq m),
%\]
%where $\eta_{i}$ denotes the additive noise, $\boldsymbol{a}^{(i)}\in\mathbb{R}^{n}$
%represents the $i$th sensing vector. In general, perfect recovery
%is unlikely when the number $m$ of measurements is far below the
%ambient dimension, unless the information structure can be preserved
%in a stable manner. This motivates us to examine the class of sampling
%methods that can better exploit the signal geometry. Inspired by the
%success of compressed sensing, we explore random sensing strategies
%that often preserve signal structures. Specifically, the sampling
%model considered in this paper is summarized as follows. 

%The situation considered in Section \ref{eq:ExactPR} is unrealistic.
%A more common scenario in practical applications deviates from the
%model of Section \ref{eq:ExactPR} in the following two aspects:
%\begin{itemize}
%\item The covariance $\bSigma$ is approximately low rank;
%\item The measurements are contaminated by noise.
%\end{itemize}% 

Suppose that $\boldsymbol{\Sigma}$ is approximately {low-rank},
a natural heuristic is to perform rank minimization to encourage the
low-rank structure 
\begin{align}
\hat{\bSigma}=\argmin_{\boldsymbol{M}}\mbox{rank}(\boldsymbol{M})\quad\mbox{subject to} & \quad\boldsymbol{M}\succeq0,\label{rankmin}\\
 & \quad\|\by-\mathcal{A}(\boldsymbol{M})\|_{1}\leq\epsilon_{1},\nonumber 
\end{align}
where $\epsilon_{1}$ is an upper bound on $\left\Vert \boldsymbol{\eta}\right\Vert _{1}$
and assumed known {\em a priori}. However, the rank minimization
problem is in general NP-hard. Therefore, we replace it with trace
minimization over all matrices compatible with the measurements 
\begin{align}
\hat{\bSigma}=\argmin_{\boldsymbol{M}}\trace(\boldsymbol{M})\quad\mbox{subject to} & \quad\boldsymbol{M}\succeq0,\label{tracemin}\\
 & \quad\|\by-\mathcal{A}(\boldsymbol{M})\|_{1}\leq\epsilon_{1}.\nonumber 
\end{align}
Since $\bSigma$ is PSD, the trace norm forms a convex surrogate for
the rank function, which has proved successful in matrix completion
and phase retrieval problems \cite{ExactMC09,RecFazPar07,candes2012phaselift}.
It turns out that this convex relaxation approach \eqref{tracemin}
admits stable and faithful estimates even when $\bSigma$ is approximately
low rank and/or when the measurements are corrupted by bounded noise.
This is formally stated in the following theorem.

\begin{theorem}\label{thm:ApproxLR} Consider the sub-Gaussian sampling
model in \eqref{sampling} and assume that $\left\Vert \boldsymbol{\eta}\right\Vert _{1}\leq\epsilon_{1}$.
Then with probability exceeding $1-C_{0}\exp(-c_{0}m)$, the solution
$\hat{\bSigma}$ to \eqref{tracemin} satisfies 
\begin{equation}
\Vert\hat{\bSigma}-\bSigma\Vert_{\mathrm{F}}\leq C_{1}\frac{\left\Vert \bSigma-\bSigma_{r}\right\Vert _{*}}{\sqrt{r}}+C_{2}\frac{\epsilon_{1}}{m}\label{eq:ApproxLR}
\end{equation}
simultaneously for all $\bSigma\in\mathbb{R}^{n\times n}$, provided that $m>c_{1}nr$.
Here, $\bSigma_{r}$ represents the best rank-$r$ approximation of
$\bSigma$, and $c_{0}$, $c_{1}$, $C_{0}$, $C_{1}$ and $C_{2}$
are some positive numerical constants. \end{theorem}

The main implications of Theorem \ref{thm:ApproxLR} and its associated
performance bound (\ref{eq:ApproxLR}) are listed as follows. 
\begin{enumerate}
\itemsep0.3em
\item \textbf{Exact Recovery from Noiseless Measurement}s\textbf{.} Consider
the case where $\text{rank}\left(\bSigma\right)=r$. In the absence
of noise, one can see from \eqref{eq:ApproxLR} that the trace minimization
program \eqref{tracemin} (with $\epsilon_{1}=0$) allows perfect
covariance recovery with exponentially high probability, provided
that the number $m$ of measurements exceeds the order of $nr$. Notice
that each PSD matrix can be uniquely decomposed as $\boldsymbol{\Sigma}=\boldsymbol{L}\boldsymbol{L}^{\top}$, where $\boldsymbol{L}$ has orthogonal
columns. That said, the the intrinsic degrees of freedom carried
by PSD matrices is $\Theta(nr)$, indicating that our algorithm
achieves order-wise optimal recovery. 
\item \textbf{Near-Optimal Universal Recovery}. The trace minimization program
(\ref{tracemin}) allows universal recovery, in the sense that once
the sensing vectors are chosen, \emph{all} low-rank covariance matrices
can be perfectly recovered in the absence of noise. This highlights
the power of convex programming, which allows universally accurate
estimates as soon as the number of measurements exceeds the order
of the information theoretic limit. In addition, the universality
and optimality results hold for a large class of sub-Gaussian measurements
beyond the Gaussian sampling model. 
\item \textbf{Robust Recovery for Approximately Low-Rank Matrices}. In the
absence of noise ($\epsilon_1=0$), if $\bSigma$ is approximately low-rank,
then by (\ref{eq:ApproxLR}) the reconstruction inaccuracy is at most
\[
\Vert\hat{\bSigma}-\bSigma\Vert_{\text{F}}\leq O\left(\frac{\left\Vert \bSigma-\bSigma_{r}\right\Vert _{*}}{\sqrt{r}}\right)
\]
with probability at least $1-\exp(-c_{1}m)$, as soon as $m$ is about
the same order of $nr$. One can obtain a more intuitive understanding
through the following \emph{power-law} covariance model. Let $\lambda_{\ell}$
represent the $\ell$th largest singular value of $\bSigma$, and
suppose the decay of $\lambda_{\ell}$ obeys a power law, i.e. $\lambda_{\ell}\leq\frac{\alpha}{\ell^{\beta}}$
for some constant $\alpha>0$ and decay rate exponent $\beta>1$.
Then simple computation reveals that 
\[
\frac{\left\Vert \bSigma-\bSigma_{r}\right\Vert _{*}}{\sqrt{r}}\leq\frac{1}{\sqrt{r}}\sum_{\ell=r+1}^{n}\frac{\alpha}{\ell^{\beta}}\leq\frac{\alpha}{(\beta-1)r^{\beta-\frac{1}{2}}},
\]
%and
%\[ \left\Vert \bSigma- \bSigma_{r} \right\Vert _{\mathrm{F}}\leq\sqrt{\sum_{\ell=r+1}^{n}\frac{\alpha^2}{\ell^{2\beta}}}\leq\frac{\alpha}{\sqrt{2\beta-1}r^{\beta-\frac{1}{2}}}.
%\]
which in turn implies 
\begin{equation}
\Vert\hat{\bSigma}-\bSigma\Vert_{\text{F}}={O}\left(\frac{1}{r^{\beta-\frac{1}{2}}}\right).\label{eq:PowerLawEstimate}
\end{equation}
This asserts that \eqref{tracemin} returns an almost accurate estimate
of $\bSigma$ in a manner which requires no prior knowledge on the
signal (other than the power law decay that is natural for a broad
class of data). 
\item \textbf{Stable Recovery from Noisy Measurements}. When $\bSigma$
is exactly of rank $r$ and the noise is bounded $\|\bfeta\|_{1}\leq\epsilon_{1}$,
the reconstruction inaccuracy of \eqref{tracemin} is bounded above
by 
\begin{equation}
\Vert\hat{\bSigma}-\bSigma\Vert_{\text{F}}= O\left(\frac{\epsilon_{1}}{m}\right)\label{eq:stable}
\end{equation}
with exponentially high probability, provided that $m$ exceeds $\Theta\left(nr\right)$.
This reveals that the algorithm \eqref{tracemin} recovers an unknown
object with an error at most proportional to the average \emph{per-entry}
noise level, which makes it practically appealing. 
\item \textbf{Phase Retrieval with Sub-Gaussian Measurements}. The proposed
algorithm~\eqref{tracemin} appears in the same form as the convex
algorithm called \textit{PhaseLift}, which was proposed in \cite{candes2012phaselift}
for phase retrieval. It is equivalent to treating $\bSigma$ as the
rank-one lifted matrix $\boldsymbol{x}\boldsymbol{x}^{\top}$ from
an unknown signal $\boldsymbol{x}$. It has been established in \cite{candes2012solving}
that with high probability, it is feasible to recover $\boldsymbol{x}$
exactly from $\Theta\left(n\right)$ quadratic measurements, assuming
that the sensing vectors are i.i.d. Gaussian. Our result immediately
recovers all results of \cite{candes2012phaselift,candes2012solving}
including exact and stable recovery. In fact, our analysis framework
yields a much simpler and shorter proof of all these results, and
immediately extends to a broader class of sub-Gaussian sampling mechanisms.
We will further discuss our improvement of sparse recovery from magnitude
measurements \cite{li2012sparse,ohlsson2011compressive} in Section~\ref{sec:PhaseRetrieval}. 
\end{enumerate}
\begin{remark} A lower bound on the minimax risk has recently been
established by Cai and Zhang \cite[Theorem 2.4]{cai2013rop}. Specifically,
if the noise $\boldsymbol{\eta}\sim\mathcal{N}\left({\bf 0},\sigma^{2}\boldsymbol{I}\right)$
with $\sigma=\Theta\left(\frac{\epsilon_{1}}{m}\right)$, then for
any estimator $\tilde{\boldsymbol{\Sigma}}\left(\boldsymbol{y}\right)$,
\[
\inf_{\tilde{\boldsymbol{\Sigma}}\left(\cdot\right)}\sup_{\boldsymbol{\Sigma}:\text{ }\text{rank}\left(\boldsymbol{\Sigma}\right)=r}\sqrt{\mathbb{E}_{\boldsymbol{\eta}}\left[\left\Vert \tilde{\boldsymbol{\Sigma}}\left(\boldsymbol{y}\right)-\boldsymbol{\Sigma}\right\Vert _{\mathrm{F}}^{2}\right]}\gtrsim \sigma=\Theta\left(\frac{\epsilon_{1}}{m}\right),
\]
provided that $m=\Theta\left(nr\right)$. While our results are established
for bounded (possibly adversarial) noise, it is straightforward to
see that the above argument reveals the orderwise minimaxity of our
stability bound. \end{remark}

\subsection{Recovery of Low-Rank Covariance Matrices for Stationary Instances\label{sub:ToepResult}}

Suppose that $\boldsymbol{\Sigma}\in\mathbb{R}^{n\times n}$ is simultaneously
low-rank and Toeplitz, which can represent the covariance matrix of
a wide-sense stationary random process. Similar to recovery in the
general low-rank model, we propose to seek a nuclear norm minimizer
over all matrices compatible with the measurements as well as the
Toeplitz constraint, which results in the following estimator: 
\begin{align}
\hat{\boldsymbol{\Sigma}}=\argmin_{\boldsymbol{M}}\mathrm{Tr}(\boldsymbol{M})\quad\mbox{subject to} & \quad\boldsymbol{M}\succeq0,\label{eq:ToepMC}\\
 & \quad\|\boldsymbol{y}-\mathcal{A}(\boldsymbol{M})\|_{2}\leq\epsilon_{2},\nonumber \\
 & \quad\boldsymbol{M}\text{ is }\text{Toeplitz},\nonumber 
\end{align}
where $\epsilon_{2}$ is an upper bound of $\left\Vert \boldsymbol{\eta}\right\Vert _{2}$.

\begin{figure}
\centering
\includegraphics[width=0.3\textwidth]{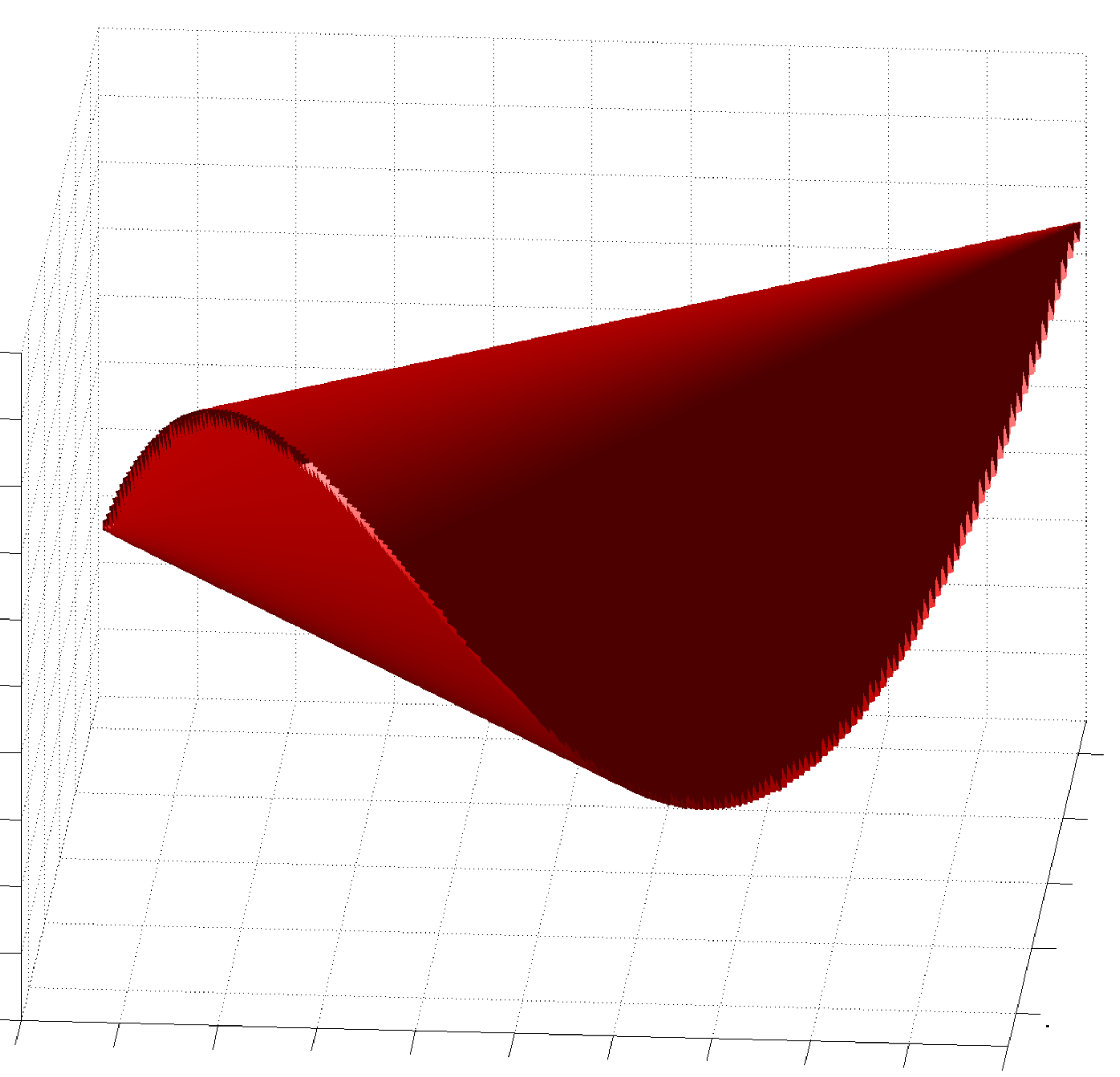}
\caption{Representation of the unit PSD Toeplitz ball consisting of all $\left(x,y,z\right)$
such that $\boldsymbol{T}=\left[\protect\begin{array}{cccc}
1/4 & x & y & z\protect\\
x & 1/4 & x & y\protect\\
y & x & 1/4 & x\protect\\
z & y & x & 1/4
\protect\end{array}\right]\succeq{\bf 0}$ and $\mathrm{Tr}\left(\boldsymbol{T}\right)=1$.}
\label{fig:toeplitz}
\end{figure}

Encouragingly,  the PSD Toeplitz cone can be very pointy around many low-rank feasible points,  as illustrated in 
Fig. \ref{fig:toeplitz}. Therefore, the intersection between the PSD Toeplitz cone and a random hyperplane passing through $\boldsymbol{\Sigma}$ often contains only a single point.  As a result, the semidefinite relaxation (\ref{eq:ToepMC}) is exact with high probability
under noise-free measurements, as stated in the following theorem.

\begin{theorem}\label{thm:ToeplitzPhaseLift}Consider the sub-Gaussian
sampling model in \eqref{sampling}, and assume that $\mu_{4}\leq3$
and $\left\Vert \boldsymbol{\eta}\right\Vert _{2}\leq\epsilon_{2}$.
Then with probability exceeding $1-1/n^{2}$, 
\begin{equation}
\Vert\hat{\bSigma}-\bSigma\Vert_{\mathrm{F}}\leq C_{2}\frac{\epsilon_{2}}{\sqrt{m}}\label{eq:ToepEstimate}
\end{equation}
holds simultaneously for all Toeplitz covariance matrices $\boldsymbol{\Sigma}$
of rank at most $r$, provided that $m>c_{0}r\log^{10}n$. Here, $c_{0}$
and $C_{2}$ are some universal constants. \end{theorem}

We highlight some implications of Theorem \ref{thm:ToeplitzPhaseLift}
as follows. 
\begin{enumerate}
\itemsep0.3em
\item \textbf{Exact Recovery without Noise}. As any rank-$r$ PSD Toeplitz
matrix admits a unique rank-$r$ Vandemonde decomposition that can
be specified by $2r$ parameters, by Theorem \ref{thm:ToeplitzPhaseLift},
exact recovery of Toeplitz low-rank covariance matrices occurs as
soon as $m$ is slightly larger than the information theoretic limit
$\Omega\left(r\right)$ (modulo some poly-logarithmic factor). Note
that this sampling requirement is much smaller than that for general
low-rank matrices, and also much smaller than the degrees of freedom
for general Toeplitz matrices (which is $n$). 
\item \textbf{Stable and Universal Recovery from Noisy Measurements}. The
proposed convex relaxation (\ref{eq:ToepMC}) returns faithful estimates
in the presence of noise, as revealed by Theorem \ref{thm:ToeplitzPhaseLift}.
This feature is universal: if $\mathcal{A}$ is randomly sampled and
then fixed thereafter, then, with high probability, the error bounds
(\ref{eq:ToepEstimate}) hold simultaneously for all Toeplitz low-rank
matrices. Note that the error bound (\ref{eq:ToepEstimate}) is stated
in terms of the $\ell_{2}$ norm of $\boldsymbol{\eta}$. This is
out of mathematical convenience for this special setup, which will
be discussed later. 
\end{enumerate}
\begin{remark}Two aspects of Theorem \ref{thm:ToeplitzPhaseLift}
are worth noting. First, Theorem \ref{thm:ToeplitzPhaseLift}
does not guarantee recovery with exponentially high probability as
ensured in Theorem \ref{thm:ApproxLR}. This arises from our use of
stochastic RIP, as will be seen in the analysis. Secondly, we are
only able to provide theoretical guarantees when $\mu_{4}\leq3$;
roughly speaking, the tails of these distributions are typically not
heavier than those of the Gaussian measure (e.g. $\mu_{4}=3$ for
Gaussian distribution and $\mu_{4}=1$ for Bernoulli distribution).
We conjecture that these two aspects can be improved via other proof
techniques.\end{remark}

%Once we obtain accurate recovery of $\boldsymbol{\Sigma}$, the underlying spectrum can be identified by conventional harmonic retrieval methods, e.g. ESPRIT \cite{RoyKailathESPIRIT1989}. 

\subsection{Recovery of Sparse Covariance Matrices\label{sub:Recovery-of-Sparse}}

Assume that $\boldsymbol{\Sigma}$ is approximately sparse, we propose
to seek a matrix with minimal support size that is compatible with
observations: 
\begin{align}
\hat{\bSigma}=\argmin_{\boldsymbol{M}}\|\boldsymbol{M}\|_{0}\quad\text{subject to} & \quad\boldsymbol{M}\succeq0,\label{cardmin}\\
 & \quad\|\by-\mathcal{A}(\boldsymbol{M})\|_{1}\leq\epsilon_{1},\nonumber 
\end{align}
where $\epsilon_{1}$ is an upper bound on $\left\Vert \boldsymbol{\eta}\right\Vert _{1}$.
However, the $\ell_{0}$ minimization problem in \eqref{cardmin}
is also intractable, and one can instead solve a tractable convex
relaxation of \eqref{cardmin}, given as 
\begin{align}
\hat{\bSigma}=\argmin_{\boldsymbol{M}}\|\boldsymbol{M}\|_{1}\quad\text{subject to} & \quad\boldsymbol{M}\succeq0,\label{l1min}\\
 & \quad\|\by-\mathcal{A}(\boldsymbol{M})\|_{1}\leq\epsilon_{1}.\nonumber 
\end{align}
Here, the $\ell_{1}$ norm is the convex relaxation of the support
size, which has proved successful in many compressed sensing algorithms
\cite{candes2008restricted,CandTao06}. It turns out that the convex
relaxation \eqref{l1min} allows stable and reliable estimates even
when $\bSigma$ is only approximately sparse and the measurements
are contaminated by noise, as stated in the following theorem.

\begin{theorem}\label{thm:ApproxSP} Consider the sub-Gaussian sampling
model in \eqref{sampling} and assume that $\left\Vert \boldsymbol{\eta}\right\Vert _{1}\leq\epsilon_{1}$.
Then with probability exceeding $1-C_{0}\exp(-c_{0}m)$, the solution
$\hat{\bSigma}$ to \eqref{l1min} satisfies 
\begin{equation}
\Vert\hat{\bSigma}-\bSigma\Vert_{\mathrm{F}}\leq C_{1}\frac{\left\Vert \bSigma-\bSigma_{\Omega}\right\Vert _{1}}{\sqrt{k}}+C_{2}\frac{\epsilon_{1}}{m},\label{eq:ApproxSP}
\end{equation}
simultaneously for all $\boldsymbol{\Sigma}\in\mathbb{R}^{n\times n}$, provided
that $m>c_{1}k\log(n^{2}/k)$. Here, $\bSigma_{\Omega}$ denotes the
best $k$-sparse approximation of $\bSigma$, and $c_{0}$, $c_{1}$,
$C_{0}$, $C_{1}$ and $C_{2}$ are positive universal constants. \end{theorem}

Theorem~\ref{thm:ApproxSP} leads to similar implications as those
listed in Section \ref{sub:Recovery-of-Low-Rank}, which we briefly
summarize as follows. 
\begin{enumerate}
\itemsep0.3em
\item \textbf{Exact Recovery without Noise:} When $\bSigma$ is exactly
$k$-sparse and no noise is present, by setting $\epsilon_{1}=0$,
the solution to \eqref{l1min} is exactly equal to the ground truth
with exponentially high probability, as soon as the number $m$ of
measurements is about the order of $k\log(n^{2}/k)$. Therefore our
performance guarantee in \eqref{eq:ApproxSP} is \textit{optimal}
within a constant factor. 
\item \textbf{Universal Recovery:} Our performance guarantee in \eqref{eq:ApproxSP}
is universal in the sense that the same sensing mechanism simultaneously
works for all sparse covariance matrices. 
\item \textbf{Imperfect Structural Models:} The estimate \eqref{eq:ApproxSP}
allows robust recovery for approximately sparse matrices (which appears
in a similar form as that for CS \cite{candes2008restricted}), indicating
that quadratic measurements are order-wise at least as good as linear
measurements. 
\end{enumerate}

\subsection{Recovery of Jointly Sparse and Rank-One Matrices\label{sec:PhaseRetrieval}}

If we set the covariance matrix $\bSigma=\boldsymbol{x}\boldsymbol{x}^{\top}$
to be a rank-one matrix, then covariance estimation from quadratic
measurements is equivalent to phase retrieval as studied in \cite{candes2012phaselift}.
%Thus, our main results in Theorem \ref{thm:ApproxLR} subsume as special cases the theoretical guarantee for phase retrieval. In fact, Theorem \ref{thm:ApproxLR} recovers the best-known performance guarantee on general phase retrieval problems \cite{candes2012phaselift,candes2012solving} (for which only Gaussian sensing vectors are considered), and improves upon them by generalizing to a larger class of sub-Gaussian sampling vectors.
In addition to the general rank-one model, our approach allows simple
analysis for recovering jointly sparse and rank-one covariance matrices
or, equivalently, sparse signal recovery from magnitude measurements.
Specifically, suppose that $\boldsymbol{x}$ is (approximately) sparse,
and we collect a small number of phaseless measurements as 
\[
\boldsymbol{y}:=\left\{ \left|\left\langle \boldsymbol{a}_{i},\boldsymbol{x}\right\rangle \right|^{2}+\eta_{i}\right\} _{1\leq i\leq m}.
\]

When $\boldsymbol{x}$ is sparse, the lifting matrix $\boldsymbol{x}\boldsymbol{x}^{\top}$
is \emph{simultaneously} low rank and sparse, which motivates us to
adapt the convex program proposed in \cite{li2012sparse} to accommodate
bounded noise as follows 
\begin{align}
\hat{\boldsymbol{X}}=\argmin_{\boldsymbol{M}}\quad & \text{Tr}\left(\boldsymbol{M}\right)+\lambda\left\Vert \boldsymbol{M}\right\Vert _{1}\label{eq:AlgorithmSparsePR}\\
\text{subject to}\quad & \boldsymbol{M}\succeq0,\nonumber \\
 & \left\Vert \boldsymbol{y}-\mathcal{A}\left(\boldsymbol{M}\right)\right\Vert _{1}\leq\epsilon_{1}.\nonumber 
\end{align}
Here, $\lambda$ is a regularization parameter that balances the two
convex surrogates (i.e. trace norm and $\ell_{1}$ norm) associated
with the low-rank and sparse structural assumptions, respectively,
and $\epsilon_{1}$ is an upper bound of $\left\Vert \boldsymbol{\eta}\right\Vert _{1}$.
Our analysis framework ensures stable recovery of an approximately
sparse signal, as stated in the following theorem.

\begin{theorem}\label{thm:SparsePR} Set $\lambda\in\left[\frac{1}{n},\frac{1}{\sqrt{k}}\rho\right]$
for some quantity $\rho$. Consider the sub-Gaussian sampling model
in \eqref{sampling} and assume that $\left\Vert \boldsymbol{\eta}\right\Vert _{1}\leq\epsilon_{1}$.
Then with probability at least $1-C_{0}\exp\left(-c_{0}m\right)$,
the solution $\hat{\boldsymbol{X}}$ to (\ref{eq:AlgorithmSparsePR})
satisfies 
\begin{align}
\left\Vert \hat{\boldsymbol{X}}-\boldsymbol{x}\boldsymbol{x}^{\top}\right\Vert _{\mathrm{F}} & \leq C_{1}\left\{ \left\Vert \boldsymbol{x}\boldsymbol{x}^{\top}-\boldsymbol{x}_{\Omega}\boldsymbol{x}_{\Omega}^{\top}\right\Vert _{*}\right.\nonumber \\
 & \quad\quad+\left.\lambda\left\Vert \boldsymbol{x}\boldsymbol{x}^{\top}-\boldsymbol{x}_{\Omega}\boldsymbol{x}_{\Omega}^{\top}\right\Vert _{1}+\frac{\epsilon_{1}}{m}\right\} \label{eq:StableSparsePRBound}
\end{align}
simultaneously for all signals $\boldsymbol{x}\in\mathbb{R}^{n}$ satisfying $\frac{\left\Vert \boldsymbol{x}_{\Omega}\right\Vert _{2}}{\left\Vert \boldsymbol{x}_{\Omega}\right\Vert _{1}}\geq\rho$,
provided that $m>\frac{C_{2}\log n}{\lambda^{2}}$. Here, $\boldsymbol{x}_{\Omega}$
denotes the best $k$-sparse approximation of $\boldsymbol{x}$, and
$C_{0}$, $C_{1},C_{2}$ and $c_{0}$ are positive universal constants.\end{theorem}

%Theorem \ref{thm:SparsePR} recovers all the theoretical performance
%guarantees established in \cite{li2012sparse} with a simpler proof,
%and improves upon them in two aspects: (i) Theorem \ref{thm:SparsePR}
%establishes the performance guarantees of the algorithm (\ref{eq:AlgorithmSparsePR})
%when the structural assumption is \emph{imperfect }and when the samples
%are \emph{noisy}; (ii) while \cite{li2012sparse} considers only Gaussian
%sensing vectors, we extend the results to a large class of sub-Gaussian sensing vectors. 

Theorem \ref{thm:SparsePR}, depending on the choice of $\lambda$,
provides universal recovery guarantees over a large class of signals
obeying $\frac{\left\Vert \boldsymbol{x}_{\Omega}\right\Vert _{2}}{\left\Vert \boldsymbol{x}_{\Omega}\right\Vert _{1}}\geq\rho$.
Some implications of Theorem \ref{thm:SparsePR} are as follows. 
\begin{enumerate}
\itemsep0.3em
\item \textbf{Exact Recovery for Exactly Sparse Signals}. When $\boldsymbol{x}$
is an exactly $k$-sparse signal, we can set $\rho=\frac{1}{\sqrt{k}}$
and $\lambda=\frac{1}{k}$ in Theorem~\ref{thm:SparsePR}, which
implies the algorithm (\ref{eq:AlgorithmSparsePR}) universally recovers
all $k$-sparse signals $\boldsymbol{x}$ from $O(k^{2}\log n)$ noise-free
measurements, with exponentially high probability. This recovers the
theoretical performance guarantees established in \cite{li2012sparse}
for Gaussian sensing vectors, but extends it to a large class of sub-Gaussian
sensing vectors, using a simpler proof. 
\item \textbf{Near-Optimal Recovery for Power-Law Exactly Sparse Signals}.
Somewhat surprisingly, if the \emph{nonzero entries} of $\boldsymbol{x}$
are known to be decaying in a power-law fashion, then the algorithm
(\ref{eq:AlgorithmSparsePR}) allows near-optimal recovery. Specifically,
suppose that the \emph{non-zero entries} of $\boldsymbol{x}$ satisfies
the power-law decay such that the magnitude of the $l$th largest
entry of $\boldsymbol{x}_{\Omega}/\left\Vert \boldsymbol{x}_{\Omega}\right\Vert _{2}$
is bounded above by ${c_{\text{pl}}}/{l^{\alpha}}$ for some constants
$c_{\text{pl}}$ and exponent $\alpha>1$, then $$\|\boldsymbol{x}_{\Omega}\|_{2}/\|\boldsymbol{x}_{\Omega}\|_{1}=O\left(1/\log k\right):=\rho.$$
By setting $\lambda=\Theta((\sqrt{k}\log n)^{-1})$, one can obtain
accurate recovery from $O\left(k\log^{2}n\right)$ noiseless samples,
which is only a logarithmic factor from the minimum sample complexity
requirement. 
\item \textbf{Stable and Universal Recovery for Imperfect Models and Noisy
Samples}. When the sparsity assumption is inexact, or measurements
are noisy, the estimate $\hat{\boldsymbol{X}}$ will not be exact,
and we can recover the estimate of the signal $\hat{\bx}$ as the
top (normalized) eigenvector of $\hat{\bX}$. Using the Davis-Kahan
theorem in standard matrix perturbation theory \cite{stewart1990matrix},
we have 
\[
\sin\angle(\hat{\bx},\bx)\leq\frac{1}{\|\bx\|_{2}^{2}}\left\Vert \hat{\boldsymbol{X}}-\boldsymbol{x}\boldsymbol{x}^{\top}\right\Vert _{\mathrm{F}}
\]
bounded by Theorem~\ref{thm:SparsePR}, where $\angle(\hat{\bx},\bx)$
represents the angle between $\hat{\bx}$ and $\bx$. The recovered
signal $\hat{\bx}$ is a highly accurate estimate if $\bx_{\Omega^{c}}$
is small enough. The estimation inaccuracy due to noise corruption
is also small, in the sense that it is at most proportional to the
per-entry noise level. This generalizes prior work \cite{li2012sparse}
to imperfect structural assumptions as well as noisy measurements. 
\end{enumerate}

\subsection{Extension to General Matrices\label{sub:Extension-to-General-Model}}

Table \ref{main_results} summarizes the main results of Theorems~\ref{thm:ApproxLR}
-- \ref{thm:SparsePR}. We further remark that the main results hold
even when $\bSigma$ is not PSD but a symmetric matrix%
\footnote{The proposed framework and proof arguments can also be easily extended to handle asymmetric matrices
without difficulty, using bilinear rank-one measurements.%
}. When $\bSigma$ is not a covariance matrix but a general low-rank,
Toeplitz low-rank, or sparse matrix, one can simply drop the PSD constraint
in the proposed algorithms, and replace the trace norm objective by
the nuclear norm in \eqref{tracemin}. As will be shown, the PSD constraint
is never invoked in the proof, hence it is straightforward to extend
all results to the more general cases where $\bSigma$ is a general
$n\times n$ low-rank, Toeplitz low-rank, or sparse matrix. Note that
in this more general scenario, the measurements in \eqref{measurements}
are no longer nonnegative.

\begin{table*}
\begin{centering}
\caption{Summary of Main Results.}

\par\end{centering}

\begin{centering}
\begin{tabular}{c|c|c|c}
\hline 
Structure  & Number of Measurements  & Noise  & RIP \tabularnewline
\hline 
rank-$r$  & $O(nr)$  & $\ell_{1}$  & $\ell_{2}/\ell_{1}$ \tabularnewline
\hline 
Toeplitz rank-$r$  & $O(r\mbox{polylog}n)$  & $\ell_{2}$  & $\ell_{2}/\ell_{2}$ \tabularnewline
\hline 
$k$-sparse  & $O(k\log(n^{2}/k))$  & $\ell_{1}$  & $\ell_{2}/\ell_{1}$\tabularnewline
\hline 
$k$-sparse and rank-one  & $O(k^{2}\log n)$ $\text{(general sparse);}$  & $\ell_{1}$  & $\ell_{2}/\ell_{1}$ \tabularnewline
 & $O(k\log^{2}n)$ $\text{(power-law sparse)}$  &  & \tabularnewline
\hline 
\end{tabular}
\par\end{centering}

\centering{}\label{main_results} 
\end{table*}

%Suppose that
%the measurements $ $$y_{i}$ is contaminated by noise satisfying
%\[
%\left\Vert \boldsymbol{y}-\mathcal{A}\left(\bSigma\right)\right\Vert _{1}\leq\epsilon,
%\]
%and the upper bound $\epsilon$ is known \emph{a priori}. To accommodate
%inaccurate measurements, we modify the algorithm as follows
%\begin{align}
%\underset{\boldsymbol{M}\in\mathcal{S}^{n\times n}}{\text{minimize}}\quad & \text{tr}\left(\boldsymbol{M}\right)\nonumber \\
%\text{subject to}\quad & \left\Vert \boldsymbol{y}-\mathcal{A}\left(\bSigma\right)\right\Vert _{1}\leq\epsilon;\label{eq:PR-Stable}\\
% & \boldsymbol{M}\succeq{\bf 0}.
%\end{align}

\section{Approximate $\ell_{2}/\ell_{1}$ Isometry for Low-rank and Sparse
Matrices\label{sec:RIP}}

In this section, we present a novel concept called the mixed-norm
restricted isometry property (RIP-$\ell_{2}/\ell_{1}$) that allows
us to establish Theorems~\ref{thm:ApproxLR}, \ref{thm:ApproxSP},
and \ref{thm:SparsePR} concerning universal recovery of low-rank,
sparse and sparse rank-one covariance matrices from quadratic measurements.

Prevailing wisdom in CS asserts that perfect recovery from minimal
samples is possible if the dimensionality reduction projection preserves
the signal strength when acting on the class of matrices of interest
\cite{CandTao06,RecFazPar07}. While there are various ways to define
the restricted isometry properties (RIP), an appropriately chosen
approximate isometry leads to a very simple yet powerful theoretical
framework.

%\begin{remark}We emphasize that the quadratic sampling operator $\mathcal{A}$
%does not satisfy the RIP-$\ell_{2}/\ell_{2}$ introduced
%in \cite{RecFazPar07,CandesPlan2011Tight}, which is crucial in guaranteeing
%exact matrix recovery in many prior works. This fact has been formally
%pointed out by Cand\`es et. al. in \cite{candes2012phaselift}, which
%motivates the exploration of a new analysis framework.\end{remark}

\subsection{Mixed-Norm Restricted Isometry (RIP-$\ell_{2}/\ell_{1}$)}

Recall that the RIP occurs if the sampling output preserves the input
strength under certain metrics. The most commonly used one is RIP-$\ell_{2}/\ell_{2}$,
for which the signal strength before and after the projection are
both measured in terms of the Frobenius norm \cite{candes2008restricted,RecFazPar07}.
This, however, fails to hold under rank-one measurements -- see detailed
arguments by Candes et. al. in \cite{candes2012phaselift}. Another
isometry concept called RIP-$\ell_{1}/\ell_{1}$ has also been investigated,
for which the signal strength before and after the operation $\mathcal{A}$
are measured both in terms of the $\ell_{1}$ norms%
\footnote{Note that the nuclear norm is the $\ell_{1}$-norm counterpart for
matrices.%
}. This is initially developed to account for measurements from expander
graphs \cite{jafarpour2009efficient}, and has become a powerful metric
when analyzing phase retrieval \cite{candes2012phaselift,candes2012solving,li2012sparse}.
Nevertheless, when considering general low-rank matrices, RIP-$\ell_{1}/\ell_{1}$
no longer holds. %
%\footnote
To see this, consider two matrices 
\begin{align*}
\boldsymbol{X}_{1} &= \mbox{diag}\{\boldsymbol{I}_{r/2},\boldsymbol{I}_{r/2}{\bf 0}\} \\
\boldsymbol{X}_{2} &= \mbox{diag}\{\boldsymbol{I}_{r/2},-\boldsymbol{I}_{r/2}{\bf 0}\}
\end{align*}
enjoying the same nuclear norm. When $m=\Omega\left(nr\right)$, one can
see from the Bernstein inequality (for sub-exponential variables)
that 
\begin{align*}
\frac{1}{m}\left\Vert \mathcal{A}\left(\boldsymbol{X}_{1}\right)\right\Vert _{1} =\Theta\left(r\right), \quad
\frac{1}{m}\left\Vert \mathcal{A}\left(\boldsymbol{X}_{2}\right)\right\Vert _{1} =\Theta\left(\sqrt{r}\right),
\end{align*} 
precluding the existence of a small RIP-$\ell_{1}/\ell_{1}$ constant. %
 Leaving out this matter, the proof based on RIP-$\ell_{1}/\ell_{1}$ typically
relies on delicate construction of dual certificates \cite{candes2012phaselift,candes2012solving,li2012sparse},
which is often mathematically complicated.

One of the key and novel ingredients in our analysis is a mixed-norm
approximate isometry, which measures the signal strength before and
after the projection with different metrics. Specifically, we introduce
RIP-$\ell_{2}/\ell_{1}$, where the input and output are measured
in terms of the Frobenius norm and the $\ell_{1}$ norm, respectively.
It turns out that as long as the input is measured with the Frobenius
norm, the standard trick pioneered in \cite{candes2008restricted}
in treating linear measurements carry over to quadratic measurements
with slight modifications and saves the need for dual construction.
We make formal definitions of RIP-$\ell_{2}/\ell_{1}$ for low-rank/sparse
matrices as follows.

\begin{definition}[\textbf{RIP-$\ell_{2}/\ell_{1}$ for low-rank
matrices}]\label{defn:LR-RIP}For the set of rank-$r$ matrices,
we define the RIP-$\ell_{2}/\ell_{1}$ constants $\delta_{r}^{\mathrm{lb}}$
and $\delta_{r}^{\mathrm{ub}}$ with respect to an operator $\mathcal{B}$
as the smallest numbers such that for all $\boldsymbol{X}$ of rank
at most $r$: 
\[
\left(1-\delta_{r}^{\mathrm{lb}}\right)\left\Vert \boldsymbol{X}\right\Vert _{\mathrm{F}}\leq\frac{1}{m}\left\Vert \mathcal{B}\left(\boldsymbol{X}\right)\right\Vert _{1}\leq\left(1+\delta_{r}^{\mathrm{ub}}\right)\left\Vert \boldsymbol{X}\right\Vert _{\mathrm{F}}.
\]
\end{definition}

\begin{definition}[\textbf{RIP-$\ell_{2}/\ell_{1}$ for sparse matrices}]\label{defn:Sparse-RIP}For
the set of $k$-sparse matrices, we define the RIP-$\ell_{2}/\ell_{1}$
constants $\gamma_{k}^{\mathrm{lb}}$ and $\gamma_{k}^{\mathrm{ub}}$
with respect to an operator $\mathcal{B}$ as the smallest numbers
such that for all $\boldsymbol{X}$ of sparsity at most $k$: 
\[
\left(1-\gamma_{k}^{\mathrm{lb}}\right)\left\Vert \boldsymbol{X}\right\Vert _{\mathrm{F}}\leq\frac{1}{m}\left\Vert \mathcal{B}\left(\boldsymbol{X}\right)\right\Vert _{1}\leq\left(1+\gamma_{k}^{\mathrm{ub}}\right)\left\Vert \boldsymbol{X}\right\Vert _{\mathrm{F}}.
\]
\end{definition}

\begin{definition}[\textbf{RIP-$\ell_{2}/\ell_{1}$ for low-rank
plus sparse matrices}]\label{defn:SparsePR-RIP}Consider the class
of index sets 
\begin{align*}
\mathcal{S}_{k}: & =\left\{ \Omega\in[n]\times[n]\text{ }\Big|\text{ }\exists\text{ an index set }\omega\in[n]\right.\\
 & \quad\quad\quad\quad\quad\left.\text{ of cardinality }k\text{ such that }\Omega=\omega\times\omega\right\} .
\end{align*}
For the set of matrices 
\begin{align}
\mathcal{M}_{k,r,l} & =\left\{ \boldsymbol{X}_{1}+\boldsymbol{X}_{2}\text{ }\Big|\text{ }\exists\Omega\in\mathcal{S}_{k},\text{ }\mathrm{rank}\left(\boldsymbol{X}_{1}\right)\leq r,\right.\label{eq:DefnMrk}\\
 & \quad\quad\quad\quad\quad\quad\quad\left.\mathrm{supp}(\boldsymbol{X}_{1})\in\Omega,\text{ }\left\Vert \boldsymbol{X}_{2}\right\Vert _{0}\leq l\right\} .\nonumber 
\end{align}
we define the RIP-$\ell_{2}/\ell_{1}$ constants $\delta_{k,r,l}^{\mathrm{lb}}$
and $\delta_{k,r,l}^{\mathrm{ub}}$ with respect to an operator $\mathcal{B}$
as the smallest numbers such that $\forall\boldsymbol{X}\in\mathcal{M}_{k,r,l}$:
\[
\left(1-\delta_{k,r,l}^{\mathrm{lb}}\right)\left\Vert \boldsymbol{X}\right\Vert _{\mathrm{F}}\leq\frac{1}{m}\left\Vert \mathcal{B}\left(\boldsymbol{X}\right)\right\Vert _{1}\leq\left(1+\delta_{k,r,l}^{\mathrm{ub}}\right)\left\Vert \boldsymbol{X}\right\Vert _{\mathrm{F}}.
\]
\end{definition}

\begin{remark}In short, any matrix within $\mathcal{M}_{k,r,l}$
can be decomposed into two components $\bX_{1}$ and $\bX_{2}$, where
$\bX_{1}$ is simultaneously low-rank and sparse, and $\bX_{2}$ is
sparse. This allows us to treat each matrix perturbation as a superposition
of a collection of jointly low-rank and sparse matrices and a collection
of general sparse matrices, where the rank-one measurements of each
term can be well controlled under minimal sample complexity.\end{remark}

\subsection{RIP-$\ell_{2}/\ell_{1}$ of Quadratic Measurements for Low-rank and
Sparse Matrices\label{sub:RIP-for-Quadratic-Measurements}}

Unfortunately, the original sampling operator $\mathcal{A}$ does
not satisfy RIP-$\ell_{2}/\ell_{1}$. This occurs primarily because
each measurement matrix $\boldsymbol{A}_{i}$ has non-zero mean, which
biases the output measurements. In order to get rid of this undesired
bias effect, we introduce a set of ``debiased'' auxiliary measurement
matrices as follows 
\begin{equation}
\boldsymbol{B}_{i}:=\boldsymbol{A}_{2i-1}-\boldsymbol{A}_{2i}.\label{eq:ZeroMeanOperator}
\end{equation}
Without loss of generality, denote $\mathcal{B}_{i}\left(\boldsymbol{X}\right):=\left\langle \boldsymbol{B}_{i},\boldsymbol{X}\right\rangle $
for all $1\leq i\leq m$, and let $\mathcal{B}\left(\boldsymbol{X}\right)$
represent the linear transformation that maps $\bX$ to $\{\mathcal{B}_{i}\left(\boldsymbol{X}\right)\}_{i=1}^{m}$.
Note that by representing the sensing process using $m$ rank-2 measurements
$\mathcal{B}_{i}$, we have implicitly doubled the number of measurements
for notational simplicity. This, however, will not change our order-wise
results.

It turns out that the auxiliary operator $\mathcal{B}$ exhibits the
RIP-$\ell_{2}/\ell_{1}$ in the presence of minimal measurements,
which can be shown by combining the following proposition with a standard
covering argument as applied in \cite{CandesPlan2011Tight}.

\begin{proposition}\label{lemma:RIP_lr} Let $\mathcal{A}$ be sampled
from the sub-Gaussian model in \eqref{sampling}. For any matrix $\bX$,
there exist universal constants $c_{1},c_{2},c_{3}>0$ such that with
probability exceeding $1-\exp\left(-c_{3}m\right)$, one has 
\begin{equation}
c_{1}\left\Vert \boldsymbol{X}\right\Vert _{\mathrm{F}}\leq\frac{1}{m}\left\Vert \mathcal{B}\left(\boldsymbol{X}\right)\right\Vert _{1}\leq c_{2}\left\Vert \boldsymbol{X}\right\Vert _{\mathrm{F}}.\label{eq:ApproximateRIP}
\end{equation}
\end{proposition} \begin{proof}See Appendix \ref{sec:Proof-of-Lemma-RIP}.\end{proof}

\begin{remark}This statement extends without difficulty to the bilinear
rank-one measurement model where $\boldsymbol{y}_{i}=\boldsymbol{a}_{i}^{\top}\bSigma\boldsymbol{b}_{i}$
for some independently generated sensing vectors $\boldsymbol{a}_{i}$
and $\boldsymbol{b}_{i}$. This indicates that all our results hold
for this asymmetric sensing model as well. \end{remark}

An immediate consequence of Proposition~\ref{corollary:RIP_lr} is
the establishment of RIP-$\ell_{2}/\ell_{1}$ of the sampling operator
$\mathcal{B}$ for either general low-rank or sparse matrices. The
proof of the corollaries below follows immediately from a standard
covering argument detailed in \cite[Section III.B]{CandesPlan2011Tight}
and \cite[Section 5]{baraniuk2008simple}. We thus omit the details
but refer interested readers to the above references for details.

\begin{corollary}[\textbf{RIP-$\ell_{2}/\ell_{1}$ for low-rank matrices}]\label{corollary:RIP_lr}Consider
the sub-Gaussian sampling model in \eqref{sampling} and the universal
constants $c_{1},c_{2}>0$ given in (\ref{eq:ApproximateRIP}). There
exist universal constants $c_{3},c_{4},C_{3}>0$ such that with probability
exceeding $1-C_{3}\exp\left(-c_{3}m\right)$, $\mathcal{B}$ satisfies
RIP-$\ell_{2}/\ell_{1}$ for all matrices $\boldsymbol{X}$ of rank
at most $r$, and obeys 
\begin{equation}
1-\delta_{r}^{\mathrm{lb}}\geq\frac{c_{1}}{2},\quad1+\delta_{r}^{\mathrm{ub}}\leq2c_{2},\label{eq:RIP-Bound-LR}
\end{equation}
provided that $m>c_{4}nr$. \end{corollary}

\begin{comment}
\begin{proof}See Appendix \ref{sec:Proof-of-Lemma-RIP_lowrank}.\end{proof} 
\end{comment}

\begin{corollary}[\textbf{RIP-$\ell_{2}/\ell_{1}$ for sparse matrices}]\label{corollary:RIP_sparse}Consider
the sub-Gaussian sampling model in \eqref{sampling} and the universal
constants $c_{1},c_{2}>0$ given in (\ref{eq:ApproximateRIP}). Then
with probability exceeding $1-C_{3}\exp\left(-c_{3}m\right)$, $\mathcal{B}$
satisfies the RIP-$\ell_{2}/\ell_{1}$ for all matrices $\boldsymbol{X}$
of sparsity at most $k$, and obeys 
\begin{equation}
1-\gamma_{k}^{\mathrm{lb}}\geq\frac{c_{1}}{2},\quad1+\gamma_{k}^{\mathrm{ub}}\leq2c_{2},\label{eq:RIP-Bound-Sparse}
\end{equation}
provided that $m>c_{4}k\log(n^{2}/k)$, where $c_{3},c_{4},C_{3}>0$
are some universal constants.\end{corollary}

\begin{comment}
\begin{proof}See Appendix \ref{sec:Proof-of-Lemma-RIP_sparse}.\end{proof} 
\end{comment}

\begin{corollary}[\textbf{RIP-$\ell_{2}/\ell_{1}$ for low-rank plus
sparse matrices}]\label{corollary:RIP_sparse-lowrank}Consider the
sub-Gaussian sampling model in \eqref{sampling} and the universal
constants $c_{1},c_{2}>0$ given in (\ref{eq:ApproximateRIP}). Then
with probability exceeding $1-C_{3}\exp\left(-c_{3}m\right)$, $\mathcal{B}$
satisfies the RIP-$\ell_{2}/\ell_{1}$ with respect to $\mathcal{M}_{k,r,l}$
(defined in (\ref{eq:DefnMrk})), and obeys 
\begin{equation}
1-\delta_{k,r,l}^{\mathrm{lb}}\geq\frac{c_{1}}{2},\quad1+\delta_{k,r,l}^{\mathrm{ub}}\leq2c_{2},
\end{equation}
provided that $m>c_{4}\max\left\{ kr\log(n/k),l\log(n^{2}/l)\right\} $,
where $c_{3},c_{4},C_{3}>0$ are some universal constants.\end{corollary}

\begin{remark}Recall that each matrix in $\mathcal{M}_{k,r,l}$ is
a sum of some $\boldsymbol{X}_{1}$ and $\boldsymbol{X}_{2}$, where
$\boldsymbol{X}_{1}$ is a rank-$r$ matrix in a $k\times k$ subspace,
while $\boldsymbol{X}_{2}$ is an $l$-sparse matrix. Consequently,
if we let $\mathcal{C}_{\epsilon}\left(\mathcal{M}\right)$ stand for
the covering number of a set $\mathcal{M}$ (i.e. the the fewest number
of points in any $\epsilon$-net of $\mathcal{M}$), then $\mathcal{C}_{\epsilon}\left(\mathcal{M}_{k,r,l}\right)$
is apparently bounded above by the product of $\mathcal{C}_{\epsilon/2}\left(\mathcal{M}_{r}\right)$
and $\mathcal{C}_{\epsilon/2}\left(\mathcal{M}_{l}\right)$, where
$\mathcal{M}_{r}$ and $\mathcal{M}_{l}$ denotes the rank-$r$ manifold
(with ambient dimension $k$) and the $\ell$-sparse manifold (with
ambient dimension $n^{2}$), respectively. Thus, $\log\mathcal{C}_{\epsilon}\left(\mathcal{M}_{k,r,l}\right)$
cannot exceed $kr\log(n/k)+l\log\left(n^{2}/l\right)$.\end{remark}

%Note that the covering numbers for either low-rank manifold or the class of sparse matrices are well explored, which facilitates the proof for low-rank, sparse, or jointly rank-one and sparse matrices. In fact, when the operator $\mathcal{B}$ satisfies RIP-$\ell_{2}/\ell_{1}$ for these classes of matrices, Theorems~\ref{thm:ApproxLR}, \ref{thm:ApproxSP}, and \ref{thm:SparsePR} can be proved in a fairly short and simple approach without delicate construction of dual certificates. The details of the proof are deferred to Appendix~\ref{sec:proof}. 

\subsection{Proof of Theorems \ref{thm:ApproxLR}, \ref{thm:ApproxSP} and \ref{thm:SparsePR}
via RIP-$\ell_{2}/\ell_{1}$}

Theorems \ref{thm:ApproxLR} and \ref{thm:ApproxSP} can thus be proved
given that the auxiliary operator $\mathcal{B}$ satisfies RIP-$\ell_{2}/\ell_{1}$
with sufficiently small constants, as asserted in Corollaries~\ref{corollary:RIP_lr}
and \ref{corollary:RIP_sparse}. We first present Lemma~\ref{lemma:duality}
which in turn establishes Theorem~\ref{thm:ApproxLR}.

\begin{lemma}\label{lemma:duality}Consider any matrix $\bSigma=\bSigma_{r}+\bSigma_{\mathrm{c}}$,
where $\bSigma_{r}$ is the best rank-$r$ approximation of $\bSigma$.
If there exists a number $K_{1}>2r$ such that 
\begin{equation}
\frac{1-\delta_{2r+K_{1}}^{\mathrm{lb}}}{\sqrt{2}}-\left(1+\delta_{K_{1}}^{\mathrm{ub}}\right)\sqrt{\frac{2r}{K_{1}}}\geq\beta_{1}>0\label{condition_lr}
\end{equation}
holds for some numerical value $\beta_{1}$, then the minimizer $\hat{\bSigma}$
to \eqref{tracemin} obeys 
\begin{equation}
\Vert\hat{\bSigma}-\bSigma\Vert_{\mathrm{F}}\leq\left(\frac{C_{1}}{\beta_{1}}+C_{3}\right)\frac{\left\Vert \bSigma_{\mathrm{c}}\right\Vert _{*}}{\sqrt{K_{1}}}+\frac{C_{2}}{\beta_{1}}\cdot\frac{\epsilon_{1}}{m}\label{eq:ApproxDualityBound}
\end{equation}
for some positive universal constants $C_{1}$, $C_{2}$ and $C_{3}$
depending only on the RIP-$\ell_{2}/\ell_{1}$ constants. \end{lemma}

\begin{proof} See Appendix~\ref{proof:duality}. \end{proof}

By choosing 
\[
K_{1}=8\left(\frac{4c_{2}}{c_{1}}\right)^{2}r\geq8\left(\frac{1+\delta_{K_{1}}^{\text{ub}}}{1-\delta_{2r+K_{1}}^{\text{ub}}}\right)^{2}r
\]
for the universal constants $c_{1},c_{2}$ given in Corollary \ref{corollary:RIP_lr},
we obtain \eqref{condition_lr} when $m>c_{4}\left(K_{1}+2r\right)n$
for some constant $c_{4}$. This establishes Theorem \ref{thm:ApproxLR}.

Theorem~\ref{thm:ApproxSP} is a direct consequence from the following lemma.

%\subsection{Proof of Theorem~\ref{thm:ApproxSP}}

\begin{lemma}\label{lemma:duality-sparse}Consider any matrix $\bSigma=\bSigma_{\Omega}+\bSigma_{\Omega^{c}}$,
where $\bSigma_{\Omega}$ is the best $k$-term approximation of $\bSigma$.
If there exists a number $K_{2}>2k$ such that 
\begin{equation}
\frac{(1-\gamma_{k+K_{2}}^{\mathrm{lb}})}{\sqrt{2}}-\left(1+\gamma_{K_{2}}^{\mathrm{ub}}\right)\sqrt{\frac{k}{K_{2}}}\geq\beta_{2}>0\label{condition-sparse}
\end{equation}
holds for some numerical value $\beta_{2}$, then the minimizer $\hat{\bSigma}$
to \eqref{tracemin} obeys 
\begin{equation}
\Vert\hat{\bSigma}-\bSigma\Vert_{\mathrm{F}}\leq\left(\frac{C_{1}}{\beta_{2}}+C_{3}\right)\frac{\left\Vert \boldsymbol{\Sigma}_{\Omega^{c}}\right\Vert _{1}}{\sqrt{K_{2}}}+\frac{C_{2}}{\beta_{2}}\frac{\epsilon_{1}}{m}\label{eq:ApproxDualityBound}
\end{equation}
for some positive universal constants $C_{1}$, $C_{2}$, $C_{3}$
depending only on the RIP-$\ell_{2}/\ell_{1}$ constants. \end{lemma}

\begin{proof} See Appendix~\ref{proof:duality_sparse}. \end{proof}

By picking 
\[
K_{2}=4\left(\frac{4c_{2}}{c_{1}}\right)^{2}k\geq4\left(\frac{1+\gamma_{K_{2}}^{\text{ub}}}{1-\gamma_{k+K_{2}}^{\text{lb}}}\right)^{2}k,
\]
one obtains \eqref{condition-sparse} as soon as $m>c_{4}\left(K_{2}+2k\right)\log(n^{2}/k)$
for the constant $c_{4}$ given in Corollary~\ref{corollary:RIP_sparse}.
This concludes the proof of Theorem \ref{thm:ApproxSP}.

Furthermore, the specialized RIP-$\ell_{2}/\ell_{1}$ concept allows
us to prove Theorem~\ref{thm:SparsePR} through the following lemma.

\begin{lemma}\label{lemma:duality-SparsePR-Stable}Set $\lambda$
to be any number within the interval $\left[\frac{1}{n},\frac{1}{\sqrt{k}}\frac{\left\Vert \boldsymbol{x}_{\Omega}\right\Vert _{2}}{\left\Vert \boldsymbol{x}_{\Omega}\right\Vert _{1}}\right]$.
Suppose that $\boldsymbol{x}_{\Omega}$ is the best $k$-term approximation
of $\boldsymbol{x}$. If there exists a number $K_{1}$ such that
\begin{equation}
\begin{cases}
\frac{\frac{1}{\sqrt{3}}\left(1-\delta_{k,2K_{1},\frac{2K_{1}}{\lambda^{2}}}^{\mathrm{lb}}\right)-\frac{3}{\sqrt{K_{1}}}\left(1+\delta_{k,K_{1},\frac{K_{1}}{\lambda^{2}}}^{\mathrm{ub}}\right)}{2\max\left\{ \frac{1}{\sqrt{K_{1}}}\left(1+\delta_{k,K_{1},\frac{K_{1}}{\lambda^{2}}}^{\mathrm{ub}}\right),1\right\} }\geq\beta_{3}>0,\\
\frac{1+\delta_{k,K_{1},\frac{K_{1}}{\lambda^{2}}}^{\mathrm{ub}}}{\left(1-\delta_{k,K_{1},\frac{K_{1}}{\lambda^{2}}}^{\mathrm{lb}}\right)\sqrt{K_{1}}}\leq\beta_{4}
\end{cases}\label{eq:Beta3}
\end{equation}
for some absolute constants $\beta_{3}$ and $\beta_{4}$, then the
solution $\hat{\boldsymbol{X}}$ to (\ref{eq:AlgorithmSparsePR})
satisfies 
\begin{align}
\left\Vert \hat{\boldsymbol{X}}-\boldsymbol{x}\boldsymbol{x}^{\top}\right\Vert _{\mathrm{F}} & \leq C\left\{ \left\Vert \boldsymbol{x}\boldsymbol{x}^{\top}-\boldsymbol{x}_{\Omega}\boldsymbol{x}_{\Omega}^{\top}\right\Vert _{*}\right.\nonumber \\
 & \quad\quad\quad\left.+\lambda\left\Vert \boldsymbol{x}\boldsymbol{x}^{\top}-\boldsymbol{x}_{\Omega}\boldsymbol{x}_{\Omega}^{\top}\right\Vert _{1}+\frac{\epsilon_{1}}{m}\right\} \label{eq:ErrorBound}
\end{align}
for some constant $C$ that depends only on $\beta_{3}$ and $\beta_{4}$.\end{lemma}

\begin{proof}See Appendix \ref{sec:Proof-of-Lemma-Duality-SparsePR-Stable}.
\end{proof}

From Corollary \ref{corollary:RIP_sparse-lowrank}, one can ensure small
RIP-$\ell_{2}/\ell_{1}$ constants satisfying \eqref{eq:Beta3}, provided that
\[
m>c_{4}\max\left\{ kK_{1}\log n,\frac{K_{1}}{\lambda^{2}}\log n\right\} =c_{4}\frac{K_{1}}{\lambda^{2}}\log n.
\]
This in turn establishes Theorem \ref{thm:SparsePR}.

Finally, note that we have not discussed general Toeplitz low-rank
matrices using RIP-$\ell_{2}/\ell_{2}$. We are unaware of a rigorous
approach to prove exact recovery using RIP-$\ell_{2}/\ell_{1}$ for
the Toeplitz case, partly due to the difficulty in characterizing the covering number for general low-rank Toeplitz matrices. Fortunately, the analysis for Toeplitz low-rank
matrices can be performed by means of a different method, as detailed
in the next section.

%Now that we have quantified the $\text{RIP}_{1,2}$ constants, it
%remains to establish the relevance between $\text{RIP}_{1,2}$ and
%exact recovery. This is asserted in the following lemma.

%\begin{lemma}\label{lemma:duality}Consider any $\boldsymbol{X}_{0}$
%obeying $\mathrm{rank}\left(\boldsymbol{X}_{0}\right)=r$. If there
%exists a number $K_{1}>2r$ such that $1-\delta_{2r+K_{1}}^{\text{lb}}>\left(1+\delta_{K_{1}}^{\text{ub}}\right)\sqrt{\frac{2r}{K_{1}}}$,
%then $\boldsymbol{X}_{0}$ is the unique solution to PhaseLift.\end{lemma}
%\begin{proof}See Appendix \ref{sec:Proof-of-Lemma-Duality}.\end{proof}

\section{Approximate $\ell_{2}/\ell_{2}$ Isometry for Toeplitz Low-Rank Matrices\label{sec:Approximate-Isometry:-RIP-22}}

While quadratic measurements in general do not exhibit RIP-$\ell_{2}/\ell_{2}$
(as introduced in \cite{RecFazPar07}) with respect to the set of
general low-rank matrices (as pointed out in \cite{candes2012phaselift}),
a slight variant of them can indeed satisfy RIP-$\ell_{2}/\ell_{2}$
when restricted to \emph{Toeplitz} low-rank matrices. In this section,
we first provide a characterization of RIP-$\ell_{2}/\ell_{2}$ for
the set of general low-rank matrices under bounded and near-isotropic
measurements, and then convert quadratic measurements into equivalent
isotropic measurements.

\subsection{RIP-$\ell_{2}/\ell_{2}$ for Near-Isotropic and Bounded Measurements}

Before proceeding to the Toeplitz low-rank matrices, we investigate
near-isotropic and bounded operators for the set of general low-rank
matrices as follows. For convenience of presentation, we repeat the
definition of RIP-$\ell_{2}/\ell_{2}$ as follows, followed by a theorem
characterizing RIP-$\ell_{2}/\ell_{2}$ for near-isotropic and bounded
operators.

\begin{definition}[\textbf{RIP-$\ell_{2}/\ell_{2}$ for low-rank
matrices}]\label{defn:LR-RIP-22}For the set of rank-$r$ matrices,
we define the RIP-$\ell_{2}/\ell_{2}$ constants $\delta_{r}$ w.r.t.
an operator $\mathcal{B}$ as the smallest number such that for all
$\boldsymbol{X}$ of rank at most $r$, 
\[
\left(1-\delta_{r}\right)\left\Vert \boldsymbol{X}\right\Vert _{\mathrm{F}}\leq\frac{1}{m}\left\Vert \mathcal{B}\left(\boldsymbol{X}\right)\right\Vert _{2}\leq\left(1+\delta_{r}\right)\left\Vert \boldsymbol{X}\right\Vert _{\mathrm{F}}.
\]
\end{definition}

\begin{theorem}\label{thm:RIP_Isotropic}Suppose that for all $1\leq i\leq m$,
\begin{equation}
\left\Vert \boldsymbol{B}_{i}\right\Vert \leq K\quad\text{and}\quad\left\Vert \mathbb{E}\left[\mathcal{B}_{i}^{*}\mathcal{B}_{i}\right]-\mathcal{I}\right\Vert \leq\frac{c_{5}}{n}\label{eq:B_bounded_RIP}
\end{equation}
hold for some quantity $K\leq n^{2}$. For any small constant $\delta>0$,
if $m>c_{0}rK^{2}\log^{7}n$, then with probability at least $1-1/n^{2}$,
one has%
\footnote{The proof of Theorem \ref{thm:RIP_Isotropic} follows the entropy
method introduced in \cite{rudelson2008sparse}, where $\log^{7}n$
factor is a natural consequence, and might be refined a bit by generic
chaining due to Talagrand \cite{talagrand1996majorizing} as employed
in \cite{CandesPlan2011RIPless}. However, we are unaware of an approach
that can get rid of the logarithmic factor.%
}

\vspace{0.3em}
i) $\mathcal{B}$ satisfies RIP-$\ell_{2}/\ell_{2}$ w.r.t. all matrices
of rank at most $r$ and obeys $\delta_{r}\leq\delta$;

\vspace{0.3em}
ii) Suppose that $\mathcal{K}$ is some convex set. Then for all $\boldsymbol{\Sigma}$
of rank at most $r$ and $\boldsymbol{\Sigma}\in\mathcal{K}$, if
$\left\Vert \by-\mathcal{B}(\boldsymbol{\Sigma})\right\Vert _{2}\leq\epsilon_{2}$,
the solution 
\begin{align*}
\hat{\boldsymbol{\Sigma}}=\argmin_{\boldsymbol{M}}\|\boldsymbol{M}\|_{*}\quad\mathrm{subject}\text{ }\mathrm{to}\quad & \left\Vert \by-\mathcal{B}(\boldsymbol{M})\right\Vert _{2}\leq\epsilon_{2},\\
 & \boldsymbol{M}\in\mathcal{K},
\end{align*}
satisfies 
\begin{equation}
\Vert\hat{\boldsymbol{\Sigma}}-\boldsymbol{\Sigma}\Vert_{\mathrm{F}}\leq C_{2}\frac{\epsilon_{2}}{\sqrt{m}}\label{eq:ApproxLR-2}
\end{equation}
for some universal constants $c_{0},C_{2},c_{5}>0$.\end{theorem}

\begin{proof}See Appendix \ref{sec:Proof-of-Theorem-RIP-Isotropic}.\end{proof}

In fact, the bound on $\left\Vert \boldsymbol{B}_{i}\right\Vert $
can be as small as $\Theta\left(\sqrt{n}\right)$, and we say a measurement
matrix $\boldsymbol{B}_{i}$ is \emph{well-bounded} if $K=O\left(\sqrt{n}\mathrm{polylog}n\right)$.
Simultaneously well-bounded and near-isotropic operators (i.e. those
satisfying (\ref{eq:B_bounded_RIP})) subsume the Fourier-type basis
discussed in \cite{Gross2011recovering}. Theorem \ref{thm:RIP_Isotropic}
strengthens the result in \cite{Gross2011recovering} to admit universal
and stable recovery of low-rank matrices with random subsampling using
Fourier-type basis, by justifying RIP-$\ell_{2}/\ell_{2}$ as soon
as $m=\Omega\left(nr\mathrm{polylog}n\right)$.

Unfortunately, Theorem \ref{thm:RIP_Isotropic} cannot be directly
applied to the class of Toeplitz low-rank matrices for the following
reasons: i) The sampling operator $\mathcal{A}$ is neither isotropic
nor well-bounded; ii) Theorem \ref{thm:RIP_Isotropic} requires $m=\Omega\left(nr\mathrm{polylog}n\right)$
measurements, which far exceeds the ambient dimension of a Toeplitz
matrix, which is $n$. This motivates us to construct another set
of equivalent sampling operators that satisfies the assumptions of Theorem
\ref{thm:RIP_Isotropic}, which is the focus of the following subsection.

\subsection{Construction of RIP-$\ell_{2}/\ell_{2}$ Operators for Toeplitz Low-rank
Matrices}

Note that the quadratic measurement matrices $\boldsymbol{A}_{i}=\boldsymbol{a}_{i}\boldsymbol{a}_{i}^{\top}$
are neither non-isotropic nor well bounded. For instance, when $\boldsymbol{a}_{i}\sim\mathcal{N}\left(\boldsymbol{0},\boldsymbol{I}_{n}\right)$,
simple calculation reveals that 
\begin{equation}
\left\Vert \boldsymbol{A}_{i}\right\Vert =\Theta\left(\sqrt{n}\right),\quad\text{and}\quad\mathbb{E}\left[\boldsymbol{A}_{i}\left\langle \boldsymbol{A}_{i},\boldsymbol{X}\right\rangle \right]=2\boldsymbol{X}+\text{tr}\left(\boldsymbol{X}\right)\cdot\boldsymbol{I},\label{eq:GaussianExpected}
\end{equation}
precluding $\boldsymbol{A}_{i}$'s from being isotropic. In order
to facilitate the use of Theorem \ref{thm:RIP_Isotropic}, we generate
a new set of measurement matrices $\tilde{\boldsymbol{B}}_{i}$ through
the following procedure. 
\begin{enumerate}
\itemsep0.3em
\item Define a set of matrices $\boldsymbol{B}_{i}$ of rank at most 3 
\begin{equation}
\boldsymbol{B}_{i}:=\begin{cases}
\frac{1}{2}\left(\boldsymbol{A}_{2i-1}-\boldsymbol{A}_{2i}\right),\quad & \text{if }\mu_{4}=3,\\
\alpha\boldsymbol{A}_{3i}+\beta\boldsymbol{A}_{3i-1}+\gamma\boldsymbol{A}_{3i-2},\quad & \text{if }\mu_{4}<3,
\end{cases}
\end{equation}
where $\alpha,\beta,\gamma$ are specified in Lemma \ref{lemma:isotropy}. 
\item Generate $M$ (whose choice will be specified later) matrices independently
such that 
\begin{equation}
\hat{\boldsymbol{B}}_{i}=\begin{cases}
\sqrt{n}\mathcal{T}\left(\boldsymbol{B}_{i}\right),\quad & \text{with probability }\frac{1}{n},\\
\sqrt{\frac{n}{n-1}}\mathcal{T}^{\perp}\left(\boldsymbol{G}_{i}\right), & \text{with probability }\frac{n-1}{n},
\end{cases}\label{eq:IsotropicBhat}
\end{equation}
where $\boldsymbol{G}_{i}$ is a random matrix with i.i.d. standard
Gaussian entries. 
\item Define a \emph{truncated} version $\tilde{\boldsymbol{B}}_{i}$ of
$\hat{\boldsymbol{B}}_{i}$ as follows 
\begin{equation}
\tilde{\boldsymbol{B}}_{i}=\hat{\boldsymbol{B}}_{i}\cdot{1}_{\left\{ \left\Vert \hat{\boldsymbol{B}}_{i}\right\Vert \leq c_{10}\log^{3/2}n\right\} },\quad1\leq i\leq M.\label{eq:TruncatedBtilde}
\end{equation}

\end{enumerate}
We will demonstrate that the $\tilde{\boldsymbol{B}}_{i}$'s are nearly-isotropic
and well-bounded, and hence by Theorem \ref{thm:RIP_Isotropic} the
associated operator $\tilde{\mathcal{B}}$ enables exact and stable
recovery for all rank-$r$ matrices when $M$ exceeds $O(nr\mathrm{polylog}n)$.
This in turn establishes Theorem \ref{thm:ToeplitzPhaseLift} through
an equivalence argument, detailed later.

\subsubsection{Isotropy Trick}

While $\boldsymbol{A}_{i}$'s are in general non-isotropic, a linear
combination of them can be made isotropic when restricted to Toeplitz
matrices. This is stated in the following lemma.

\begin{lemma}\label{lemma:isotropy}Consider the sub-Gaussian sampling
model in \eqref{sampling}.

1) When $\mu_{4}=3$, then for any $\boldsymbol{X}$, the matrix 
\begin{equation}
\boldsymbol{B}_{i}=\frac{1}{2}\left(\boldsymbol{A}_{2i-1}-\boldsymbol{A}_{2i}\right)\label{eq:GaussianIsotropyConstruction}
\end{equation}
satisfies 
\begin{equation}
\mathbb{E}\left[\boldsymbol{B}_{i}\left\langle \boldsymbol{B}_{i},\boldsymbol{X}\right\rangle \right]=\text{ }\boldsymbol{X}.
\end{equation}

2) When $\mu_{4}<3$, take any constant $\xi>0$ obeying $\xi^{2}>1.5(3-\mu_{4})$
and set 
\begin{equation}
\boldsymbol{B}_{i}=\alpha\boldsymbol{A}_{3i}+\beta\boldsymbol{A}_{3i-1}+\gamma\boldsymbol{A}_{3i-2},\label{eq:IsotropyCombGeneral}
\end{equation}
with the choice of $\Delta:=-\left(1-\frac{\xi}{n}\right)^{2}-2+\frac{2\xi^{2}}{3-\mu_{4}}$,
\begin{equation}
\begin{cases}
\alpha=\sqrt{\frac{3-\mu_{4}}{2\xi^{2}}},\\
\beta:=\frac{-\left(1-\frac{\xi}{\sqrt{n}}\right)+\sqrt{\Delta}}{2}\alpha,\\
\gamma:=\frac{-\left(1-\frac{\xi}{\sqrt{n}}\right)-\sqrt{\Delta}}{2}\alpha.
\end{cases}\label{eq:alpha_beta_gamma}
\end{equation}
Then, for any norm $\left\Vert \cdot\right\Vert _{\mathrm{n}}$ and
any $\boldsymbol{X}$ that satisfies $\boldsymbol{X}_{11}=\boldsymbol{X}_{22}=\cdots=\boldsymbol{X}_{nn}$,
one has 
\begin{equation}
\begin{cases}
\mathbb{E}\left[\boldsymbol{B}_{i}\right]=\text{ }\sqrt{\frac{3-\mu_{4}}{2n}};\\
\mathbb{E}\left[\boldsymbol{B}_{i}\left\langle \boldsymbol{B}_{i},\boldsymbol{X}\right\rangle \right]=\text{ }\boldsymbol{X};\\
\left\Vert \boldsymbol{B}_{i}\right\Vert _{\mathrm{n}}\leq\sqrt{3}\max_{i:1\leq i\leq m}\left\Vert \boldsymbol{A}_{i}\right\Vert _{\mathrm{n}}.
\end{cases}\label{eq:IsotropyB-General}
\end{equation}
\end{lemma} \begin{proof} See Appendix~\ref{proof:lemma:isotropy}.
\end{proof}

Lemma \ref{lemma:isotropy} asserts that a large class of measurement
matrices can be made isotropic when restricted to the class of matrices
with identical diagonal entries (e.g. Toeplitz matrices). This immediately
implies that the operator $\hat{\mathcal{B}}$ associated with $\hat{\boldsymbol{B}}_{i}$'s
(defined in (\ref{eq:IsotropicBhat})) is isotropic. Specifically,
for any symmetric $\boldsymbol{X}$, 
\begin{align*}
 & \mathbb{E}\left[\hat{\boldsymbol{B}}_{i}\left\langle \hat{\boldsymbol{B}}_{i},\boldsymbol{X}\right\rangle \right]\\
 & \quad=\mathbb{E}\left[\mathcal{T}\left(\boldsymbol{B}_{i}\right)\left\langle \boldsymbol{B}_{i},\mathcal{T}\left(\boldsymbol{X}\right)\right\rangle \right]+\mathbb{E}\left[\mathcal{T}^{\perp}\left(\boldsymbol{G}_{i}\right)\left\langle \boldsymbol{G}_{i},\mathcal{T}^{\perp}\left(\boldsymbol{X}\right)\right\rangle \right]\\
 & \quad=\mathcal{T}\left(\mathbb{E}\left[\boldsymbol{B}_{i}\left\langle \boldsymbol{B}_{i},\mathcal{T}\left(\boldsymbol{X}\right)\right\rangle \right]\right)+\mathcal{T}^{\perp}\left(\mathbb{E}\left[\boldsymbol{G}_{i}\left\langle \boldsymbol{G}_{i},\mathcal{T}^{\perp}\left(\boldsymbol{X}\right)\right\rangle \right]\right)\\
 & \quad=\mathcal{T}\left(\mathcal{T}\left(\boldsymbol{X}\right)\right)+\mathcal{T}^{\perp}\left(\mathcal{T}^{\perp}\left(\boldsymbol{X}\right)\right)=\boldsymbol{X},
\end{align*}
which follows since $\boldsymbol{B}_{i}$ and $\boldsymbol{G}_{i}$
are both isotropic matrices, a consequence of Lemma \ref{lemma:isotropy}.

\subsubsection{Truncation of $\hat{\mathcal{B}}$ is near-isotropic}

The operators associated with $\hat{\boldsymbol{B}}_{i}$'s are in
general not well-bounded. Fortunately, $\hat{\boldsymbol{B}}_{i}$'s
are well-bounded with high probability, which comes from the following
lemma whose proof can be found in Appendix~\ref{sec:Proof-of-Lemma:UB-T(zz)}.

\begin{comment}
Throughout the rest of the paper, we use $\boldsymbol{B}_{i}$ to
denote the isotropic operators in the form of (\ref{eq:IsotropyCombGeneral})
if $\mu_{4}<3$ and in the form of (\ref{eq:GaussianIsotropyConstruction})
if $\mu_{4}=3$. For notational simplicity, we introduce sampling
operators $ $$\mathcal{B}_{i}$'s such that 
\begin{equation}
\mathcal{B}_{i}\left(\boldsymbol{X}\right)=\left\langle \boldsymbol{B}_{i},\boldsymbol{X}\right\rangle .\label{eq:SamplingMatrixOperator}
\end{equation}
Set $\tilde{m}=\frac{m}{3}$ if $\mu_{4}<3$ and $\tilde{m}=\frac{m}{2}$
if $\mu_{4}=3$. Let $\mathcal{B}:\mathbb{C}^{n\times n}\mapsto\mathbb{C}^{\tilde{m}}$
be the linear operator that maps a Hermitian matrix $\boldsymbol{X}$
to $\left\{ \mathcal{B}_{i}\left(\boldsymbol{X}\right)\right\} _{1\leq i\leq\tilde{m}}$,
and let $\mathcal{B}:\mathbb{C}^{\tilde{m}}\mapsto\mathbb{C}^{n\times n}$
denote its conjugate operator. Specifically, 
\begin{equation}
\mathcal{B}\left(\boldsymbol{X}\right):=\left(\mathcal{B}_{1}\left(\boldsymbol{X}\right),\cdots,\mathcal{B}_{\tilde{m}}\left(\boldsymbol{X}\right)\right),\quad\text{and}\quad\mathcal{B}^{*}\left(\boldsymbol{\lambda}\right)=\sum_{i=1}^{\tilde{m}}\lambda_{i}\boldsymbol{B}_{i}.\label{eq:IsotropicB}
\end{equation}
In the rest of the paper, we will abuse notation and use $\tilde{m}$
and $m$ interchangeably. 
\end{comment}

\begin{lemma}\label{lemma:UB-T(zz)}Consider a random vector $\boldsymbol{z}$
that follows the sub-Gaussian sampling model as described in \eqref{sampling}.
There exists an absolute constant $c_{10}>0$ such that 
\begin{equation}
\left\Vert \mathcal{T}\left(\boldsymbol{z}\boldsymbol{z}^{\top}\right)\right\Vert \leq c_{12}\log^{\frac{3}{2}}n\label{eq:ToeplitzNorm}
\end{equation}
holds with probability exceeding $1-n^{-10}$. \end{lemma}

As $\left\Vert \boldsymbol{B}_{i}\right\Vert $ can be bounded above
by $\max_{1\leq i\leq m}\left\Vert \boldsymbol{A}_{i}\right\Vert $
up to some constant factor, Lemma \ref{lemma:UB-T(zz)} reveals that $\left\Vert \mathcal{T}\left(\boldsymbol{B}_{i}\right)\right\Vert $ can be well controlled
for sub-Gaussian vectors, i.e. 
\begin{equation}
\left\Vert \mathcal{T}\left(\boldsymbol{B}_{i}\right)\right\Vert \leq c_{10}\log^{\frac{3}{2}}n,\quad1\leq i\leq m\label{eq:BoundTB}
\end{equation}
with probability exceeding $1-3n^{-8}$. Similarly, classical results
in random matrices (e.g. \cite{Tao2012RMT}) assert that $\left\Vert \boldsymbol{G}_{i}\right\Vert $
can also be bounded above by $O\left(\sqrt{n}\log n\right)$ with
overwhelming probability. These bounds taken collectively suggest
that 
\begin{equation}
\Vert\hat{\boldsymbol{B}}_{i}\Vert\leq K:=c_{10}\sqrt{n}\log^{\frac{3}{2}}n,\quad1\leq i\leq m
\end{equation}
for some constant $c_{10}>0$ with probability exceeding $1-n^{-7}$.

The above stochastic boundedness property motivates us to study
the truncated version $\tilde{\boldsymbol{B}}_{i}$ of $\hat{\boldsymbol{B}}_{i}$
as defined in (\ref{eq:TruncatedBtilde}). Interestingly, $\tilde{\boldsymbol{B}}_{i}$
is near-isotropic, a consequence of the following lemma whose proof
can be found in Appendix~\ref{sec:Proof-of-Lemma-Deviation-BB}.

\begin{lemma}\label{lemma-Deviation-BB}Suppose that the restriction
of $\mathcal{B}_{i}$ to Toeplitz matrices is isotropic. Consider
any event $E$ obeying $\mathbb{P}\left(E\right)\geq1-\frac{1}{n^{5}}.$
Then there is some constant $c_{5}>0$ such that 
\begin{equation}
\left\Vert \mathbb{E}\left[\mathcal{T}\mathcal{B}_{i}^{*}\mathcal{B}_{i}\mathcal{T}{\bf 1}_{E}\right]-\mathcal{T}\right\Vert \leq\frac{c_{5}}{n^{2}}.\label{eq:StochasticIsotropy}
\end{equation}
\end{lemma}

The truncated version of $\boldsymbol{G}_{i}$ can be easily bounded
as in \cite{CandesPlan2011RIPless}, which we omit for simplicity
of presentation. This combined with (\ref{eq:StochasticIsotropy})
indicates that 
\begin{align*}
 & \left\Vert \mathbb{E}\left[\tilde{\mathcal{B}}_{i}^{*}\tilde{\mathcal{B}}_{i}\right]-\mathcal{I}\right\Vert \\
 & \quad\leq\left\Vert \mathbb{E}\left[\mathcal{T}\mathcal{B}_{i}^{*}\mathcal{B}_{i}\mathcal{T}\right]-\mathcal{\mathcal{T}}\right\Vert +\left\Vert \mathbb{E}\left[\mathcal{T}^{\perp}\mathcal{G}_{i}^{*}\mathcal{G}_{i}\mathcal{T}^{\perp}\right]-\mathcal{\mathcal{T}}^{\perp}\right\Vert \\
 & \quad\leq\frac{c_{5}}{n}.
\end{align*}

\subsection{Proof of Theorem \ref{thm:ToeplitzPhaseLift}}

So far we have demonstrated that $\tilde{\boldsymbol{B}}_{i}$'s are
near-isotropic and satisfy $\Vert \tilde{\boldsymbol{B}}_{i}\Vert =O\left(\sqrt{n}\log^{\frac{3}{2}}n\right)$.
Suppose that $\Vert \boldsymbol{y}-\tilde{\mathcal{B}}\left(\boldsymbol{\Sigma}\right)\Vert _{2}\leq\tilde{\epsilon}_{2}$.
Theorem \ref{thm:RIP_Isotropic} implies that if $M$ exceeds $\Theta\left(nr\log^{10}(n)\right)$,
then the solution to 
\begin{align}
\tilde{\boldsymbol{\Sigma}}:=\argmin_{\boldsymbol{M}}\left\Vert \boldsymbol{M}\right\Vert _{*}\quad\text{subject to } & \Vert \boldsymbol{y}-\tilde{\mathcal{B}}\left(\boldsymbol{M}\right)\Vert _{2}\leq\tilde{\epsilon}_{2},\nonumber \\
 & \boldsymbol{M}\text{ is Toeplitz},\label{eq:Equivalent}
\end{align}
satisfies 
\begin{equation}
\left\Vert \boldsymbol{\Sigma}-\tilde{\boldsymbol{\Sigma}}\right\Vert _{\text{F}}\leq C_{2}\frac{\tilde{\epsilon}_{2}}{\sqrt{M}}\label{eq:ApproxLR-2-1}
\end{equation}
for the entire set of rank-$r$ matrices $\boldsymbol{\Sigma}$. Apparently,
such low-rank manifolds subsume all rank-$r$ Toeplitz matrices as
special cases. This claim in turn establishes Theorem \ref{thm:ToeplitzPhaseLift}
through the following argument: 
\begin{enumerate}
\itemsep0.3em
\item From (\ref{eq:IsotropicBhat}) and the Chernoff bound, $\tilde{\mathcal{B}}$
entails $\Theta\left(\frac{M}{n}\right)=\Theta\left(r\log^{10}n\right)$
independent copies of $\sqrt{n}\mathcal{T}\left(\boldsymbol{B}_{i}\right)$,
and all other measurements are on the orthogonal complement of the
Toeplitz space. 
\item For any rank-$r$ Toeplitz matrix $\boldsymbol{\Sigma}$, the original
$\mathcal{A}$ entails $m/3>\Theta\left(r\log^{10}n\right)$ measurement
matrices of the form $\mathcal{T}\left(\boldsymbol{B}_{i}\right)$,
and any non-Toeplitz component of $\boldsymbol{X}$ is perfectly known
(i.e. equal to 0). This indicates that the convex program (\ref{eq:ToepMC})
is tighter than (\ref{eq:Equivalent}) when $\tilde{\epsilon}_{2}=\Theta\left(\sqrt{n}\epsilon_{2}\right)$,
i.e. one can construct (via coupling) a new probability space over
which if the solution $\tilde{\boldsymbol{\Sigma}}$ to (\ref{eq:Equivalent})
is exact and unique, then it will be the unique solution to (\ref{eq:ToepMC})
as well. This combined with the universal bound (\ref{eq:ApproxLR-2-1})
establishes Theorem \ref{thm:ToeplitzPhaseLift}. 
\end{enumerate}

\section{Numerical Examples}

\label{sec:numerical}

To demonstrate the practical applicability of the proposed convex
relaxation under quadratic sensing, in this section we present a variety
of numerical examples for low-rank or sparse covariance matrix estimation.

\subsection{Recovery of Low-Rank Covariance Matrices}

We conduct a series of Monte Carlo trials for various parameters.
Specifically, we choose $n=50$, and for each $(m,r)$ pair, we repeat
the following experiments 20 times. We generate $\bSigma$, an $n\times n$
PSD matrix via $\boldsymbol{\Sigma}=\boldsymbol{L}\boldsymbol{L}^{\top}$,
where $\boldsymbol{L}$ is a randomly generated $n\times r$ matrix
with independent Gaussian components. The sensing vectors are generated
as i.i.d. Gaussian vectors and Bernoulli vectors, and we obtain noiseless
quadratic measurements $\boldsymbol{y}$. We use the off-the-shelf
SDP solver SDPT3 with the modeling software CVX, and declare a matrix
$\bSigma$ to be recovered if the solution $\hat{\bSigma}$ returned
by the solver satisfies $\|\hat{\boldsymbol{\Sigma}}-\boldsymbol{\Sigma}\|_{\text{F}}/\|\boldsymbol{\Sigma}\|_{\text{F}}<10^{-3}$.
Figure \ref{fig:PTlowrank} illustrates the empirical probability
of successful recovery in these Monte Carlo trials, which is reflected
through the color of each cell. In order to compare the optimality
of the practical performance, we also plot the information theoretic
limit in red lines, i.e. the fundamental lower limit on $m$ required
to recover all rank-$r$ matrices, which is $nr-r(r-1)/2$ in our
case. It turns out that the practical phase transition curve is very
close to the theoretic sampling limit, which demonstrates the optimality
of our algorithm. 
\begin{figure*}[htp]
\centering %
\begin{tabular}{cc}
\includegraphics[width=0.37\textwidth]{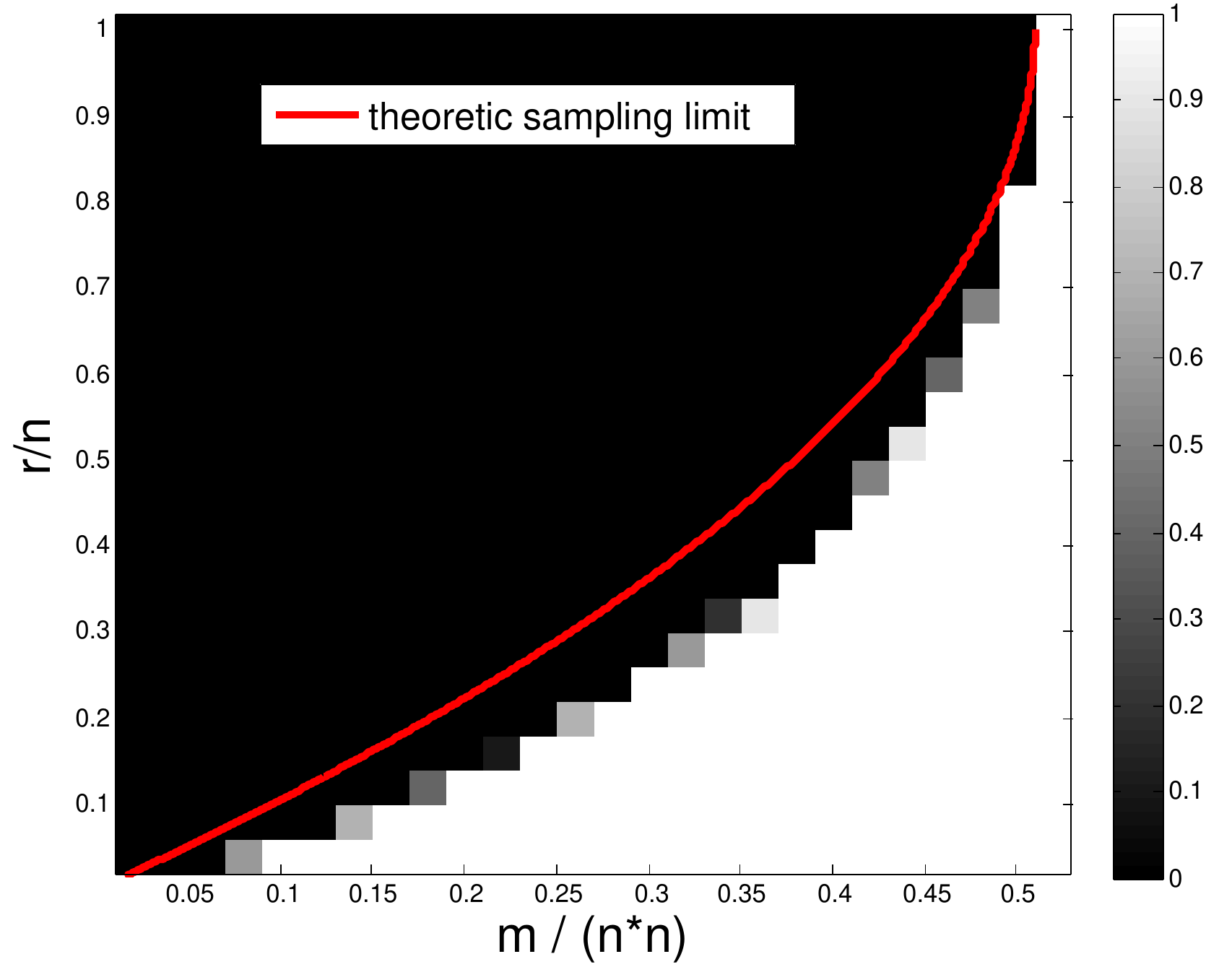}  & \includegraphics[width=0.38\textwidth]{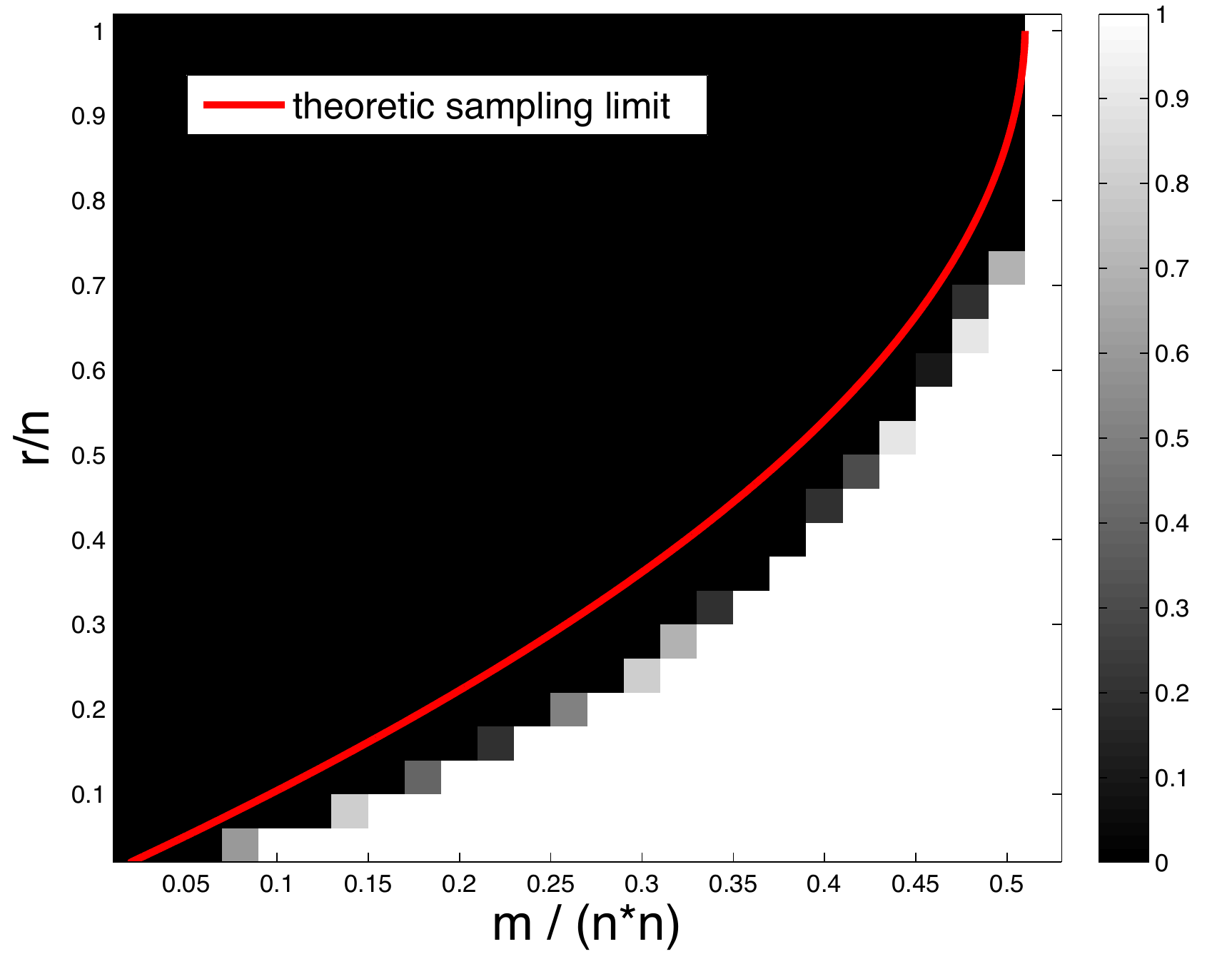}\tabularnewline
(a)  & (b)\tabularnewline
\end{tabular}\caption{\label{fig:PTlowrank}Recovery of covariance matrices from quadratic
measurements when $n=50$. For each ($m,r$) pair, we repeated Monte
Carlo trials 20 times. A PSD matrix $\bSigma$ and $m$ sensing vectors
are selected at random. The colormap for each cell indicates the empirical
probability of success, and the red line reflects the fundamental
information theoretic limit. The results are shown for (a) Gaussian
sensing vectors and (b) symmetric Bernoulli sensing vectors.}
\end{figure*}

In the second numerical example, we consider a random covariance matrix
generated via the same procedure as above but with $n=40$. We let
the rank $r$ vary as $1,3,5,10$ and the number of measurements $m$
vary from $20$ to $600$. For each pair of $(r,m)$, we perform $10$
independent experiments where in each run the sensing matrix is generated
with i.i.d. Gaussian entries. Fig.~\ref{sketch_tracemin} (a) shows
the average Normalized Mean Squared Error (NMSE) defined as $\|\hat{\boldsymbol{\Sigma}}-\boldsymbol{\Sigma}\|_{\text{F}}^{2}/\|\boldsymbol{\Sigma}\|_{\text{F}}^{2}$
with respect to $m$ for different ranks when there is no noise. We
further introduce additive bounded noise to each measurement by letting
$\lambda_{i}$ be generated from $\sigma\cdot\mathcal{U}[-1,1]$,
where $\mathcal{U}[-1,1]$ is a uniform distribution on $[-1,1]$,
$\sigma$ is the noise level. Fig.~\ref{sketch_tracemin} (b) shows
the average NMSE when $r=5$ for different noise levels by setting
$\epsilon=\sigma m$ in \eqref{tracemin}. 
\begin{figure*}[htb]
\centering %
\begin{tabular}{cc}
\includegraphics[width=0.37\textwidth]{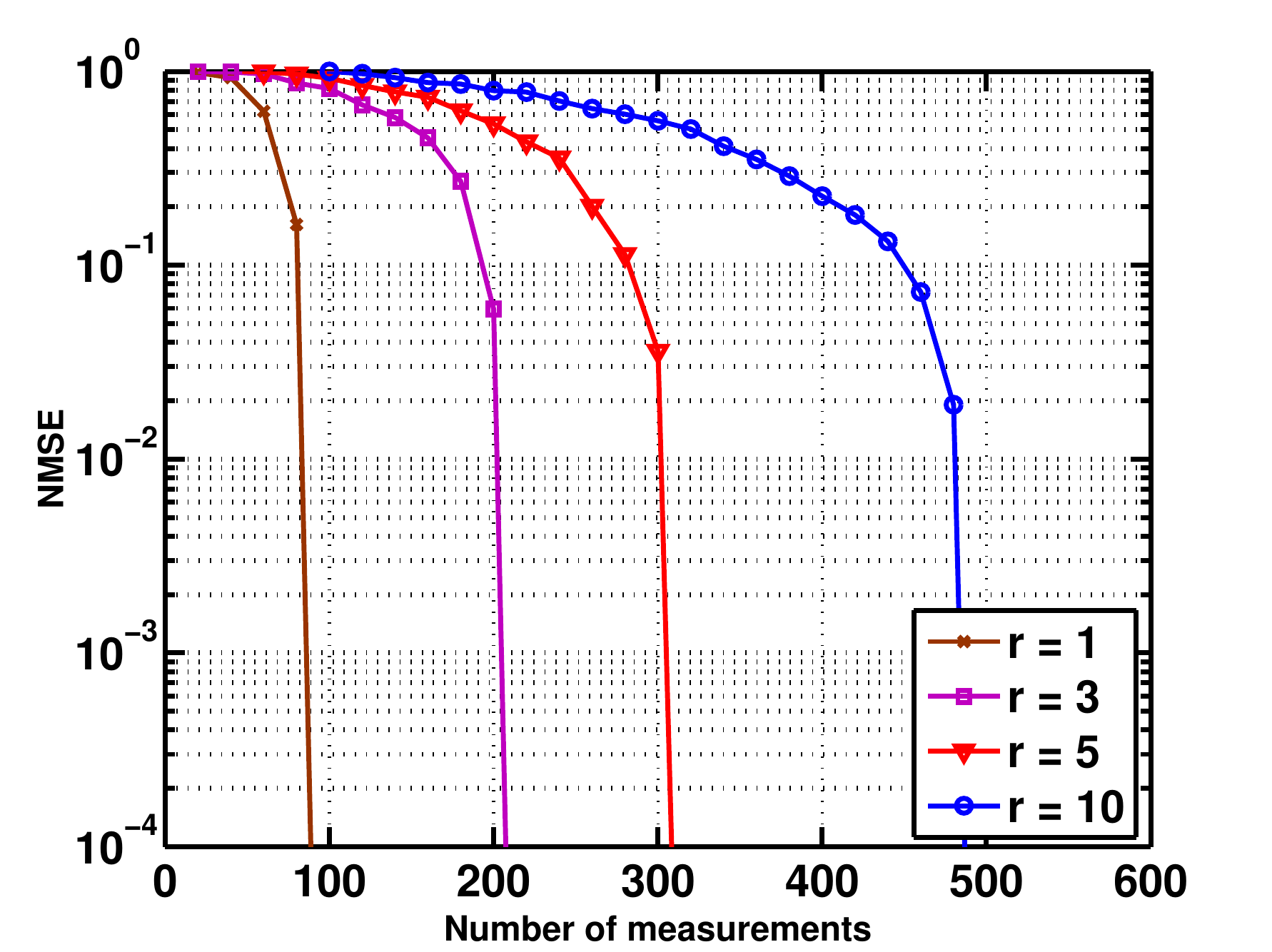}  & \includegraphics[width=0.37\textwidth]{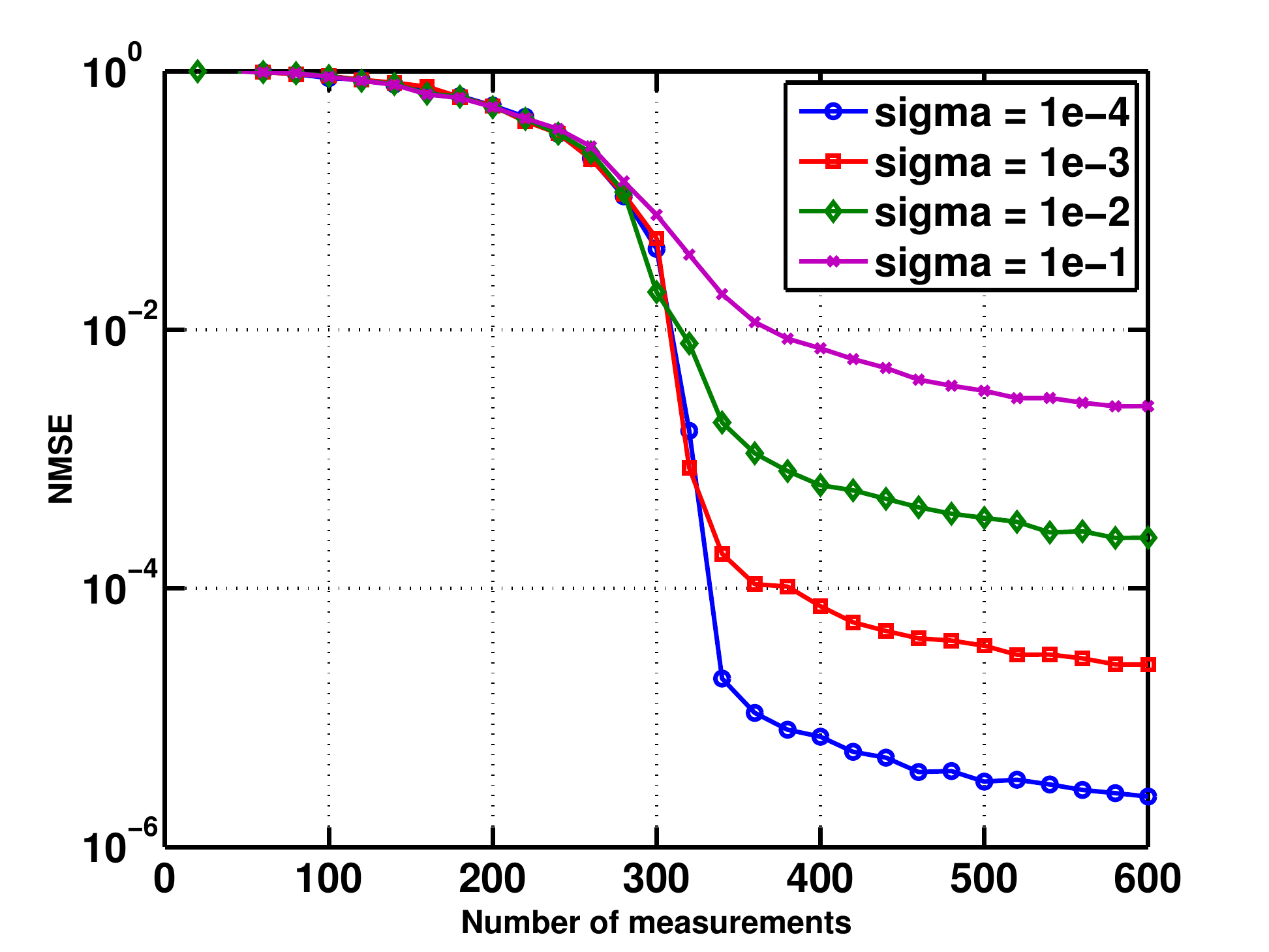} \tabularnewline
(a)  & (b) \tabularnewline
\end{tabular}\caption{The NMSE of the reconstructed covariance matrix via trace minimization
vs. the number of measurements when $n=40$: (a) for different ranks
when no noise is present; (b) for different noise levels when $r=5$.}

\label{sketch_tracemin} 
\end{figure*}

Interestingly, \cite{candes2012solving,demanet2012stable} showed
that when the covariance matrix is rank-one, if $m=O(n)$, the intersection
of two convex sets, namely $\mathcal{S}_{1}=\{\boldsymbol{M}:\mathcal{A}(\boldsymbol{M})=\boldsymbol{y}\}$
and $\mathcal{S}_{2}=\{\boldsymbol{M}:\boldsymbol{M}\succ0\}$ is
a singleton, with high probability. For the low-rank case, if the
same conclusion holds, we can find the solution via alternating projection
between two convex sets. Therefore, we experiment on the following
Projection Onto Convex Sets (POCS) procedure: 
\begin{equation}
\boldsymbol{\Sigma}_{t+1}=\mathcal{P}_{\mathcal{S}_{2}}\mathcal{P}_{\mathcal{S}_{1}}\boldsymbol{\Sigma}_{t},
\end{equation}
where $\mathcal{P}_{\mathcal{S}_{2}}$ denotes the projection onto
the PSD cone, and 
\begin{equation}
\mathcal{P}_{\mathcal{S}_{1}}\boldsymbol{\Sigma}_{t}:=\boldsymbol{\Sigma}_{t}-\mathcal{A}^{*}(\mathcal{A}\mathcal{A}^{*})^{-1}(\mathcal{A}(\boldsymbol{\Sigma_{t}})-\boldsymbol{y}).
\end{equation}

Fig.~\ref{sketch_pocs_eg} (a) shows the NMSE of the reconstruction
with respect to the number of iterations for $r=3$ and different
$m=200,250,300,350$. By comparing Fig.~\ref{sketch_tracemin}, we
see that it requires more measurements for the POCS procedure to succeed,
but the computational cost is much lower than the trace minimization.
This is further validated from Fig.~\ref{sketch_pocs_eg} (b), which
is obtained under the same simulation setup as Fig.~\ref{sketch_tracemin}
by repeating POCS with $2000$ iterations. 
\begin{figure*}[htp]
\centering %
\begin{tabular}{cc}
\includegraphics[width=0.37\textwidth]{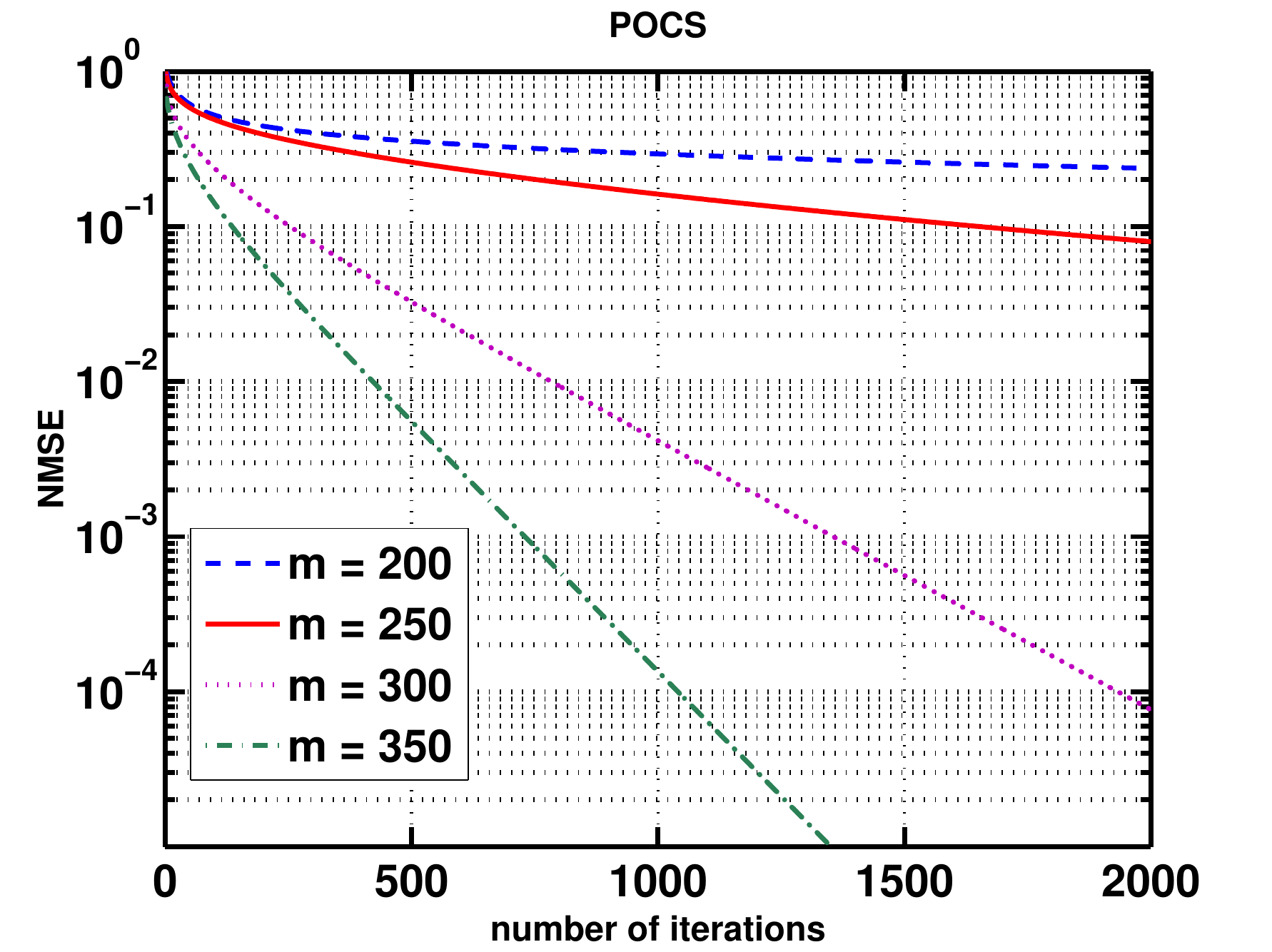}  & \includegraphics[width=0.37\textwidth]{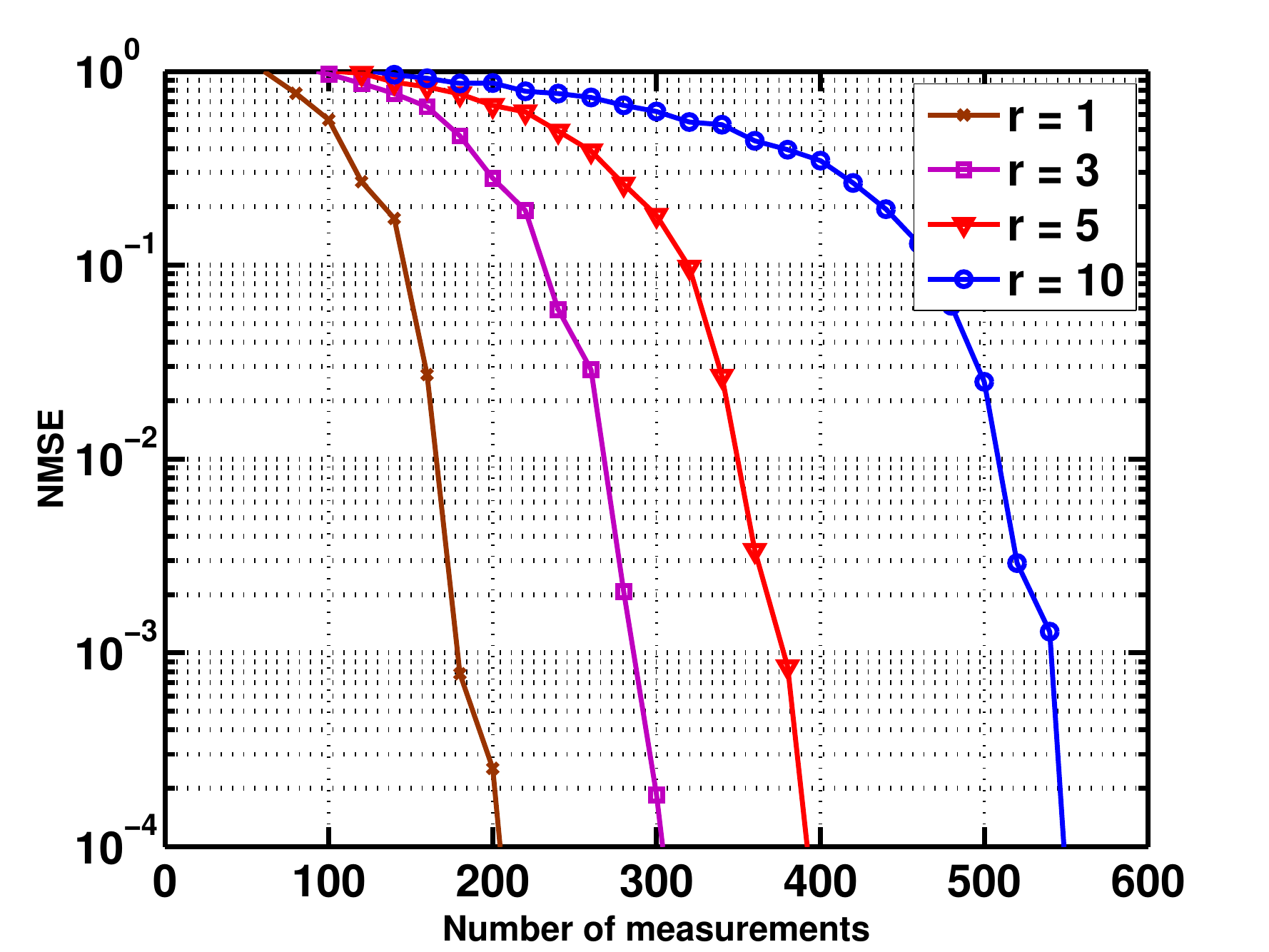} \tabularnewline
(a)  & (b) \tabularnewline
\end{tabular}\caption{The NMSE of the reconstructed covariance matrix via POCS for $n=40$:
(a) the NMSE vs. the number of iterations for different numbers of
measurements when $r=3$; (b) the NMSE vs. the number of measurements
for different ranks when running $2000$ iterations.}

\label{sketch_pocs_eg} 
\end{figure*}

\subsection{Recovery of Toeplitz Low-rank Matrices}

To justify the convex heuristic for Toeplitz low-rank matrices, we
perform a series of numerical experiments for matrices of dimension
$n=50$. By Caratheodory's theorem, each PSD Toeplitz matrix can be
uniquely decomposed into a linear combination of line spectrum \cite{grenander1958toeplitz}.
Thus, we generate the PSD Toeplitz matrix by randomly generating the
frequencies and amplitudes of each line spectra. In the real-valued
case, the underlying spectral spikes occur in conjugate pairs (i.e.
$\left(f_{1},-f_{1}\right),\left(f_{2},-f_{2}\right),\cdots$). We
independently generate $r/2$ frequency pairs within the unit disk
uniformly at random, and the amplitudes are generated as the absolute
values of i.i.d. Gaussian variables. Figure \ref{fig:Phase-transition}
illustrates the phase transition diagram for varying choices of $(m,r)$.
Each trial is declared successful if the estimate $\hat{\boldsymbol{\Sigma}}$
satisfies $\Vert\hat{\boldsymbol{\Sigma}}-\boldsymbol{\Sigma}\Vert_{\text{F}}/\left\Vert \boldsymbol{\Sigma}\right\Vert _{\text{F}}<10^{-3}.$
The empirical success rate is calculated by averaging over 50 Monte
Carlo trials, and is reflected by the color of each cell. While there
are in total $r$ degrees of freedom, our algorithm exhibits an approximately
linear phase transition curve, which confirms our theoretical prediction
in the absence of noise.

\begin{figure}[htp]
\centering{}\emph{\includegraphics[width=0.37\textwidth]{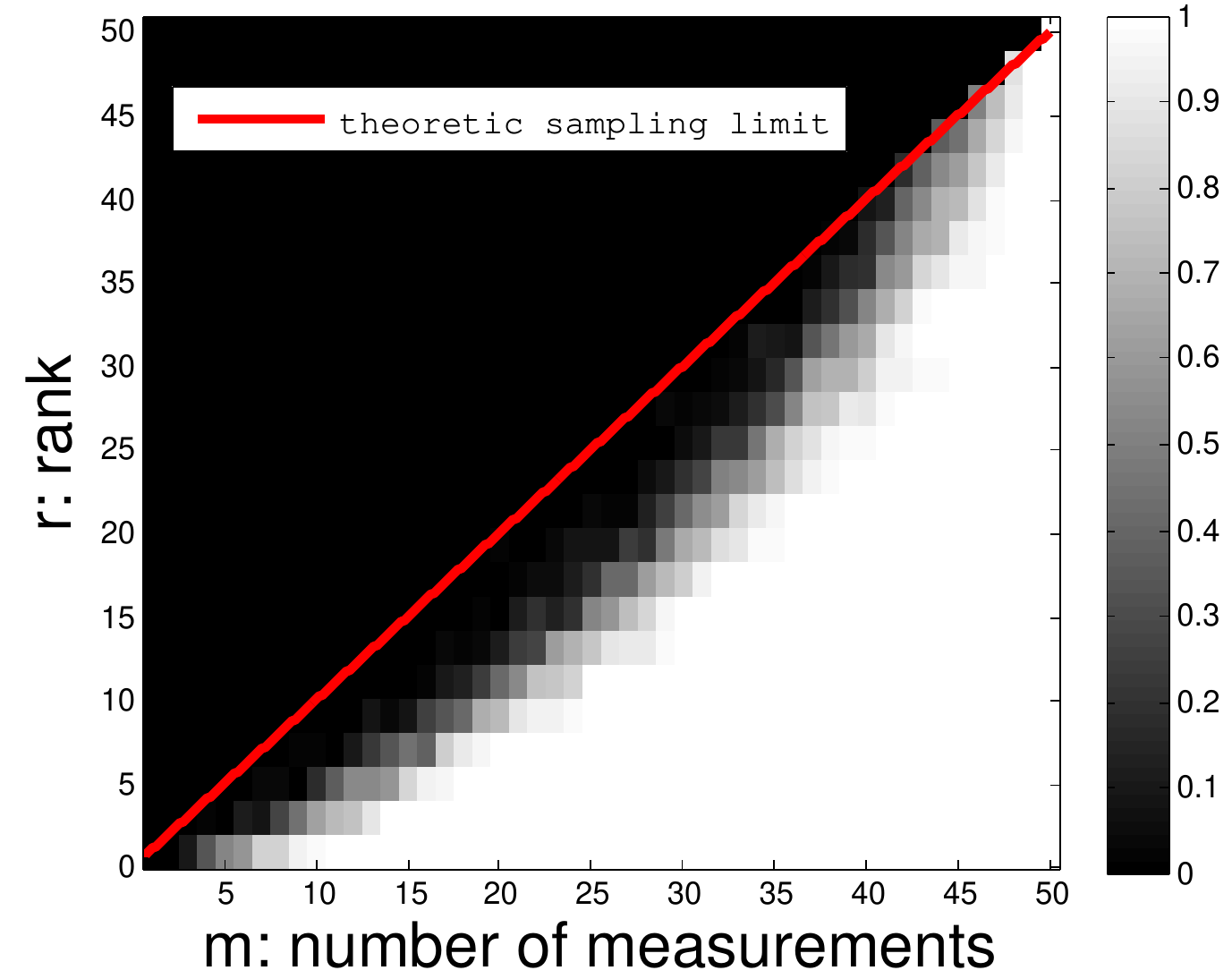}}\caption{\label{fig:Phase-transition}Phase transition plots where frequency
locations are randomly generated. The plot corresponds to the situation
where $n=50$. The empirical success rate is calculated by averaging
over 50 Monte Carlo trials. }
\end{figure}

\subsection{Recovery of Sparse Matrices}

We perform a series of Monte Carlo trials for various parameters for
matrices of dimensions $50\times50$. We first generate PSD sparse
covariance matrices in the following way. For each sparsity value
$k$, we generate a $\sqrt{k}\times\sqrt{k}$ matrix via $\boldsymbol{\Sigma}_{k}=\boldsymbol{L}\boldsymbol{L}^{\top}$,
where $\boldsymbol{L}$ is a $\sqrt{k}\times\sqrt{k}$ matrix with
independent Gaussian components. We then randomly select $\sqrt{k}$
rows and columns of $\bSigma$ and embed $\bSigma_{k}$ into the corresponding
$\sqrt{k}\times\sqrt{k}$ submatrix; all other entries of $\bSigma$
are set to 0. In addition, we also conduct numerical simulations for
general symmetric sparse matrices, where the non-zero entries are
drawn from an i.i.d. Gaussian distribution and the support is randomly
chosen. For each $(m,k)$ pair in each scenario, we repeated the experiments
20 times, and solve it using CVX. Again, a matrix $\bSigma$ is claimed
to be recovered if the solution $\hat{\bSigma}$ returned by the solver
satisfies $\|\hat{\boldsymbol{\Sigma}}-\boldsymbol{\Sigma}\|_{\text{F}}/\|\boldsymbol{\Sigma}\|_{\text{F}}<10^{-3}$.
Figure \ref{fig:PTSparse} illustrates the empirical success probability
in these Monte Carlo experiments. For ease of comparison, we also
plot the degrees of freedom in red lines, which is $\frac{\sqrt{k}\left(\sqrt{k}+1\right)}{2}$
in our case. It turns out that the practical phase transition curve
is close to the theoretic sampling limit, which demonstrates the optimality
of our algorithm. 
\begin{figure*}[htp]
\centering %
\begin{tabular}{cc}
\includegraphics[width=0.37\textwidth]{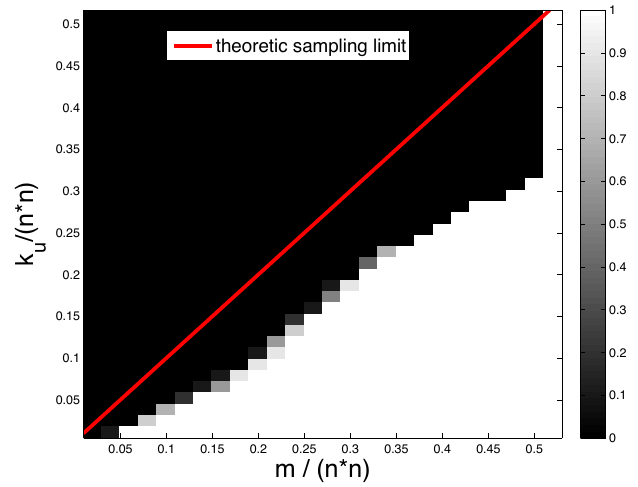}  & \includegraphics[width=0.37\textwidth]{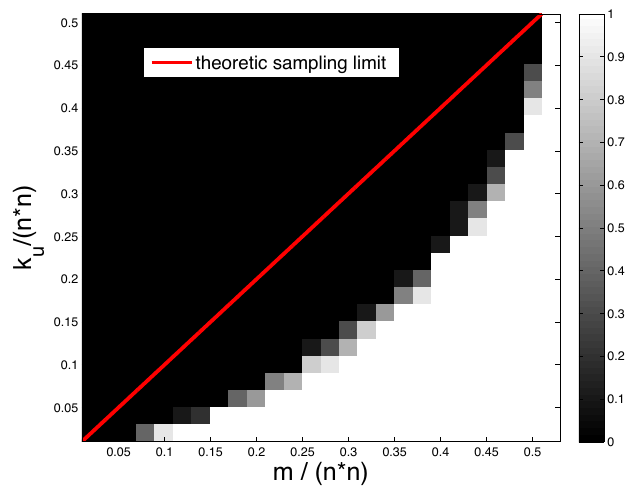}\tabularnewline
(a)  & (b)\tabularnewline
\end{tabular}\caption{\label{fig:PTSparse}Reconstruction of sparse matrices from Gaussian
quadratic measurements when $n=50$. For ease of comparison, we let
$k_{\text{u}}$ denote the number of non-zero entries above or on
the main diagonal, which represents the degrees of freedom for symmetric
matrices. For each ($m,k_{\text{u}}$) pair, we conducted Monte Carlo
experiments 20 times. A PSD matrix $\bSigma$ and $m$ sensing vectors
are selected at random. The colormap for each cell and the red line
reflects the empirical probability of success and the information
theoretic limit, respectively. The results are shown for (a) sparse
PSD matrices, and (b) sparse symmetric matrices.}
\end{figure*}

Another numerical example concerns recovery of a random symmetric
sparse matrix (not necessarily PSD). We randomly generated a symmetric
sparse matrix of sparsity level $k$ with $n=40$, and sketched it
with i.i.d. Gaussian vectors. For each pair of $(r,m)$, we perform
$10$ independent runs where in each run the sensing matrix is generated
with i.i.d. standard Gaussian entries. Fig.~\ref{sketch_l1min} (a)
shows the average NMSE with respect to $m$ for different sparsity
levels when there is no noise. We further introduce additive bounded
noise to each measurement by letting $\lambda_{i}$ be generated from
$\sigma\cdot\mathcal{U}[-1,1]$, and run $10$ trials for each pair
of $(\sigma,m)$. Fig.~\ref{sketch_l1min} (b) shows the average
NMSE when $k=240$ for different noise levels by setting $\epsilon=\sigma m$
in \eqref{l1min}. 
\begin{figure*}[htp]
\centering %
\begin{tabular}{cc}
\includegraphics[width=0.37\textwidth]{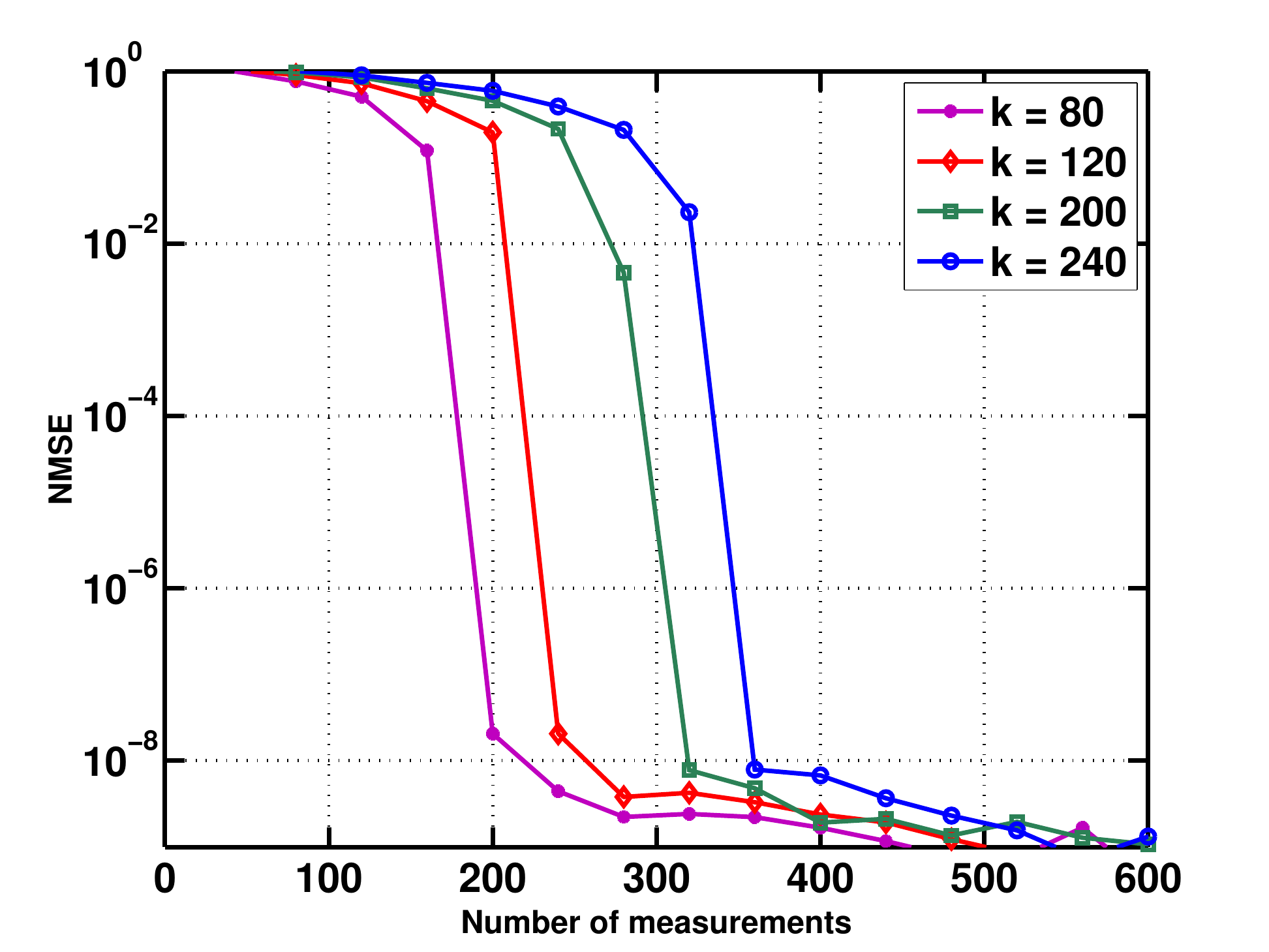}  & \includegraphics[width=0.37\textwidth]{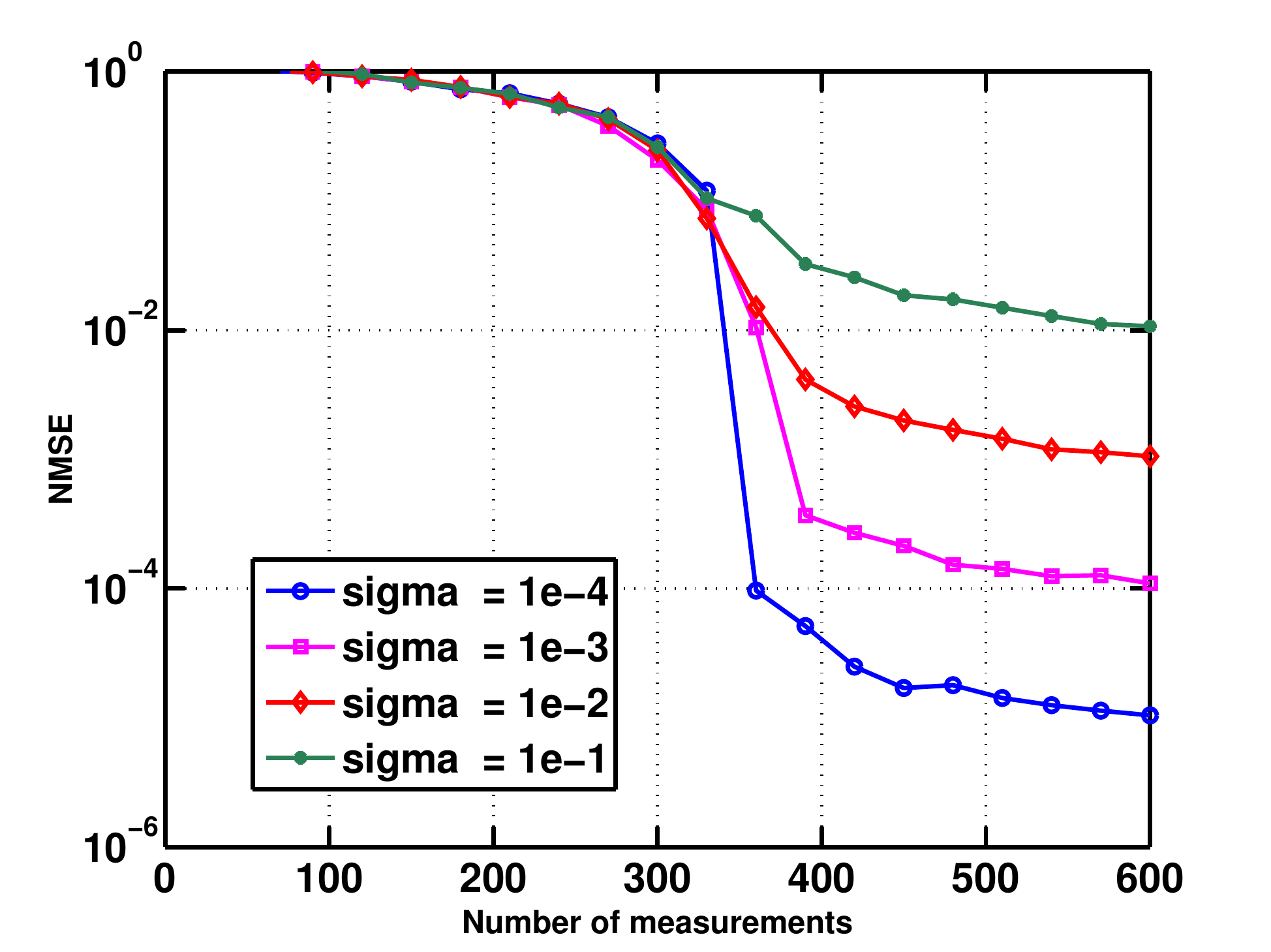} \tabularnewline
(a)  & (b) \tabularnewline
\end{tabular}\caption{The NMSE of the reconstructed sparse matrix via $\ell_{1}$ minimization
vs. the number of measurements when $n=40$: (a) for different sparsity
level when no noise is present; (b) for different noise levels when
$k=240$.}

\label{sketch_l1min} 
\end{figure*}

\section{Conclusions and Future Work\label{sec:Conclusions-and-Future}}

We have investigated a general covariance estimation problem under
a quadratic (rank-one) sampling model. This sampling model acts as an effective
signal processing method for real-time data with limited processing
power and memory at the data acquisition devices, and subsumes many
sampling strategies where we can only obtain magnitude or energy samples.
Three of the most popular covariance structures, i.e. sparsity, low
rank, and jointly Toeplitz and low-rank structure, have been explored
as well as sparse phase retrieval.

Our results indicate that covariance matrices under the above structural
assumptions can be perfectly recovered from a small set of quadratic
measurements and minimal storage, as long as the sensing vectors are
i.i.d. drawn from sub-Gaussian distributions. The recovery can be
achieved via efficient convex programming as soon as the memory complexity
exceeds the fundamental sampling theoretic limit. We also observe
universal recovery phenomena, in the sense that once the sensing vectors
are chosen, all covariance matrices possessing the presumed structure
can be recovered. Our results highlight the stability and robustness
of the convex program in the presence of noise and imperfect structural
assumptions. The performance guarantees for low-rank, sparse and jointly
rank-one and sparse models are established via a novel notion of a
mixed-norm restricted isometry property (RIP-$\ell_{2}/\ell_{1}$),
which significantly simplifies the proof. Our contribution also includes
a systematic approach to analyze Toeplitz low-rank structure, which
relies on RIP-$\ell_{2}/\ell_{2}$ under near-isotropic and bounded
operators.

Several future directions of interest are as follows. 
\begin{itemize}
\itemsep0.3em
\item Another covariance structure of interest is an approximately \textit{sparse
inverse covariance matrix} rather than a sparse covariance matrix.
In particular, when the signals are jointly Gaussian, the inverse
covariance matrix encodes the conditional independence, which is often
sparse. It remains to be seen whether the measurement scheme in \eqref{measurements}
can be used to recover a sparse inverse covariance matrix. 
\item It will be interesting to explore whether more general types of sampling
models satisfy RIP-$\ell_{2}/\ell_{1}$. For instance, when the sensing
vectors do not have i.i.d. entries, more delicate mathematical tricks
are necessary to establish RIP-$\ell_{2}/\ell_{1}$. 
\item In the case where RIP-$\ell_{2}/\ell_{1}$ is difficult to evaluate
(e.g. the case with random Fourier sensing vectors), it would be interesting
to develop an RIP-less theory in a similar flavor for linear measurement
models \cite{CandesPlan2011RIPless}. 
\end{itemize}
%which holds when most variables are conditionally independent given
%the rest variables; this structure finds applications in time series
%analysis and so on. 

\section*{Acknowledgments}

The authors would like to thank Emmanuel Candes for stimulating
discussions. We would also like to thank Yihong Wu for his helpful
comments and suggestions on Theorem \ref{thm:SparsePR}, Yudong
Chen for fruitful discussions about statistical consistency, and 
Xiaodong Li for helpful discussion on sparse phase retrieval. The
work of Y. Chen and A. J. Goldsmith is supported in part by the NSF
Center for Science of Information, and the AFOSR under MURI Grant
FA9550-12-1-0215. The work of Y. Chi is partially supported by NSF
CCF-1422966, the AFOSR Young Investigator Program under FA9550-15-1-0205,
and a Google Faculty Research Award.

\appendices{ }

%dummy comment inserted by tex2lyx to ensure that this paragraph is not empty
%dummy comment inserted by tex2lyx to ensure that this paragraph is not empty
%dummy comment inserted by tex2lyx to ensure that this paragraph is not empty
%dummy comment inserted by tex2lyx to ensure that this paragraph is not empty
%dummy comment inserted by tex2lyx to ensure that this paragraph is not empty
%\section{Proof of Main Results\label{sec:proof}}

%We provide an analysis framework for matrix recovery from quadratic measurements that achieves theoretic sampling limits. Our techniques allow short and self-contained proofs with very simple analysis. 

\section{Proof of Proposition \ref{corollary:RIP_lr}\label{sec:Proof-of-Lemma-RIP}}

To prove Proposition \ref{corollary:RIP_lr}, we will first derive
an upper bound and a lower bound on $\mathbb{E}\left[\left|\left\langle \boldsymbol{B}_{i},\boldsymbol{X}\right\rangle \right|\right]$,
and then apply the Bernstein-type inequality \cite[Proposition 5.16]{Vershynin2012}
to establish the large deviation bound.

In order to derive an upper bound on $\mathbb{E}\left[\left|\left\langle \boldsymbol{B}_{i},\boldsymbol{X}\right\rangle \right|\right]$,
the key step is to apply the Hanson-Wright inequality \cite{hanson1971bound,rudelson2013hanson},
which characterizes the concentration of measure for quadratic forms
in sub-Gaussian random variables. We adopt the version in \cite{rudelson2013hanson}
and repeat it below for completeness.

\begin{lemma}[\textbf{Hanson-Wright Inequality}]\label{lemma:hanson-wright}
Let $\bX=(X_{1},\ldots,X_{n})\in\mathbb{R}^{n}$ be a random vector
with independent components $X_{i}$ which satisfy $\mathbb{E}\left[X_{i}\right]=0$
and $\|X_{i}\|_{\psi_{2}}\leq K$. Let $\bA$ be an $n\times n$ matrix.
Then for any $t>0$, 
\begin{align}
 & \mathbb{P}\left\{ \left|\bX^{\top}\bA\bX-\mathbb{E}\left[\bX^{\top}\bA\bX\right]\right|>t\right\} \nonumber \\
 & \quad\leq2\exp\left[-c\min\left(\frac{t^{2}}{K^{4}\|\bA\|_{\mathrm{F}}^{2}},\frac{t}{K^{2}\|\bA\|}\right)\right].\label{hanson}
\end{align}

\end{lemma}

\begin{remark}Here, $\|\cdot\|_{\psi_{2}}$ denotes the sub-Gaussian
norm 
\[
\|X\|_{\psi_{2}}:=\min_{p\geq1}p^{-1/2}\left(\mathbb{E}\left[\left|X\right|^{p}\right]\right)^{1/p}.
\]
Similarly, the sub-exponential norm $\|\cdot\|_{\psi_{1}}$ is defined
as 
\[
\|X\|_{\psi_{1}}:=\min_{p\geq1}p^{-1}\left(\mathbb{E}\left[\left|X\right|^{p}\right]\right)^{1/p}.
\]
See \cite[Section 5.2.3 and 5.2.4]{Vershynin2012} for an introduction.\end{remark}

Observe that $\left\langle \boldsymbol{B}_{i},\boldsymbol{X}\right\rangle $
can be written as a symmetric quadratic form in $2n$ i.i.d. sub-Gaussian
random variables 
\[
\left\langle \boldsymbol{B}_{i},\boldsymbol{X}\right\rangle =\begin{bmatrix}\ba_{2i-1}^{\top} & \ba_{2i}^{\top}\end{bmatrix}\begin{bmatrix}\boldsymbol{X}\\
 & -\boldsymbol{X}
\end{bmatrix}\begin{bmatrix}\ba_{2i-1}\\
\ba_{2i}
\end{bmatrix}.
\]
The Hanson-Wright inequality \eqref{hanson} then asserts that: there
exists an absolute constant $c>0$ such that for any matrix $\boldsymbol{X}$,
$\left|\left\langle \boldsymbol{B}_{i},\boldsymbol{X}\right\rangle \right|\leq t$
with probability at least 
\[
1-2\exp\left[-c\min\left(\frac{t^{2}}{4K^{4}\left\Vert \boldsymbol{X}\right\Vert _{\text{F}}^{2}},\frac{t}{K^{2}\left\Vert \boldsymbol{X}\right\Vert }\right)\right].
\]
This indicates that $\left\langle \boldsymbol{B}_{i},\boldsymbol{X}\right\rangle $
is a sub-exponential random variable \cite[Section 5.2.4]{Vershynin2012}
satisfying 
\begin{equation}
\mathbb{E}\left[\left|\left\langle \boldsymbol{B}_{i},\boldsymbol{X}\right\rangle \right|\right]\leq c_{1}\left\Vert \boldsymbol{X}\right\Vert _{\text{F}}\label{eq:UBMean}
\end{equation}
for some positive constant $c_{1}$.

%can be bounded by 
%\begin{equation}
%\frac{1}{m}\left\Vert \left\langle \boldsymbol{B}_{i},\boldsymbol{X}\right\rangle \right\Vert _{\psi_{1}}\leq c_{1}\left\Vert \boldsymbol{X}\right\Vert _{\text{F}}.\label{eq:BoundSubExponentialNorm}
%\end{equation}
%This immediately gives an upper bound on $\mathbb{E}\left|\left\langle \boldsymbol{B}_{i},\boldsymbol{X}\right\rangle \right|$:

On the other hand, to derive a lower bound on $\mathbb{E}\left[\left|\left\langle \boldsymbol{B}_{i},\boldsymbol{X}\right\rangle \right|\right]$,
we notice that for a random variable $\xi$, repeatedly applying the
Cauchy-Schwartz inequality yields 
\[
\left(\mathbb{E}\left[\xi^{2}\right]\right)^{2}\leq\mathbb{E}\left[\left|\xi\right|\right]\mathbb{E}\left[\left|\xi\right|^{3}\right]\leq\mathbb{E}\left[\left|\xi\right|\right]\sqrt{\mathbb{E}\left[\xi^{2}\right]\mathbb{E}\left[\xi^{4}\right]},
\]
which further leads to 
\begin{equation}
\mathbb{E}\left[\left|\xi\right|\right]\geq\sqrt{\frac{\left(\mathbb{E}\left[\xi^{2}\right]\right)^{3}}{\mathbb{E}\left[\xi^{4}\right]}}.\label{eq:AbsoluteMeanLB}
\end{equation}

Let $\xi:=\left\langle \boldsymbol{B}_{i},\boldsymbol{X}\right\rangle $,
of which the second moment can be expressed as 
\[
\mathbb{E}\left[\xi^{2}\right]=\mathbb{E}\left[\left|\left\langle \boldsymbol{B}_{i},\boldsymbol{X}\right\rangle \right|^{2}\right]=\left\langle \boldsymbol{X},\mathbb{E}\left[\mathcal{B}_{i}^{*}\mathcal{B}_{i}\left(\boldsymbol{X}\right)\right]\right\rangle .
\]
Simple algebraic manipulation yields 
\[
\mathbb{E}\left[\mathcal{B}_{i}^{*}\mathcal{B}_{i}\left(\boldsymbol{X}\right)\right]=4\boldsymbol{X}+2\left(\mu_{4}-3\right)\diag(\bX),
\]
and hence 
\begin{align}
\mathbb{E}\left[\xi^{2}\right] & =4\left\Vert \boldsymbol{X}\right\Vert _{\text{F}}^{2}+2\left(\mu_{4}-3\right)\sum_{i=1}^{n}\left|\boldsymbol{X}_{ii}\right|^{2}\nonumber \\
 & \geq\min\{4,2(\mu_{4}-1)\}\left\Vert \boldsymbol{X}\right\Vert _{\text{F}}^{2}=c_{2}\left\Vert \boldsymbol{X}\right\Vert _{\text{F}}^{2},\label{eq:LBVar}
\end{align}
where $c_{2}:=\min\{4,2(\mu_{4}-1)\}$. Furthermore, since $\xi:=\left\langle \boldsymbol{B}_{i},\boldsymbol{X}\right\rangle $
has been shown to be sub-exponential with sub-exponential norm $\Theta\left(\left\Vert \boldsymbol{X}\right\Vert _{\text{F}}\right)$,
one can derive \cite{Vershynin2012} 
\begin{equation}
\mathbb{E}\left[\xi^{4}\right]=\left(4\left\Vert \xi\right\Vert _{\psi_{1}}\right)^{4}\leq c_{3}\left\Vert \boldsymbol{X}\right\Vert _{\text{F}}^{4}\label{eq:FourthMoment}
\end{equation}
for some constant $c_{7}>0$. This taken collectively with (\ref{eq:AbsoluteMeanLB})
and (\ref{eq:LBVar}) gives rise to 
\[
\mathbb{E}\left[\left|\left\langle \boldsymbol{B}_{i},\boldsymbol{X}\right\rangle \right|\right]\geq\sqrt{\frac{c_{2}^{3}\left\Vert \boldsymbol{X}\right\Vert _{\text{F}}^{6}}{c_{3}\left\Vert \boldsymbol{X}\right\Vert _{\text{F}}^{4}}}=c_{4}\left\Vert \boldsymbol{X}\right\Vert _{\text{F}}
\]
for some constant $c_{4}>0$.

Now, we are ready to characterize the concentration of $\left\langle \boldsymbol{B}_{i},\boldsymbol{X}\right\rangle $,
which is a simple consequence of the following sub-exponential variant
of Bernstein inequality. \begin{lemma}[{\cite[Proposition 5.16]{Vershynin2012}}]
\label{lemma:bernstein}Let $X_{1},\ldots,X_{m}$ be independent sub-exponential
random variables with $\mathbb{E}\left[X_{i}\right]=0$ and $K=\max_{i}\|X_{i}\|_{\psi_{1}}$.
Then for every $t>0$, we have 
\begin{equation}
\mathbb{P}\left\{ \frac{1}{m}\left|\sum_{i=1}^{m}X_{i}\right|\geq t\right\} \leq2\exp\left[-cm\min\left(\frac{t^{2}}{K^{2}},\frac{t}{K}\right)\right]
\end{equation}
where $c$ is an absolute constant. \end{lemma}

Recall that it has been shown in (\ref{eq:UBMean}) that the sub-exponential
norm of $X_{i}:=\left|\left\langle \boldsymbol{B}_{i},\boldsymbol{X}\right\rangle \right|-\mathbb{E}\left[\left|\left\langle \boldsymbol{B}_{i},\boldsymbol{X}\right\rangle \right|\right]$
satisfies $\|X_{i}\|_{\psi_{1}}\leq c^{\prime}\|\bX\|_{\mathrm{F}}$
for some universal constant $c^{\prime}$. Therefore, Lemma~\ref{lemma:bernstein}
implies that for any $\epsilon>0$, one has 
\[
\left|\frac{1}{m}\left\Vert \mathcal{B}(\bX)\right\Vert _{1}-\frac{1}{m}\mathbb{E}\left[\left\Vert \mathcal{B}\left(\boldsymbol{X}\right)\right\Vert _{1}\right]\right|\leq\epsilon\|\bX\|_{\text{F}}
\]
with probability exceeding $1-2\exp(-cm\epsilon)$ for some absolute
constant $c>0$. This yields 
\[
\frac{1}{m}\left\Vert \mathcal{B}(\bX)\right\Vert _{1}\leq\frac{1}{m}\mathbb{E}\left[\left\Vert \mathcal{B}\left(\boldsymbol{X}\right)\right\Vert _{1}\right]+\epsilon\|\bX\|_{\text{F}}\leq(c_{1}+\epsilon)\|\bX\|_{\text{F}}
\]
and 
\[
\frac{1}{m}\left\Vert \mathcal{B}(\bX)\right\Vert _{1}\geq\frac{1}{m}\mathbb{E}\left[\left\Vert \mathcal{B}\left(\boldsymbol{X}\right)\right\Vert _{1}\right]-\epsilon\|\bX\|_{\text{F}}\geq(c_{4}-\epsilon)\|\bX\|_{\text{F}}
\]
with probability at least $1-2\exp(-cm\epsilon)$, where the constants
$c$, $c_{1}$ and $c_{4}$ depend only on the sub-Gaussian norm of
$a_{i}$. Renaming the universal constants establishes Proposition~\ref{corollary:RIP_lr}.

\section{Proof of Theorem \ref{thm:RIP_Isotropic}\label{sec:Proof-of-Theorem-RIP-Isotropic}}

The proof of Theorem \ref{thm:RIP_Isotropic} follows the entropy
method introduced in \cite{rudelson2008sparse} for compressed sensing
and \cite{liu2011universal} for Pauli measurements. Note, however,
that in our case, the measurement measurements do not form a basis,
and are not even bounded. Our theorem extend the results in \cite{liu2011universal}
(which focuses on Pauli basis) to general near-isotropic measurements.

Specifically, the RIP-$\ell_{2}/\ell_{2}$ constant can be bounded
by 
\begin{align}
\delta_{r} & =\sup_{\left\Vert \boldsymbol{X}\right\Vert _{\text{F}}\leq1,\mathrm{rank}\left(\boldsymbol{X}\right)\leq r}\left|\frac{1}{m}\sum_{i=1}^{m}\left|\left\langle \boldsymbol{B}_{i},\boldsymbol{X}\right\rangle \right|^{2}-\left\Vert \boldsymbol{X}\right\Vert _{\mathrm{F}}^{2}\right|\nonumber \\
 & =\sup_{T\in\mathcal{M}_{r}^{\mathrm{t}},\boldsymbol{X}\in T,\left\Vert \boldsymbol{X}\right\Vert _{\text{F}}\leq1}\left|\left\langle \boldsymbol{X},\left(\frac{1}{m}\sum_{i=1}^{m}\mathcal{B}_{i}^{*}\mathcal{B}_{i}-\mathcal{I}\right)\boldsymbol{X}\right\rangle \right|\label{eq:Tangent}\\
 & \leq\sup_{T\in\mathcal{M}_{r}^{\mathrm{t}}}\left\Vert \mathcal{P}_{T}\left(\frac{1}{m}\sum_{i=1}^{m}\mathcal{B}_{i}^{*}\mathcal{B}_{i}-\mathcal{I}\right)\mathcal{P}_{T}\right\Vert \nonumber \\
 & \leq\sup_{T\in\mathcal{M}_{r}^{\mathrm{t}}}\left\Vert \mathcal{P}_{T}\left\{ \frac{1}{m}\sum_{i=1}^{m}\left(\mathcal{B}_{i}^{*}\mathcal{B}_{i}-\mathbb{E}\mathcal{B}_{i}^{*}\mathcal{B}_{i}\right)\right\} \mathcal{P}_{T}\right\Vert +\frac{c_{5}}{n},\label{eq:last}
\end{align}
where 
\begin{equation}
\mathcal{M}_{r}^{\mathrm{t}}:=\left\{ \text{tangent space w.r.t. }\boldsymbol{M}\mid\forall\boldsymbol{M}:\text{ }\text{rank}\left(\boldsymbol{M}\right)\leq r\right\} .
\end{equation}
and hence (\ref{eq:Tangent}) arises since the supremum is taken over
all tangent space $T$ associated with rank-$r$ matrices. The last
inequality (\ref{eq:last}) follows from the near-isotropic assumption
of $\mathcal{B}_{i}$ (i.e. (\ref{eq:B_bounded_RIP})).

The first step is to prove that $\mathbb{E}\left[\delta_{r}\right]\leq\epsilon$
for some small constant $\epsilon>0$. For sufficiently large $n$,
it suffices to prove that 
\begin{equation}
E:=\mathbb{E}\left[\sup_{T\in\mathcal{M}_{r}^{\mathrm{t}}}\left\Vert \mathcal{P}_{T}\left\{ \frac{1}{m}\sum_{i=1}^{m}\left(\mathcal{B}_{i}^{*}\mathcal{B}_{i}-\mathbb{E}\left[\mathcal{B}_{i}^{*}\mathcal{B}_{i}\right]\right)\right\} \mathcal{P}_{T}\right\Vert \right]\leq\delta.\label{eq:DefnE}
\end{equation}
This can be established by a Gaussian process approach as follows.

Observe that $\frac{1}{m}\sum_{i=1}^{m}\left(\mathcal{B}_{i}^{*}\mathcal{B}_{i}-\mathbb{E}\left[\mathcal{B}_{i}^{*}\mathcal{B}_{i}\right]\right)$
is a zero-mean operator, which can be reduced to symmetric operators
via the symmetrization argument (see, e.g. \cite{Tao2012RMT}). Specifically,
let $\tilde{\mathcal{B}}_{i}$ be an independent copy of $\mathcal{B}_{i}$.
Conditioning on $\mathcal{B}_{i}$ we have 
\begin{align*}
 & \mathbb{E}\left[\left.\frac{1}{m}\sum_{i=1}^{m}\mathcal{B}_{i}^{*}\mathcal{B}_{i}-\frac{1}{m}\sum_{i=1}^{m}\tilde{\mathcal{B}}_{i}^{*}\tilde{\mathcal{B}}_{i}\right|\mathcal{B}_{i}\text{ }(1\leq i\leq m)\right]\\
 & \quad=\frac{1}{m}\sum_{i=1}^{m}\mathcal{B}_{i}^{*}\mathcal{B}_{i}-\mathbb{E}\left[\mathcal{B}_{i}^{*}\mathcal{B}_{i}\right].
\end{align*}
Since the function $f\left(\mathcal{X}\right):=\sup_{T\in\mathcal{M}_{r}^{\mathrm{t}}}\left\Vert \mathcal{P}_{T}\mathcal{X}\mathcal{P}_{T}\right\Vert $
is convex in $\mathcal{X}$, applying Jensen's inequality yields 
\begin{align*}
 & \sup_{T\in\mathcal{M}_{r}^{\mathrm{t}}}\left\Vert \mathcal{P}_{T}\left\{ \frac{1}{m}\sum_{i=1}^{m}\left(\mathcal{B}_{i}^{*}\mathcal{B}_{i}-\mathbb{E}\left[\mathcal{B}_{i}^{*}\mathcal{B}_{i}\right]\right)\right\} \mathcal{P}_{T}\right\Vert \\
 & \text{ }\text{ }\footnotesize=\sup_{T\in\mathcal{M}_{r}^{\mathrm{t}}}\left\Vert \mathbb{E}\left[\left.\mathcal{P}_{T}\left(\frac{1}{m}\sum_{i=1}^{m}\mathcal{B}_{i}^{*}\mathcal{B}_{i}-\frac{1}{m}\sum_{i=1}^{m}\tilde{\mathcal{B}}_{i}^{*}\tilde{\mathcal{B}}_{i}\right)\mathcal{P}_{T}\right|\mathcal{B}_{i}\right]\right\Vert \\
 & \text{ }\text{ }\footnotesize\leq\text{ }\mathbb{E}\left[\left.\sup_{T\in\mathcal{M}_{r}^{\mathrm{t}}}\left\Vert \mathcal{P}_{T}\left(\frac{1}{m}\sum_{i=1}^{m}\mathcal{B}_{i}^{*}\mathcal{B}_{i}-\frac{1}{m}\sum_{i=1}^{m}\tilde{\mathcal{B}}_{i}^{*}\tilde{\mathcal{B}}_{i}\right)\mathcal{P}_{T}\right\Vert \right|\mathcal{B}_{i}\right].
\end{align*}
Undoing conditioning over $\mathcal{B}_{i}$ we get 
\begin{align}
 & \mathbb{E}\left[\sup_{T\in\mathcal{M}_{r}^{\mathrm{t}}}\left\Vert \mathcal{P}_{T}\left\{ \frac{1}{m}\sum_{i=1}^{m}\left(\mathcal{B}_{i}^{*}\mathcal{B}_{i}-\mathbb{E}\left[\mathcal{B}_{i}^{*}\mathcal{B}_{i}\right]\right)\right\} \mathcal{P}_{T}\right\Vert \right]\nonumber \\
 & \quad\leq\text{ }\mathbb{E}\left[\sup_{T\in\mathcal{M}_{r}^{\mathrm{t}}}\left\Vert \mathcal{P}_{T}\left(\frac{1}{m}\sum_{i=1}^{m}\mathcal{B}_{i}^{*}\mathcal{B}_{i}-\frac{1}{m}\sum_{i=1}^{m}\tilde{\mathcal{B}}_{i}^{*}\tilde{\mathcal{B}}_{i}\right)\mathcal{P}_{T}\right\Vert \right]\nonumber \\
 & \quad\leq\text{ }2\mathbb{E}\left[\sup_{T\in\mathcal{M}_{r}^{\mathrm{t}}}\left\Vert \frac{1}{m}\sum_{i=1}^{m}\epsilon_{i}\mathcal{P}_{T}\mathcal{B}_{i}^{*}\mathcal{B}_{i}\mathcal{P}_{T}\right\Vert \right],\label{eq:E_ub}
\end{align}
where $\epsilon_{i}$'s are i.i.d. symmetric Bernoulli random variables.
Moreover, if we generate a set of i.i.d. random variables $g_{i}\sim\mathcal{N}\left(0,1\right)$,
then the conditional expectation obeys 
\begin{align*}
 & \mathbb{E}\left[\left.\frac{1}{m}\sum_{i=1}^{m}\left|g_{i}\right|\epsilon_{i}\mathcal{P}_{T}\mathcal{B}_{i}^{*}\mathcal{B}_{i}\mathcal{P}_{T}\right|\epsilon_{i},\mathcal{B}_{i}\right]\\
 & \quad=\sqrt{\frac{2}{\pi}}\frac{1}{m}\sum_{i=1}^{m}\epsilon_{i}\mathcal{P}_{T}\mathcal{B}_{i}^{*}\mathcal{B}_{i}\mathcal{P}_{T}.
\end{align*}
Similarly, by convexity of $f\left(\mathcal{X}\right):=\sup_{T\in\mathcal{M}_{r}^{\mathrm{t}}}\left\Vert \mathcal{P}_{T}\mathcal{X}\mathcal{P}_{T}\right\Vert $,
one can obtain 
\begin{align}
 & \mathbb{E}\left[\sup_{T\in\mathcal{M}_{r}^{\mathrm{t}}}\left\Vert \frac{1}{m}\sum_{i=1}^{m}\epsilon_{i}\mathcal{P}_{T}\mathcal{B}_{i}^{*}\mathcal{B}_{i}\mathcal{P}_{T}\right\Vert \right]\nonumber \\
 & \text{ }\text{ }=\sqrt{\frac{\pi}{2}}\mathbb{E}\left[\sup_{T\in\mathcal{M}_{r}^{\mathrm{t}}}\left\Vert \mathbb{E}\left[\left.\frac{1}{m}\sum_{i=1}^{m}\left|g_{i}\right|\epsilon_{i}\mathcal{P}_{T}\mathcal{B}_{i}^{*}\mathcal{B}_{i}\mathcal{P}_{T}\right|\epsilon_{i},\mathcal{B}_{i}\right]\right\Vert \right]\nonumber \\
 & \text{ }\text{ }\leq\sqrt{\frac{\pi}{2}}\mathbb{E}\left[\sup_{T\in\mathcal{M}_{r}^{\mathrm{t}}}\left\Vert \frac{1}{m}\sum_{i=1}^{m}g_{i}\mathcal{P}_{T}\mathcal{B}_{i}^{*}\mathcal{B}_{i}\mathcal{P}_{T}\right\Vert \right].\label{eq:E_ub_Gaussian}
\end{align}
Putting (\ref{eq:E_ub}) and (\ref{eq:E_ub_Gaussian}) together we
obtain 
\begin{align}
 & \mathbb{E}\left[\sup_{T\in\mathcal{M}_{r}^{\mathrm{t}}}\left\Vert \mathcal{P}_{T}\left\{ \frac{1}{m}\sum_{i=1}^{m}\left(\mathcal{B}_{i}^{*}\mathcal{B}_{i}-\mathbb{E}\left[\mathcal{B}_{i}^{*}\mathcal{B}_{i}\right]\right)\right\} \mathcal{P}_{T}\right\Vert \right]\nonumber \\
 & \quad\leq\sqrt{2\pi}\mathbb{E}\left[\sup_{T\in\mathcal{M}_{r}^{\mathrm{t}}}\left\Vert \frac{1}{m}\sum_{i=1}^{m}g_{i}\mathcal{P}_{T}\mathcal{B}_{i}^{*}\mathcal{B}_{i}\mathcal{P}_{T}\right\Vert \right]\nonumber \\
 & \quad=\sqrt{2\pi}\mathbb{E}\left[\sup_{T\in\mathcal{M}_{r}^{\mathrm{t}},\boldsymbol{X}\in T,\left\Vert \boldsymbol{X}\right\Vert _{\text{F}}=1}\left|\frac{1}{m}\sum_{i=1}^{m}g_{i}\left|\mathcal{B}_{i}\left(\boldsymbol{X}\right)\right|^{2}\right|\right].\label{eq:E1Symmetrization}
\end{align}
It then boils down to characterizing the supremum of a Gaussian process.

We now prove the following lemma.

\begin{lemma}\label{lemma:EstimatePAAP}Suppose that $g_{i}\sim\mathcal{N}(0,1)$
are i.i.d. random variables, and that $K\leq n^{2}$. Conditional
on $\mathcal{B}_{i}$'s, we have 
\begin{align}
 & \mathbb{E}\left[\left.\sup_{T\in\mathcal{M}_{r}^{\mathrm{t}}}\Bigg\|\mathcal{P}_{T}\left(\sum_{i=1}^{m}g_{i}\mathcal{B}_{i}^{*}\mathcal{B}_{i}\right)\mathcal{P}_{T}\Bigg\|\text{ }\right|\mathcal{B}_{i}\text{ }(1\leq i\leq m)\right]\nonumber \\
 & \quad\leq C_{14}\sqrt{r}K\log^{3}n\sup_{T:T\in\mathcal{M}_{r}^{\mathrm{t}}}\sqrt{\left\Vert \sum_{i=1}^{m}\mathcal{P}_{T}\mathcal{B}_{i}^{*}\mathcal{B}_{i}\mathcal{P}_{T}\right\Vert }.\label{eq:EstimatePAAP-lem}
\end{align}
\end{lemma}

\begin{proof}See Appendix \ref{sec:Proof-of-Lemma-EstimatePAAP}.\end{proof}

Combining Lemma \ref{lemma:EstimatePAAP} with (\ref{eq:E1Symmetrization})
and undoing the conditioning on $\mathcal{B}_{i}$'s yield 
\begin{align*}
 & \mathbb{E}\left[\sup_{T\in\mathcal{M}_{r}^{\mathrm{t}}}\left\Vert \mathcal{P}_{T}\left\{ \frac{1}{m}\sum_{i=1}^{m}\left(\mathcal{B}_{i}^{*}\mathcal{B}_{i}-\mathbb{E}\mathcal{B}_{i}^{*}\mathcal{B}_{i}\right)\right\} \mathcal{P}_{T}\right\Vert \right]\\
 & \quad\leq\text{ }\frac{C_{15}\sqrt{r}K\log^{3}n}{m}\cdot\mathbb{E}\left[\sqrt{\sup_{T:T\in\mathcal{M}_{r}^{\mathrm{t}}}\left\Vert \sum_{i=1}^{m}\mathcal{P}_{T}\mathcal{B}_{i}^{*}\mathcal{B}_{i}\mathcal{P}_{T}\right\Vert }\right]\\
 & \quad\leq\frac{C_{15}\sqrt{r}K\log^{3}n}{\sqrt{m}}\sqrt{\mathbb{E}\left[\sup_{T:T\in\mathcal{M}_{r}^{\mathrm{t}}}\left\Vert \frac{1}{m}\sum_{i=1}^{m}\mathcal{P}_{T}\mathcal{B}_{i}^{*}\mathcal{B}_{i}\mathcal{P}_{T}\right\Vert \right]}
\end{align*}
for some universal constant $C_{15}>0$, where the last inequality
follows from Jensen's inequality. Recall the definition of $E$ in
(\ref{eq:DefnE}), then the above inequality implies 
\[
E\leq C_{15}\left(\frac{\sqrt{r}K\log^{3}n}{\sqrt{m}}\right)\sqrt{E+1},
\]
or more concretely, 
\begin{equation}
\mathbb{E}\left[\delta_{r}\right]\leq E\leq2C_{15}\frac{\sqrt{r}K\log^{3}n}{\sqrt{m}}<1\label{eq:BoundE}
\end{equation}
as soon as $m>\left(2C_{15}\sqrt{r}K\log^{3}n\right)^{2}$.

Now that we have established that $\mathbb{E}\left[\delta_{r}\right]$
can be a small constant if $m>4C_{15}^{2}rK^{2}\log^{6}n$, it remains
to show that $\delta_{r}$ sharply concentrates around $\mathbb{E}\sigma_{r}$.
To this end, consider the Banach space $\Upsilon$ of operators $\mathcal{H}:\mathbb{R}^{n\times n}\mapsto\mathbb{R}^{n\times n}$
equipped with the norm 
\[
\left\Vert \mathcal{H}\right\Vert _{\Upsilon}:=\sup_{T\in\mathcal{M}_{r}^{\mathrm{t}}}\left\Vert \mathcal{P}_{T}\mathcal{H}\mathcal{P}_{T}\right\Vert .
\]
Let $\varepsilon_{i}$'s be i.i.d. symmetric Bernoulli variables,
then the symmetrization trick (see, e.g. \cite{Tao2012RMT}) yields
\begin{align*}
 & \mathbb{E}\left[\left\Vert \frac{1}{m}\sum_{i=1}^{m}\mathcal{B}_{i}^{*}\mathcal{B}_{i}-\mathbb{E}\left[\mathcal{B}_{i}^{*}\mathcal{B}_{i}\right]\right\Vert _{\Upsilon}\right]\\
 & \quad\leq\mathbb{E}\left[\left\Vert \frac{1}{m}\sum_{i=1}^{m}\varepsilon_{i}\mathcal{B}_{i}^{*}\mathcal{B}_{i}\right\Vert _{\Upsilon}\right]\\
 & \quad\leq2\mathbb{E}\left\Vert \frac{1}{m}\sum_{i=1}^{m}\mathcal{B}_{i}^{*}\mathcal{B}_{i}-\mathbb{E}\left[\mathcal{B}_{i}^{*}\mathcal{B}_{i}\right]\right\Vert _{\Upsilon},
\end{align*}
and 
\begin{align*}
 & \mathbb{P}\left\{ \left\Vert \frac{1}{m}\sum_{i=1}^{m}\mathcal{B}_{i}^{*}\mathcal{B}_{i}-\mathbb{E}\left[\mathcal{B}_{i}^{*}\mathcal{B}_{i}\right]\right\Vert _{\Upsilon}>\right.\\
 & \quad\quad\quad\quad\left.2\mathbb{E}\left[\left\Vert \frac{1}{m}\sum_{i=1}^{m}\mathcal{B}_{i}^{*}\mathcal{B}_{i}-\mathbb{E}\left[\mathcal{B}_{i}^{*}\mathcal{B}_{i}\right]\right\Vert _{\Upsilon}\right]+u\right\} \\
 & \quad\leq\mathbb{P}\left\{ \left\Vert \frac{1}{m}\sum_{i=1}^{m}\left(\mathcal{B}_{i}^{*}\mathcal{B}_{i}-\tilde{\mathcal{B}}_{i}^{*}\tilde{\mathcal{B}}_{i}\right)\right\Vert _{\Upsilon}>u\right\} \\
 & \quad\leq2\mathbb{P}\left\{ \left\Vert \frac{1}{m}\sum_{i=1}^{m}\varepsilon_{i}\mathcal{B}_{i}^{*}\mathcal{B}_{i}\right\Vert _{\Upsilon}>\frac{u}{2}\right\} ,
\end{align*}
where $\tilde{\mathcal{B}}_{i}$ is an independent copy of $\mathcal{B}_{i}$.
Note that $\varepsilon_{i}\mathcal{B}_{i}^{*}\mathcal{B}_{i}$'s are
i.i.d. zero-mean random operators.

In addition, for any $1\leq i\leq m$, we know that 
\begin{align*}
\left\Vert \epsilon_{i}\mathcal{B}_{i}^{*}\mathcal{B}_{i}\right\Vert _{\Upsilon} & =\max_{T\in\mathcal{M}_{r}^{\mathrm{t}}}\left\Vert \mathcal{P}_{T}\mathcal{B}_{i}^{*}\mathcal{B}_{i}\mathcal{P}_{T}\right\Vert \\
 & =\max_{T\in\mathcal{M}_{r}^{\mathrm{t}},\left\Vert \boldsymbol{X}\right\Vert _{\text{F}}=1}\left|\left\langle \boldsymbol{B}_{i},\mathcal{P}_{T}\left(\boldsymbol{X}\right)\right\rangle \right|^{2}\\
 & \leq\max_{T\in\mathcal{M}_{r}^{\mathrm{t}},\left\Vert \boldsymbol{X}\right\Vert _{\text{F}}=1}\left\Vert \boldsymbol{B}_{i}\right\Vert ^{2}\left\Vert \mathcal{P}_{T}\left(\boldsymbol{X}\right)\right\Vert _{*}^{2}\leq K^{2}r.
\end{align*}
Theorem 3.10 of \cite{rudelson2008sparse} asserts that there is a
universal constant $C_{12}>0$ such that 
\begin{align*}
 & \footnotesize\mathbb{P}\left\{ \left\Vert \frac{1}{m}\sum_{i=1}^{m}\epsilon_{i}\mathcal{B}_{i}^{*}\mathcal{B}_{i}\right\Vert _{\Upsilon}>8q\mathbb{E}\left\Vert \frac{1}{m}\sum_{i=1}^{m}\epsilon_{i}\mathcal{B}_{i}^{*}\mathcal{B}_{i}\right\Vert _{\Upsilon}+\frac{2K^{2}r}{m}l+t\right\} \\
 & \quad\leq\footnotesize\left(\frac{C_{12}}{q}\right)^{l}+2\exp\left(-\frac{t^{2}}{256q\left(\mathbb{E}\left[\left\Vert \frac{1}{m}\sum_{i=1}^{m}\epsilon_{i}\mathcal{B}_{i}^{*}\mathcal{B}_{i}\right\Vert _{\Upsilon}\right]\right)^{2}}\right).
\end{align*}
If we take $q=2C_{12}$, $l=C_{13}\log n$ and $t=C_{14}\sqrt{\log n}\mathbb{E}\left[\left\Vert \frac{1}{m}\sum_{i=1}^{m}\epsilon_{i}\mathcal{B}_{i}^{*}\mathcal{B}_{i}\right\Vert _{\Upsilon}\right]$,
then for sufficiently large $C_{13}$ and $C_{14}$, there exists
an absolute constant $C_{20}>0$ such that if $m>C_{20}rK^{2}\log^{7}n$,
then for any small positive constant $\delta$ we have 
\begin{align*}
\left\Vert \frac{1}{m}\sum_{i=1}^{m}\epsilon_{i}\mathcal{B}_{i}^{*}\mathcal{B}_{i}\right\Vert _{\Upsilon} & <C_{15}\sqrt{\log n}\mathbb{E}\left[\left\Vert \frac{1}{m}\sum_{i=1}^{m}\epsilon_{i}\mathcal{B}_{i}^{*}\mathcal{B}_{i}\right\Vert _{\Upsilon}\right]\\
 & <\delta
\end{align*}
with probability exceeding $1-n^{-2}$.

Now that we have ensured a small RIP-$\ell_{2}/\ell_{2}$ constant,
repeating the argument as in \cite{CandesPlan2011Tight,RecFazPar07}
implies 
\begin{equation}
\Vert\hat{\boldsymbol{\Sigma}}-\boldsymbol{\Sigma}\Vert_{\mathrm{F}}\leq C_{2}\frac{\epsilon_{2}}{\sqrt{m}}
\end{equation}
for all $\boldsymbol{\Sigma}$ of rank at most $r$. This concludes
the proof.

\section{Proof of Lemma \ref{lemma:duality}\label{proof:duality}}

We first introduce a few mathematical notations before proceeding
to the proof. Let the singular value decomposition of a rank-$r$
matrix $\bSigma$ be $\bSigma=\boldsymbol{U}{\bf \Lambda}\boldsymbol{V}^{\top}$,
then the tangent space $T$ at the point $\bSigma$ is defined as
$T:=\left\{ \boldsymbol{U}\boldsymbol{M}_{1}+\boldsymbol{M}_{2}\boldsymbol{V}^{\top}\mid\boldsymbol{M}_{1}\in\mathbb{R}^{r\times n},\boldsymbol{M}_{2}\in\mathbb{R}^{n\times r}\right\} $.
We denote by $\mathcal{P}_{T}$ and $\mathcal{P}_{T^{\perp}}$ the
orthogonal projection onto $T$ and its orthogonal complement, respectively.
For notational simplicity, we denote $\boldsymbol{H}_{T}:=\mathcal{P}_{T}\left(\boldsymbol{H}\right)$
and $\boldsymbol{H}_{T^{\perp}}:=\boldsymbol{H}-\mathcal{P}_{T}\left(\boldsymbol{H}\right)$
for any matrices $\boldsymbol{H}\in\mathbb{R}^{n\times n}$. The proof
is inspired by the techniques introduced for operators satisfying
RIP-$\ell_{2}/\ell_{2}$ \cite{RecFazPar07,CandesPlan2011Tight}.

Write $\bSigma:=\bSigma_{r}+\bSigma_{\mathrm{c}}$, where $\bSigma_{r}$
represents the best rank-$r$ approximation of $\bSigma$. Denote
by $T$ the tangent space with respect to $\bSigma_{r}$. Suppose
that the solution to \eqref{tracemin} is given by $\hat{\bSigma}=\bSigma+\boldsymbol{H}$
for some matrix $\boldsymbol{H}$. The optimality of $\hat{\bSigma}$
yields 
\begin{align*}
0 & \geq\left\Vert \bSigma+\boldsymbol{H}\right\Vert _{*}-\left\Vert \bSigma\right\Vert _{*}\\
 & \geq\left\Vert \bSigma_{r}+\boldsymbol{H}\right\Vert _{*}-\left\Vert \bSigma_{c}\right\Vert _{*}-\left\Vert \bSigma\right\Vert _{*}\\
 & \geq\left\Vert \bSigma_{r}+\boldsymbol{H}_{T^{\perp}}\right\Vert _{*}-\left\Vert \boldsymbol{H}_{T}\right\Vert _{*}-\left\Vert \bSigma_{r}\right\Vert _{*}-2\left\Vert \bSigma_{\mathrm{c}}\right\Vert _{*}\\
 & =\left\Vert \bSigma_{r}\right\Vert _{*}+\left\Vert \boldsymbol{H}_{T^{\perp}}\right\Vert _{*}-\left\Vert \boldsymbol{H}_{T}\right\Vert _{*}-\left\Vert \bSigma_{r}\right\Vert _{*}-2\left\Vert \bSigma_{\mathrm{c}}\right\Vert _{*},
\end{align*}
which leads to 
\begin{equation}
\left\Vert \boldsymbol{H}_{T^{\perp}}\right\Vert _{*}\leq\left\Vert \boldsymbol{H}_{T}\right\Vert _{*}+2\left\Vert \bSigma_{\mathrm{c}}\right\Vert _{*}.\label{eq:HtHtperp-Approx}
\end{equation}

We then divide $\boldsymbol{H}_{T^{\perp}}$ into $M=\left\lceil \frac{n-r}{K_{1}}\right\rceil $
orthogonal matrices $\boldsymbol{H}_{1}$, $\boldsymbol{H}_{2}$,
$\cdots$, $\boldsymbol{H}_{M}$ satisfying the following: (i) the
largest singular value of $\boldsymbol{H}_{i+1}$ does not exceed
the smallest non-zero singular value of $\boldsymbol{H}_{i}$, and
(ii) 
\begin{equation}
\left\Vert \boldsymbol{H}_{T^{\perp}}\right\Vert _{*}=\sum_{i=1}^{M}\left\Vert \boldsymbol{H}_{i}\right\Vert _{*}\label{eq:RequirementH}
\end{equation}
and $\text{rank}\left(\boldsymbol{H}_{i}\right)=K_{1}$ for $1\leq i\leq M-1$.
Along with the bound \eqref{eq:HtHtperp-Approx}, this yields that
\begin{align}
\sum_{i\geq2}\left\Vert \boldsymbol{H}_{i}\right\Vert _{\text{F}} & \leq\frac{1}{\sqrt{K_{1}}}\sum_{i\geq2}\left\Vert \boldsymbol{H}_{i-1}\right\Vert _{*}\leq\frac{1}{\sqrt{K_{1}}}\left\Vert \boldsymbol{H}_{T^{\perp}}\right\Vert _{*}\nonumber \\
 & \leq\frac{1}{\sqrt{K_{1}}}\left(\left\Vert \boldsymbol{H}_{T}\right\Vert _{*}+2\left\Vert \bSigma_{\mathrm{c}}\right\Vert _{*}\right).\label{eq:HtperpUBApprox}
\end{align}
Since the feasibility constraint requires $\left\Vert \mathcal{A}\left(\bSigma\right)-\boldsymbol{y}\right\Vert _{1}\leq\epsilon_{1}$,
we have $\|\mathcal{A}(\bH)\|_{1}\leq\left\Vert \mathcal{A}\left(\bSigma\right)-\boldsymbol{y}\right\Vert _{1}+\left\Vert \mathcal{A}\left(\hat{\bSigma}\right)-\boldsymbol{y}\right\Vert _{1}\leq2\epsilon_{1}$,
and then following from the definition $\boldsymbol{B}_{i}=\boldsymbol{A}_{2i-1}-\boldsymbol{A}_{2i}$
that 
\[
\frac{1}{m}\left\Vert \mathcal{B}(\bH)\right\Vert _{1}\leq\frac{1}{m}\left\Vert \mathcal{A}(\bH)\right\Vert _{1}\leq\frac{2\epsilon_{1}}{m},
\]
yielding 
\begin{align*}
\frac{2\epsilon_{1}}{m} & \geq\frac{1}{m}\left\Vert \mathcal{B}\left(\boldsymbol{H}\right)\right\Vert _{1}\\
 & \geq\frac{1}{m}\left\Vert \mathcal{B}\left(\boldsymbol{H}_{T}+\boldsymbol{H}_{1}\right)\right\Vert _{1}-\sum_{i\geq2}\frac{1}{m}\left\Vert \mathcal{B}\left(\boldsymbol{H}_{i}\right)\right\Vert _{1}\\
 & \geq\left(1-\delta_{2r+K_{1}}^{\text{lb}}\right)\left\Vert \boldsymbol{H}_{T}+\boldsymbol{H}_{1}\right\Vert _{\text{F}}-\left(1+\delta_{K_{1}}^{\text{ub}}\right)\sum_{i\geq2}\left\Vert \boldsymbol{H}_{i}\right\Vert _{\text{F}}\\
 & \geq\frac{(1-\delta_{2r+K_{1}}^{\text{lb}})}{\sqrt{2}}\left(\left\Vert \boldsymbol{H}_{T}\right\Vert _{\text{F}}+\left\Vert \boldsymbol{H}_{1}\right\Vert _{\text{F}}\right)\\
 & \quad\quad\quad-\frac{\left(1+\delta_{K_{1}}^{\text{ub}}\right)}{\sqrt{K_{1}}}\left(\left\Vert \boldsymbol{H}_{T}\right\Vert _{*}+2\left\Vert \boldsymbol{\Sigma}_{\mathrm{c}}\right\Vert _{*}\right).
\end{align*}
By reorganizing the terms and using the fact that $\|\bH_{T}\|_{*}\leq\sqrt{2r}\|\bH_{T}\|_{\text{F}}$,
one can derive 
\begin{align}
 & \footnotesize\left[\frac{(1-\delta_{2r+K_{1}}^{\text{lb}})}{\sqrt{2}}-\frac{\left(1+\delta_{K_{1}}^{\text{ub}}\right)\sqrt{2r}}{\sqrt{K_{1}}}\right]\left\Vert \boldsymbol{H}_{T}\right\Vert _{\text{F}}+\frac{(1-\delta_{2r+K_{1}}^{\text{lb}})}{\sqrt{2}}\left\Vert \boldsymbol{H}_{1}\right\Vert _{\text{F}}\nonumber \\
 & \quad\leq\frac{2\left(1+\delta_{K_{1}}^{\text{ub}}\right)}{\sqrt{K_{1}}}\left\Vert \bSigma_{\mathrm{c}}\right\Vert _{*}+\frac{2\epsilon_{1}}{m}.\label{eq:ApproxBound}
\end{align}

The bound (\ref{eq:ApproxBound}) allows us to see that if $\frac{1-\delta_{2r+K_{1}}^{\text{lb}}}{\sqrt{2}}-\left(1+\delta_{K_{1}}^{\text{ub}}\right)\sqrt{\frac{2r}{K_{1}}}\geq\beta_{1}>0$
for some absolute constant $\beta_{1}$, then one has 
\begin{equation}
\left\Vert \boldsymbol{H}_{T}\right\Vert _{\text{F}}+\left\Vert \boldsymbol{H}_{1}\right\Vert _{\text{F}}\leq\frac{2}{\beta_{1}}\left(\frac{\left(1+\delta_{K_{1}}^{\text{ub}}\right)}{\sqrt{K_{1}}}\left\Vert \bSigma_{\mathrm{c}}\right\Vert _{*}+\frac{\epsilon_{1}}{m}\right).\label{eq:HtH1Bound}
\end{equation}
On the other hand, \eqref{eq:HtperpUBApprox} allows us to bound 
\begin{align}
\sum_{i\geq2}\left\Vert \boldsymbol{H}_{i}\right\Vert _{\text{F}} & \leq\frac{1}{\sqrt{K_{1}}}\left(\left\Vert \boldsymbol{H}_{T}\right\Vert _{*}+2\left\Vert \bSigma_{\mathrm{c}}\right\Vert _{*}\right)\nonumber \\
 & \leq\sqrt{\frac{2r}{K_{1}}}\left\Vert \boldsymbol{H}_{T}\right\Vert _{\text{F}}+\frac{2}{\sqrt{K_{1}}}\left\Vert \bSigma_{\mathrm{c}}\right\Vert _{*}.\label{eq:HiBoundRest}
\end{align}
Putting the above computation together establishes 
\begin{align*}
\left\Vert \boldsymbol{H}\right\Vert _{\text{F}} & \leq\left\Vert \boldsymbol{H}_{T}\right\Vert _{\text{F}}+\left\Vert \boldsymbol{H}_{1}\right\Vert _{\text{F}}+\sum_{i\geq2}\left\Vert \boldsymbol{H}_{i}\right\Vert _{\text{F}}\\
 & \leq\left(1+\sqrt{\frac{2r}{K_{1}}}\right)\left(\left\Vert \boldsymbol{H}_{T}\right\Vert _{\text{F}}+\left\Vert \boldsymbol{H}_{1}\right\Vert _{\text{F}}\right)+\frac{2}{\sqrt{K_{1}}}\left\Vert \bSigma_{\mathrm{c}}\right\Vert _{*}\\
 & \leq\left(\frac{C_{1}}{\beta_{1}}+C_{3}\right)\frac{\left\Vert \bSigma_{\mathrm{c}}\right\Vert _{*}}{\sqrt{K_{1}}}+\frac{C_{2}}{\beta_{1}}\cdot\frac{\epsilon_{1}}{m}
\end{align*}
for some positive universal constants $C_{1}$, $C_{2}$ and $C_{3}$.

\section{Proof of Lemma~\ref{lemma:duality-sparse}}

\label{proof:duality_sparse} For an index set $\Omega$, let $\mathcal{P}_{\Omega}$
as the orthogonal projection onto the index set $\Omega$. We denote
$\bH_{\Omega}$ as the matrix supported on $\bH_{\Omega}=\mathcal{P}_{\Omega}(\bH)$
and $\bH_{\Omega^{\perp}}$ as the projection onto the complement
support set $\Omega^{\perp}$. Write $\hat{\boldsymbol{\Sigma}}=\boldsymbol{\Sigma}+\bH$,
and $\boldsymbol{\Sigma}=\boldsymbol{\Sigma}_{\Omega_{0}}+\boldsymbol{\Sigma}_{\Omega_{0}^{c}}$,
where $\Omega_{0}$ denotes the support of the $k$ largest entries
of $\bSigma$. The feasibility constraint yields 
\[
\frac{1}{m}\left\Vert \mathcal{B}(\bH)\right\Vert _{1}\leq\frac{2}{m}\left\Vert \mathcal{A}(\hat{\boldsymbol{\Sigma}})-\mathcal{A}(\boldsymbol{\Sigma})\right\Vert _{1}\leq\frac{2\epsilon_{1}}{m}.
\]
The triangle inequality of $\ell_{1}$ norm gives 
\begin{align*}
\|\hat{\boldsymbol{\Sigma}}-\boldsymbol{\Sigma}\|_{1} & \leq\|\hat{\boldsymbol{\Sigma}}-\boldsymbol{\Sigma}_{\Omega_{0}}\|_{1}+\|\boldsymbol{\Sigma}_{\Omega_{0}^{\text{c}}}\|_{1}
\end{align*}

Decompose $\bH_{\Omega_{0}^{c}}$ into a collection of $M_{2}$ matrices
$\bH_{\Omega_{1}}$, $\bH_{\Omega_{2}}$, $\ldots$, $\boldsymbol{H}_{\Omega_{M_{2}}}$,
where $\left\Vert \bH_{\Omega_{i}}\right\Vert _{0}=K_{2}$ for all
$1\leq i<M_{2}$, $\bH_{\Omega_{1}}$ consists of the $K_{2}$ largest
entries of $\bH_{\Omega_{0}^{c}}$, $\bH_{\Omega_{2}}$ consists of
the $K_{2}$ largest entries of $\bH_{(\Omega_{0}\cup\Omega_{1})^{c}}$,
and so on. A similar argument as in \cite{candes2008restricted} implies
\begin{equation}
\sum_{i\geq2}\|\bH_{\Omega_{i}}\|_{\text{F}}\leq\frac{1}{\sqrt{K_{2}}}\sum_{i\geq1}\|\bH_{\Omega_{i}}\|_{1}=\frac{1}{\sqrt{K_{2}}}\|\bH_{\Omega_{0}^{c}}\|_{1}.\label{candes}
\end{equation}
The optimality of $\hat{\bSigma}$ yields 
\begin{align*}
\|\bSigma\|_{1} & \geq\|\bSigma+\boldsymbol{H}\|_{1}=\|\bSigma_{\Omega_{0}}+\boldsymbol{H}\|_{1}-\|\bSigma_{\Omega_{0}^{\text{c}}}\|_{1}\\
 & \geq\|\bSigma_{\Omega_{0}}\|_{1}+\|\bH_{\Omega_{0}^{c}}\|_{1}-\|\bH_{\Omega_{0}}\|_{1}-\|{\bSigma}_{\Omega^{c}}\|_{1},
\end{align*}
which gives 
\[
\|\bH_{\Omega_{0}^{c}}\|_{1}\leq\|\bH_{\Omega_{0}}\|_{1}+2\|{\bSigma}_{\Omega^{c}}\|_{1}.
\]
Combining the above bound and \eqref{candes} leads to 
\begin{align}
\sum_{i\geq2}\|\bH_{\Omega_{i}}\|_{\text{F}} & \leq\frac{1}{\sqrt{K_{2}}}(\|\bH_{\Omega_{0}}\|_{1}+2\|{\bSigma}_{\Omega^{c}}\|_{1})\nonumber \\
 & \leq\frac{1}{\sqrt{K_{2}}}\left(\sqrt{k}\|\bH_{\Omega_{0}}\|_{\text{F}}+2\|{\bSigma}_{\Omega^{c}}\|_{1}\right).\label{residual_norm}
\end{align}
It then follows that 
\begin{align*}
\frac{2\epsilon_{1}}{m}\text{ } & \geq\frac{1}{m}\left\Vert \mathcal{B}\left(\boldsymbol{H}\right)\right\Vert _{1}\\
 & \geq\frac{1}{m}\left\Vert \mathcal{B}\left(\boldsymbol{H}_{\Omega_{0}}+\boldsymbol{H}_{\Omega_{1}}\right)\right\Vert _{1}-\frac{1}{m}\sum_{i\geq2}\left\Vert \mathcal{B}\left(\boldsymbol{H}_{\Omega_{i}}\right)\right\Vert _{1}\\
 & \geq\left(1-\gamma_{k+K_{2}}^{\text{lb}}\right)\left\Vert \boldsymbol{H}_{\Omega_{0}}+\boldsymbol{H}_{\Omega_{1}}\right\Vert _{\text{F}}-\left(1+\gamma_{K_{2}}^{\text{ub}}\right)\sum_{i\geq2}\left\Vert \boldsymbol{H}_{\Omega_{i}}\right\Vert _{\text{F}}\\
 & \geq\frac{(1-\gamma_{k+K_{2}}^{\text{lb}})}{\sqrt{2}}\left(\left\Vert \boldsymbol{H}_{\Omega_{0}}\right\Vert _{\text{F}}+\left\Vert \boldsymbol{H}_{\Omega_{1}}\right\Vert _{\text{F}}\right)\\
 & \quad\quad\quad-\frac{\left(1+\gamma_{K_{2}}^{\text{ub}}\right)}{\sqrt{K_{2}}}\left(\sqrt{k}\left\Vert \boldsymbol{H}_{\Omega_{0}}\right\Vert _{\text{F}}+2\left\Vert \bSigma_{\Omega^{c}}\right\Vert _{1}\right).
\end{align*}
Reorganizing the above equation yields 
\begin{align*}
 & \footnotesize\left[\frac{(1-\gamma_{k+K_{2}}^{\text{lb}})}{\sqrt{2}}-\frac{\left(1+\gamma_{K_{2}}^{\text{ub}}\right)\sqrt{k}}{\sqrt{K_{2}}}\right]\|\boldsymbol{H}_{\Omega_{0}}\|_{\text{F}}+\frac{(1-\gamma_{k+K_{2}}^{\text{lb}})}{\sqrt{2}}\|\boldsymbol{H}_{\Omega_{1}}\|_{\text{F}}\\
 & \quad\leq\frac{2\left(1+\gamma_{K_{2}}^{\text{ub}}\right)}{\sqrt{K_{2}}}\left\Vert \boldsymbol{\Sigma}_{\Omega^{c}}\right\Vert _{1}+\frac{2\epsilon_{1}}{m}.
\end{align*}
Recalling Assumption \eqref{condition-sparse}, one has 
\[
\|\boldsymbol{H}_{\Omega_{0}}\|_{\text{F}}+\|\boldsymbol{H}_{\Omega_{1}}\|_{\text{F}}\leq\frac{2}{\beta_{2}}\left[\frac{\left(1+\gamma_{K_{2}}^{\text{ub}}\right)}{\sqrt{K_{2}}}\left\Vert \boldsymbol{\Sigma}_{\Omega^{c}}\right\Vert _{1}+\frac{\epsilon_{1}}{m}\right].
\]
This along with \eqref{residual_norm} gives 
\[
\left\Vert \boldsymbol{H}\right\Vert _{\text{F}}\leq\left(\frac{C_{1}}{\beta_{2}}+C_{3}\right)\frac{\left\Vert \boldsymbol{\Sigma}_{\Omega^{c}}\right\Vert _{1}}{\sqrt{K_{2}}}+\frac{C_{2}}{\beta_{2}}\frac{\epsilon_{1}}{m}
\]
for some universal constants $C_{1}$, $C_{2}$ and $C_{3}$, as claimed.

\section{Proof of Lemma \ref{lemma:duality-SparsePR-Stable}\label{sec:Proof-of-Lemma-Duality-SparsePR-Stable}}

Before proceeding to the proof, we introduce a few notations for convenience
of presentation. Let $\boldsymbol{X}:=\boldsymbol{x}\boldsymbol{x}^{\top}$,
$\boldsymbol{X}_{\Omega}:=\boldsymbol{x}_{\Omega}\boldsymbol{x}_{\Omega}^{\top}$
and $\boldsymbol{X}_{\text{c}}:=\boldsymbol{X}-\boldsymbol{X}_{\Omega}$,
where $\boldsymbol{x}_{\Omega}$ denotes the best $k$-term approximation
of $\boldsymbol{x}$. We set $\boldsymbol{u}:=\frac{1}{\left\Vert \boldsymbol{x}_{\Omega}\right\Vert _{2}}\boldsymbol{x}_{\Omega}$,
and hence the tangent space $T$ with respect to $\boldsymbol{X}_{\Omega}$
and its orthogonal complement $T^{\perp}$ are characterized by 
\begin{align*}
T & :=\left\{ \boldsymbol{u}\boldsymbol{z}^{\top}+\boldsymbol{z}\boldsymbol{u}^{\top}\mid\boldsymbol{z}\in\mathbb{R}^{n}\right\} ,\\
T^{\perp} & :=\left\{ \left(\boldsymbol{I}-\boldsymbol{u}\boldsymbol{u}^{\top}\right)\boldsymbol{M}\left(\boldsymbol{I}-\boldsymbol{u}\boldsymbol{u}^{\top}\right)\mid\boldsymbol{M}\in\mathbb{R}^{n\times n}\right\} .
\end{align*}
We adopt the notations introduced in \cite{li2012sparse} as follows:
let $\Omega$ denote the support of $\boldsymbol{X}_{\Omega}$, and
decompose the entire matrix space into the direct sum of 3 subspaces
as 
\begin{equation}
\left(T\cap\Omega\right)\oplus\left(T^{\perp}\cap\Omega\right)\oplus\left(\Omega^{\perp}\right).\label{eq:TOmegaDecomposition}
\end{equation}
In fact, one can verify that 
\[
T\cap\Omega=\left\{ \boldsymbol{u}\boldsymbol{z}^{\top}+\boldsymbol{z}\boldsymbol{u}^{\top}\mid\boldsymbol{z}_{\Omega^{c}}=\boldsymbol{0}\right\} ,
\]
and that both the column and row spaces of $T^{\perp}\cap\Omega$
can be spanned by a set of $k-1$ orthonormal vectors that are supported
on $\Omega$ and orthogonal to $\boldsymbol{u}$. As pointed out by
\cite{li2012sparse}, $T$ and $\Omega$ are compatible in the sense
that 
\begin{equation}
\mathcal{P}_{T}\mathcal{P}_{\Omega}=\mathcal{P}_{\Omega}\mathcal{P}_{T}=\mathcal{P}_{T\cap\Omega}.\label{eq:compatible}
\end{equation}
In the following, we will use $\delta_{r,l}^{\mathrm{lb}}$ and $\delta_{r,l}^{\mathrm{ub}}$
to represent $\delta_{k,r,l}^{\mathrm{lb}}$ and $\delta_{k,r,l}^{\mathrm{ub}}$
for brevity, whenever the value of $k$ is clear from the context.

Suppose that $\hat{\boldsymbol{X}}=\boldsymbol{x}\boldsymbol{x}^{\top}+\boldsymbol{H}$
is the solution to (\ref{eq:AlgorithmSparsePR}). Then for any $\boldsymbol{W}\in T^{\perp}$
and $\boldsymbol{Y}\in\Omega^{\perp}$ satisfying $\|\boldsymbol{W}\|\leq1$
and $\|\boldsymbol{Y}\|_{\infty}\leq1$, the matrix $\boldsymbol{u}\boldsymbol{u}^{\top}+\boldsymbol{W}+\lambda\text{sign}\left(\boldsymbol{u}\right)\text{sign}\left(\boldsymbol{u}\right)^{\top}+\lambda\boldsymbol{Y}$
forms a subgradient of the function $\left\Vert \cdot\right\Vert _{*}+\lambda\left\Vert \cdot\right\Vert _{1}$
at point $\boldsymbol{X}_{\Omega}$. If we pick $\boldsymbol{W}$
and $\boldsymbol{Y}$ such that $\boldsymbol{Y}=\text{sgn}\left(\boldsymbol{H}_{\Omega^{\perp}}\right)$
and $\left\langle \boldsymbol{W},\boldsymbol{H}\right\rangle =\left\Vert \boldsymbol{H}_{T^{\perp}\cap\Omega}\right\Vert _{*}$,
then 
\begin{align}
0\geq & \left\Vert \boldsymbol{X}+\boldsymbol{H}\right\Vert _{*}+\lambda\left\Vert \boldsymbol{X}+\boldsymbol{H}\right\Vert _{1}-\left\Vert \boldsymbol{X}\right\Vert _{*}-\lambda\left\Vert \boldsymbol{X}\right\Vert _{1}\label{inequality1}\\
\geq & \left\Vert \boldsymbol{X}_{\Omega}+\boldsymbol{H}\right\Vert _{*}-\left\Vert \boldsymbol{X}_{\text{c}}\right\Vert _{*}+\lambda\left\Vert \boldsymbol{X}_{\Omega}+\boldsymbol{H}\right\Vert _{1}-\lambda\left\Vert \boldsymbol{X}_{\text{c}}\right\Vert _{1}\nonumber \\
 & \quad-\left\Vert \boldsymbol{X}_{\Omega}\right\Vert _{*}-\left\Vert \boldsymbol{X}_{\text{c}}\right\Vert _{*}-\lambda\left\Vert \boldsymbol{X}_{\Omega}\right\Vert _{1}-\lambda\left\Vert \boldsymbol{X}_{\text{c}}\right\Vert _{1}\label{inequality2}\\
\geq & \left\langle \boldsymbol{u}\boldsymbol{u}^{\top}+\boldsymbol{W},\boldsymbol{H}\right\rangle +\left\langle \lambda\text{sign}\left(\boldsymbol{u}\right)\text{sign}\left(\boldsymbol{u}\right)^{\top}+\lambda\boldsymbol{Y},\boldsymbol{H}\right\rangle \nonumber \\
 & \quad-2\left\Vert \boldsymbol{X}_{\text{c}}\right\Vert _{*}-2\lambda\left\Vert \boldsymbol{X}_{\text{c}}\right\Vert _{1}\label{inequality3}\\
= & \left\langle \boldsymbol{u}\boldsymbol{u}^{\top},\boldsymbol{H}_{T}\right\rangle +\lambda\left\langle \mathcal{P}_{T}\left(\text{sign}\left(\boldsymbol{u}\right)\text{sign}\left(\boldsymbol{u}\right)^{\top}\right),\boldsymbol{H}_{T}\right\rangle \nonumber \\
 & \quad+\lambda\left\langle \mathcal{P}_{T^{\perp}}\left(\text{sign}\left(\boldsymbol{u}\right)\text{sign}\left(\boldsymbol{u}\right)^{\top}\right),\boldsymbol{H}_{T^{\perp}}\right\rangle \nonumber \\
 & \quad+\left\Vert \boldsymbol{H}_{T^{\perp}\cap\Omega}\right\Vert _{*}+\lambda\left\Vert \boldsymbol{H}_{\Omega^{\perp}}\right\Vert _{1}-2\left\Vert \boldsymbol{X}_{\text{c}}\right\Vert _{*}-2\lambda\left\Vert \boldsymbol{X}_{\text{c}}\right\Vert _{1}\nonumber \\
\geq & \left\langle \boldsymbol{u}\boldsymbol{u}^{\top}+\lambda\mathcal{P}_{T}\left(\text{sign}\left(\boldsymbol{u}\right)\text{sign}\left(\boldsymbol{u}\right)^{\top}\right),\boldsymbol{H}_{T\cap\Omega}\right\rangle +\left\Vert \boldsymbol{H}_{T^{\perp}\cap\Omega}\right\Vert _{*}\nonumber \\
 & \quad+\lambda\left\Vert \boldsymbol{H}_{\Omega^{\perp}}\right\Vert _{1}-2\left\Vert \boldsymbol{X}_{\text{c}}\right\Vert _{*}-2\lambda\left\Vert \boldsymbol{X}_{\text{c}}\right\Vert _{1},\label{eq:MyNewBD}
\end{align}
where \eqref{inequality1} follows from the optimality of $\hat{\boldsymbol{X}}$,
\eqref{inequality2} follows from the definitions of $\boldsymbol{X}_{\Omega}$
and $\boldsymbol{X}_{\mathrm{c}}$ and the triangle inequality, \eqref{inequality3}
follows from the definition of subgradient. Finally, \eqref{eq:MyNewBD}
follows from the following two facts:

(i) $\boldsymbol{H}_{T^{\perp}}\succeq0$, a consequence of the feasibility
constraint of (\ref{eq:AlgorithmSparsePR}). This further gives 
\begin{align*}
 & \left\langle \mathcal{P}_{T^{\perp}}\left(\text{sign}\left(\boldsymbol{u}\right)\text{sign}\left(\boldsymbol{u}\right)^{\top}\right),\boldsymbol{H}_{T^{\perp}}\right\rangle \\
 & \quad=\left\langle \text{sign}\left(\boldsymbol{u}\right)\text{sign}\left(\boldsymbol{u}\right)^{\top},\boldsymbol{H}_{T^{\perp}}\right\rangle \\
 & \quad=\text{sign}\left(\boldsymbol{u}\right)^{\top}\boldsymbol{H}_{T^{\perp}}\text{sign}\left(\boldsymbol{u}\right)\geq0.
\end{align*}

(ii) It follows from (\ref{eq:compatible}) and the fact $\text{sign}\left(\boldsymbol{u}\right)\text{sign}\left(\boldsymbol{u}\right)^{\top}\in\Omega$
that 
\begin{align}
 & \left\langle \mathcal{P}_{T}\left(\text{sign}\left(\boldsymbol{u}\right)\text{sign}\left(\boldsymbol{u}\right)^{\top}\right),\boldsymbol{H}_{T}\right\rangle \nonumber \\
 & \quad=\left\langle \mathcal{P}_{T\cap\Omega}\left(\text{sign}\left(\boldsymbol{u}\right)\text{sign}\left(\boldsymbol{u}\right)^{\top}\right),\boldsymbol{H}_{T\cap\Omega}\right\rangle .\label{eq:compatibilityTOmega}
\end{align}

Since any matrix in $T$ has rank at most 2, one can bound 
\begin{align}
 & \left\Vert \mathcal{P}_{T}\left(\text{sign}\left(\boldsymbol{u}\right)\text{sign}\left(\boldsymbol{u}\right)^{\top}\right)\right\Vert _{*}^{2}\leq2\left\Vert \mathcal{P}_{T}\left(\text{sign}\left(\boldsymbol{u}\right)\text{sign}\left(\boldsymbol{u}\right)^{\top}\right)\right\Vert _{\text{F}}^{2}\nonumber \\
 & \quad\leq4\left\Vert \boldsymbol{u}\boldsymbol{u}^{\top}\text{sign}\left(\boldsymbol{u}\right)\text{sign}\left(\boldsymbol{u}\right)^{\top}\right\Vert _{\text{F}}^{2}\label{eq:uuT}\\
 & \quad=4\left|\left\langle \boldsymbol{u},\text{sign}\left(\boldsymbol{u}\right)\right\rangle \right|^{2}\left\Vert \boldsymbol{u}\cdot\text{sign}\left(\boldsymbol{u}\right)^{\top}\right\Vert _{\text{F}}^{2}\nonumber \\
 & \quad=4\left|\left\langle \boldsymbol{u},\text{sign}\left(\boldsymbol{u}\right)\right\rangle \right|^{2}\left\Vert \text{sign}\left(\boldsymbol{u}\right)\right\Vert _{\text{F}}^{2}\nonumber \\
 & \quad\leq4k\left\Vert \boldsymbol{u}\right\Vert _{1}^{2}\left\Vert \text{sign}\left(\boldsymbol{u}\right)\right\Vert _{\infty}^{2}\leq4k\left\Vert \boldsymbol{u}\right\Vert _{1}^{2}\leq\frac{4}{\lambda^{2}},\label{eq:BoundTsignU}
\end{align}
where (\ref{eq:uuT}) follows from the definition of $\mathcal{P}_{T}$
that 
\begin{align*}
 & \small\left\Vert \mathcal{P}_{T}\left(\text{sign}\left(\boldsymbol{u}\right)\text{sign}\left(\boldsymbol{u}\right)^{\top}\right)\right\Vert _{\mathrm{F}}^{2}\\
 & \quad=\footnotesize\left\Vert \boldsymbol{u}\boldsymbol{u}^{\top}\text{sign}\left(\boldsymbol{u}\right)\text{sign}\left(\boldsymbol{u}\right)^{\top}+\left(\boldsymbol{I}-\boldsymbol{u}\boldsymbol{u}^{\top}\right)\text{sign}\left(\boldsymbol{u}\right)\text{sign}\left(\boldsymbol{u}\right)^{\top}\boldsymbol{u}\boldsymbol{u}^{\top}\right\Vert _{\mathrm{F}}^{2}\\
 & \quad=\small\left\Vert \boldsymbol{u}\boldsymbol{u}^{\top}\text{sign}\left(\boldsymbol{u}\right)\text{sign}\left(\boldsymbol{u}\right)^{\top}\right\Vert _{\mathrm{F}}^{2}\\
 & \quad\quad\quad\quad\small+\left\Vert \left(\boldsymbol{I}-\boldsymbol{u}\boldsymbol{u}^{\top}\right)\text{sign}\left(\boldsymbol{u}\right)\text{sign}\left(\boldsymbol{u}\right)^{\top}\boldsymbol{u}\boldsymbol{u}^{\top}\right\Vert _{\mathrm{F}}^{2}\\
 & \quad\leq\small\left\Vert \boldsymbol{u}\boldsymbol{u}^{\top}\text{sign}\left(\boldsymbol{u}\right)\text{sign}\left(\boldsymbol{u}\right)^{\top}\right\Vert _{\mathrm{F}}^{2}+\left\Vert \text{sign}\left(\boldsymbol{u}\right)\text{sign}\left(\boldsymbol{u}\right)^{\top}\boldsymbol{u}\boldsymbol{u}^{\top}\right\Vert _{\mathrm{F}}^{2}\\
 & \quad=\small2\left\Vert \boldsymbol{u}\boldsymbol{u}^{\top}\text{sign}\left(\boldsymbol{u}\right)\text{sign}\left(\boldsymbol{u}\right)^{\top}\right\Vert _{\mathrm{F}}^{2},
\end{align*}
and \eqref{eq:BoundTsignU} arises from the assumption on $\lambda$.
Combining \eqref{eq:BoundTsignU} with (\ref{eq:MyNewBD}) yields
\begin{align}
 & \left\Vert \boldsymbol{H}_{T^{\perp}\cap\Omega}\right\Vert _{*}+\lambda\left\Vert \boldsymbol{H}_{\Omega^{\perp}}\right\Vert _{1}\nonumber \\
 & \quad\leq-\left\langle \boldsymbol{u}\boldsymbol{u}^{\top},\boldsymbol{H}_{T\cap\Omega}\right\rangle -\lambda\left\langle \mathcal{P}_{T}\left(\text{sign}\left(\boldsymbol{u}\right)\text{sign}\left(\boldsymbol{u}\right)^{\top}\right),\boldsymbol{H}_{T\cap\Omega}\right\rangle \nonumber \\
 & \quad\quad\quad\quad+2\left\Vert \boldsymbol{X}_{\text{c}}\right\Vert _{*}+2\lambda\left\Vert \boldsymbol{X}_{\text{c}}\right\Vert _{1}\nonumber \\
 & \quad\leq\left|\boldsymbol{u}^{\top}\boldsymbol{H}_{T\cap\Omega}\boldsymbol{u}\right|+\lambda\left\Vert \mathcal{P}_{T}\left(\text{sign}\left(\boldsymbol{u}\right)\text{sign}\left(\boldsymbol{u}\right)^{\top}\right)\right\Vert _{*}\cdot\left\Vert \boldsymbol{H}_{T\cap\Omega}\right\Vert \nonumber \\
 & \quad\quad\quad\quad+2\left\Vert \boldsymbol{X}_{\text{c}}\right\Vert _{*}+2\lambda\left\Vert \boldsymbol{X}_{\text{c}}\right\Vert _{1}\nonumber \\
 & \quad\leq3\left\Vert \boldsymbol{H}_{T\cap\Omega}\right\Vert +2\left\Vert \boldsymbol{X}_{\text{c}}\right\Vert _{*}+2\lambda\left\Vert \boldsymbol{X}_{\text{c}}\right\Vert _{1},\label{eq:TcapOmegaBound-StablePR}
\end{align}
where \eqref{eq:TcapOmegaBound-StablePR} results from $\left\Vert \boldsymbol{u}\right\Vert _{2}=1$
and (\ref{eq:BoundTsignU}).

Divide $\boldsymbol{H}_{T^{\perp}\cap\Omega}$ into $M_{1}:=\left\lceil \frac{k-2}{K_{1}}\right\rceil $
orthogonal matrices $\boldsymbol{H}_{T^{\perp}\cap\Omega}^{\left(1\right)}$,
$\boldsymbol{H}_{T^{\perp}\cap\Omega}^{\left(2\right)}$, $\cdots$,
$\boldsymbol{H}_{T^{\perp}\cap\Omega}^{\left(M_{1}\right)}\in T^{\perp}\cap\Omega$
satisfying the following properties: (i) the largest singular value
of $\boldsymbol{H}_{T^{\perp}\cap\Omega}^{\left(i+1\right)}$ does
not exceed the smallest non-zero singular value of $\boldsymbol{H}_{T^{\perp}\cap\Omega}^{\left(i\right)}$,
and (ii) 
\begin{align*}
\left\Vert \boldsymbol{H}_{T^{\perp}\cap\Omega}\right\Vert _{*} & =\sum_{i=1}^{M}\left\Vert \boldsymbol{H}_{T^{\perp}\cap\Omega}^{\left(i\right)}\right\Vert _{*},\\
\text{rank}\left(\boldsymbol{H}_{T^{\perp}\cap\Omega}^{\left(i\right)}\right) & =K_{1}\quad\text{ }(1\leq i\leq M_{1}-1).
\end{align*}
In the meantime, divide $\boldsymbol{H}_{\Omega^{\perp}}$ into $M_{2}=\left\lceil \frac{n^{2}-k^{2}}{K_{2}}\right\rceil $
orthogonal matrices $\boldsymbol{H}_{\Omega^{\perp}}^{\left(1\right)}$,
$\boldsymbol{H}_{\Omega^{\perp}}^{\left(2\right)}$, $\cdots$, $\boldsymbol{H}_{\Omega^{\perp}}^{\left(M_{2}\right)}\in\Omega^{\perp}$
of \emph{non-overlapping} support such that (i) the largest entry
magnitude of $\boldsymbol{H}_{\Omega^{\perp}}^{\left(i+1\right)}$
does not exceed the magnitude of the smallest non-zero entry of $\boldsymbol{H}_{\Omega^{\perp}}^{\left(i\right)}$,
and (ii) 
\[
\left\Vert \boldsymbol{H}_{\Omega^{\perp}}^{\left(i\right)}\right\Vert _{0}=K_{2}\text{ }(1\leq i\leq M_{2}-1).
\]
This decomposition gives rise to the following bound 
\begin{align*}
\sum_{i=2}^{M_{1}}\left\Vert \boldsymbol{H}_{T^{\perp}\cap\Omega}^{\left(i\right)}\right\Vert _{\text{F}}\leq\sum_{i=2}^{M_{1}}\frac{1}{\sqrt{K_{1}}}\left\Vert \boldsymbol{H}_{T^{\perp}\cap\Omega}^{\left(i-1\right)}\right\Vert _{*}\leq\frac{1}{\sqrt{K_{1}}}\left\Vert \boldsymbol{H}_{T^{\perp}\cap\Omega}\right\Vert _{*},
\end{align*}
which combined with the RIP-$\ell_{2}/\ell_{1}$ property of $\mathcal{B}$
yields 
\begin{align}
\sum_{i=2}^{M_{1}}\frac{1}{m}\left\Vert \mathcal{B}\left(\boldsymbol{H}_{T^{\perp}\cap\Omega}^{\left(i\right)}\right)\right\Vert _{1}\leq\sum_{i=2}^{M_{1}}\frac{\left(1+\delta_{K_{1},K_{2}}^{\text{ub}}\right)}{m}\left\Vert \boldsymbol{H}_{T^{\perp}\cap\Omega}^{\left(i\right)}\right\Vert _{\mathrm{F}}\nonumber \\
\leq\frac{\left(1+\delta_{K_{1},K_{2}}^{\text{ub}}\right)}{\sqrt{K_{1}}}\left\Vert \boldsymbol{H}_{T^{\perp}\cap\Omega}\right\Vert _{*},\label{eq:BHbound}
\end{align}
and, similarly, 
\begin{align}
\sum_{i=2}^{M_{2}}\frac{1}{m}\left\Vert \mathcal{B}\left(\boldsymbol{H}_{\Omega^{\perp}}^{\left(i\right)}\right)\right\Vert _{1} & \leq\sum_{i=2}^{M_{1}}\frac{\left(1+\delta_{K_{1},K_{2}}^{\text{ub}}\right)}{m}\left\Vert \boldsymbol{H}_{\Omega^{\perp}}^{\left(i\right)}\right\Vert _{\mathrm{F}}\label{eq:BHOmegaBound}\\
 & \leq\frac{\left(1+\delta_{K_{1},K_{2}}^{\text{ub}}\right)}{\sqrt{K_{2}}}\left\Vert \boldsymbol{H}_{\Omega^{\perp}}\right\Vert _{1}.\nonumber 
\end{align}
The above argument relies on our construction scheme that $\text{rank}(\boldsymbol{H}_{T^{\perp}\cap\Omega}^{\left(i\right)})\leq K_{1}$,
$\text{supp}\left(\boldsymbol{H}_{T^{\perp}\cap\Omega}^{\left(i\right)}\right)\subseteq\Omega$,
and $\left\Vert \boldsymbol{H}_{\Omega^{\perp}}^{\left(i\right)}\right\Vert _{0}\leq K_{2}$,
and hence all of $\boldsymbol{H}_{T^{\perp}\cap\Omega}^{\left(i\right)}$
and $\boldsymbol{H}_{\Omega^{\perp}}^{\left(i\right)}$ ($i\geq1$)
belong to $\mathcal{M}_{k,K_{1},K_{2}}$.

Set $K_{2}:=\left\lceil \frac{K_{1}}{\lambda^{2}}\right\rceil $,
and hence $\sqrt{\frac{K_{1}}{K_{2}}}\leq\lambda$. Recalling $\boldsymbol{H}=\boldsymbol{H}_{T\cap\Omega}+\boldsymbol{H}_{T^{\perp}\cap\Omega}+\boldsymbol{H}_{\Omega^{\perp}}$,
one can proceed as follows 
\begin{align*}
\frac{2\epsilon_{1}}{m}\geq & \text{ }\small\frac{1}{m}\left\Vert \mathcal{B}\left(\boldsymbol{H}\right)\right\Vert _{1}\\
\geq & \text{ }\small\frac{1}{m}\left\Vert \mathcal{B}\left(\boldsymbol{H}_{T\cap\Omega}+\boldsymbol{H}_{T^{\perp}\cap\Omega}^{(1)}+\boldsymbol{H}_{\Omega^{\perp}}^{(1)}\right)\right\Vert _{1}\\
 & \quad-\small\sum_{i=2}^{M_{1}}\frac{1}{m}\left\Vert \mathcal{B}\left(\boldsymbol{H}_{T^{\perp}\cap\Omega}^{\left(i\right)}\right)\right\Vert _{1}-\sum_{i=2}^{M_{2}}\frac{1}{m}\left\Vert \mathcal{B}\left(\boldsymbol{H}_{\Omega^{\perp}}^{\left(i\right)}\right)\right\Vert _{1}\\
\geq & \small\left(1-\delta_{2K_{1},2K_{2}}^{\text{lb}}\right)\left\Vert \boldsymbol{H}_{T\cap\Omega}+\boldsymbol{H}_{T^{\perp}\cap\Omega}^{(1)}+\boldsymbol{H}_{\Omega^{\perp}}^{(1)}\right\Vert _{\text{F}}\\
 & \quad\small-\frac{\left(1+\delta_{K_{1},K_{2}}^{\text{ub}}\right)\left\Vert \boldsymbol{H}_{T^{\perp}\cap\Omega}\right\Vert _{*}}{\sqrt{K_{1}}}-\frac{\left(1+\delta_{K_{1},K_{2}}^{\text{ub}}\right)\left\Vert \boldsymbol{H}_{\Omega^{\perp}}\right\Vert _{1}}{\sqrt{K_{2}}}\\
\geq & \small\left(1-\delta_{2K_{1},2K_{2}}^{\text{lb}}\right)\left\Vert \boldsymbol{H}_{T\cap\Omega}+\boldsymbol{H}_{T^{\perp}\cap\Omega}^{(1)}+\boldsymbol{H}_{\Omega^{\perp}}^{(1)}\right\Vert _{\text{F}}\\
 & \quad\small-\frac{\left(1+\delta_{K_{1},K_{2}}^{\text{ub}}\right)}{\sqrt{K_{1}}}\left(\left\Vert \boldsymbol{H}_{T^{\perp}\cap\Omega}\right\Vert _{*}+\lambda\left\Vert \boldsymbol{H}_{\Omega^{\perp}}\right\Vert _{1}\right)\\
\geq & \small\text{ }\frac{\left(1-\delta_{2K_{1},2K_{2}}^{\text{lb}}\right)}{\sqrt{3}}\left(\left\Vert \boldsymbol{H}_{T\cap\Omega}\right\Vert _{\text{F}}+\left\Vert \boldsymbol{H}_{T^{\perp}\cap\Omega}^{(1)}\right\Vert _{\text{F}}+\left\Vert \boldsymbol{H}_{\Omega^{\perp}}^{(1)}\right\Vert _{\text{F}}\right)\\
 & \quad\small-\frac{\left(1+\delta_{K_{1},K_{2}}^{\text{ub}}\right)}{\sqrt{K_{1}}}\left(\left\Vert \boldsymbol{H}_{T^{\perp}\cap\Omega}\right\Vert _{*}+\lambda\left\Vert \boldsymbol{H}_{\Omega^{\perp}}\right\Vert _{1}\right).
\end{align*}
This taken collectively with (\ref{eq:TcapOmegaBound-StablePR}) gives
\begin{align*}
 & \frac{2\left(1+\delta_{K_{1},K_{2}}^{\text{ub}}\right)}{\sqrt{K_{1}}}\left(\left\Vert \boldsymbol{X}_{\text{c}}\right\Vert _{*}+\lambda\left\Vert \boldsymbol{X}_{\text{c}}\right\Vert _{1}\right)+\frac{2\epsilon}{m}\\
 & \quad\geq\left(\frac{1-\delta_{2K_{1},2K_{2}}^{\text{lb}}}{\sqrt{3}}-\frac{3\left(1+\delta_{K_{1},K_{2}}^{\text{ub}}\right)}{\sqrt{K_{1}}}\right)\cdot\\
 & \quad\quad\quad\quad\left(\left\Vert \boldsymbol{H}_{T\cap\Omega}\right\Vert _{\text{F}}+\left\Vert \boldsymbol{H}_{T^{\perp}\cap\Omega}^{(1)}\right\Vert _{\text{F}}+\left\Vert \boldsymbol{H}_{\Omega^{\perp}}^{(1)}\right\Vert _{\text{F}}\right).
\end{align*}
Therefore, if we know that 
\[
\frac{\frac{1-\delta_{2K_{1},2K_{2}}^{\text{lb}}}{\sqrt{3}}-\frac{3\left(1+\delta_{K_{1},K_{2}}^{\text{ub}}\right)}{\sqrt{K_{1}}}}{2\max\left\{ \frac{1+\delta_{K_{1},K_{2}}^{\text{ub}}}{\sqrt{K_{1}}},1\right\} }\geq\beta_{3}>0
\]
for some absolute constant $\beta_{3}$, then 
\begin{align}
 & \left\Vert \boldsymbol{H}_{T\cap\Omega}\right\Vert _{\text{F}}+\left\Vert \boldsymbol{H}_{T^{\perp}\cap\Omega}^{(1)}\right\Vert _{\text{F}}+\left\Vert \boldsymbol{H}_{\Omega^{\perp}}^{(1)}\right\Vert _{\text{F}}\nonumber \\
 & \quad\leq\frac{1}{\beta_{3}}\left(\left\Vert \boldsymbol{X}_{\text{c}}\right\Vert _{*}+\lambda\left\Vert \boldsymbol{X}_{\text{c}}\right\Vert _{1}+\frac{\epsilon_{1}}{m}\right).\label{eq:UBSparsePR}
\end{align}

On the other hand, we know from (\ref{eq:BHbound}) and (\ref{eq:BHOmegaBound})
that 
\begin{align}
 & \sum_{i=2}^{M_{1}}\left\Vert \boldsymbol{H}_{T^{\perp}\cap\Omega}^{(i)}\right\Vert _{\text{F}}+\sum_{i=2}^{M_{2}}\left\Vert \boldsymbol{H}_{\Omega^{\perp}}^{(i)}\right\Vert _{\text{F}}\nonumber \\
 & \quad\leq\small\frac{1}{1-\delta_{K_{1},K_{2}}^{\text{lb}}}\sum_{i=2}^{M_{1}}\left\Vert \mathcal{B}\left(\boldsymbol{H}_{T^{\perp}\cap\Omega}^{(i)}\right)\right\Vert _{1}+\sum_{i=2}^{M_{2}}\left\Vert \mathcal{B}\left(\boldsymbol{H}_{\Omega^{\perp}}^{(i)}\right)\right\Vert _{1}\nonumber \\
 & \quad\leq\small\frac{\left(1+\delta_{K_{1},K_{2}}^{\text{ub}}\right)\left\Vert \boldsymbol{H}_{T^{\perp}\cap\Omega}\right\Vert _{*}}{\left(1-\delta_{K_{1},K_{2}}^{\text{lb}}\right)\sqrt{K_{1}}}+\frac{\left(1+\delta_{K_{1},K_{2}}^{\text{ub}}\right)\left\Vert \boldsymbol{H}_{\Omega^{\perp}}\right\Vert _{1}}{\left(1-\delta_{K_{1},K_{2}}^{\text{lb}}\right)\sqrt{K_{2}}}\label{eq:Second}\\
 & \quad=\small\frac{1+\delta_{K_{1},K_{2}}^{\text{ub}}}{\left(1-\delta_{K_{1},K_{2}}^{\text{lb}}\right)\sqrt{K_{1}}}\left(\left\Vert \boldsymbol{H}_{T^{\perp}\cap\Omega}\right\Vert _{*}+\lambda\left\Vert \boldsymbol{H}_{\Omega^{\perp}}\right\Vert _{1}\right)\nonumber \\
 & \quad\leq\small\frac{1+\delta_{K_{1},K_{2}}^{\text{ub}}}{\left(1-\delta_{K_{1},K_{2}}^{\text{lb}}\right)\sqrt{K_{1}}}\left(3\left\Vert \boldsymbol{H}_{T\cap\Omega}\right\Vert +2\left\Vert \boldsymbol{X}_{\text{c}}\right\Vert _{*}+2\lambda\left\Vert \boldsymbol{X}_{\text{c}}\right\Vert _{1}\right),\nonumber 
\end{align}
where (\ref{eq:Second}) is a consequence of (\ref{eq:BHbound}) and
(\ref{eq:BHOmegaBound}), and the last inequality arises from (\ref{eq:TcapOmegaBound-StablePR}).
This together with (\ref{eq:UBSparsePR}) completes the proof.

\section{Proof of Lemma \ref{lemma:isotropy}}

\label{proof:lemma:isotropy}Simple calculation yields that 
\begin{equation}
\mathbb{E}\left[\boldsymbol{A}_{i}\left\langle \boldsymbol{A}_{i},\boldsymbol{X}\right\rangle \right]=2\boldsymbol{X}+\left(1+\frac{\mu_{4}-3}{n}\right)\text{tr}\left(\boldsymbol{X}\right)\cdot\boldsymbol{I}.\label{eq:EAA_original}
\end{equation}

When $\mu_{4}=3$, one can see that 
\begin{align}
\mathbb{E}\left[\boldsymbol{B}_{i}\left\langle \boldsymbol{B}_{i},\boldsymbol{X}\right\rangle \right] & =\frac{1}{4}\mathbb{E}\left[\left(\boldsymbol{A}_{2i-1}-\boldsymbol{A}_{2i}\right)\left\langle \boldsymbol{A}_{2i-1}-\boldsymbol{A}_{2i},\boldsymbol{X}\right\rangle \right]\nonumber \\
 & =\boldsymbol{X}.\label{eq:EAA_original-1}
\end{align}

When $\mu_{4}\neq3$, consider the linear combination 
\[
\boldsymbol{B}=a\boldsymbol{A}_{1}+b\boldsymbol{A}_{2}+c\boldsymbol{A}_{3},
\]
where we aim to find the coefficients $a,b$ and $c$ that makes $\boldsymbol{B}$
isotropic. If we further require 
\begin{equation}
\mathbb{E}\left[\boldsymbol{B}\right]=a+b+c=\frac{\epsilon}{\sqrt{n}},\label{eq:EBMean}
\end{equation}
then one can compute 
\begin{align*}
 & \mathbb{E}\left[\boldsymbol{B}\left\langle \boldsymbol{B},\boldsymbol{X}\right\rangle \right]=2\left(a^{2}+b^{2}+c^{2}\right)\boldsymbol{X}+\\
 & \quad\small\left[\left(1+\frac{\mu_{4}-3}{n}\right)\left(a^{2}+b^{2}+c^{2}\right)+2\left(ab+bc+ac\right)\right]\text{tr}\left(\boldsymbol{X}\right)\cdot\boldsymbol{I}.
\end{align*}
Our goal is thus to determine $a$, $b$ and $c$ that satisfy 
\[
\left(1+\frac{\mu_{4}-3}{n}\right)\left(a^{2}+b^{2}+c^{2}\right)+2\left(ab+bc+ac\right)=0,
\]
which combined with (\ref{eq:EBMean}) gives 
\begin{equation}
\frac{\mu_{4}-3}{n}\left(a^{2}+b^{2}+c^{2}\right)+\frac{\epsilon^{2}}{n}=0.\label{eq:Constraint_abc}
\end{equation}

If we set $a=1$, then (\ref{eq:Constraint_abc}) reduces to 
\[
\frac{\mu_{4}-3}{n}\left(1+b^{2}+\left(\frac{\epsilon}{\sqrt{n}}-1-b\right)^{2}\right)+\frac{\epsilon^{2}}{n}=0
\]
\[
\Rightarrow\quad b^{2}+b\left(1-\frac{\epsilon}{\sqrt{n}}\right)+\frac{1}{2}\left(1-\frac{\epsilon}{\sqrt{n}}\right)^{2}+\frac{1}{2}+\frac{\epsilon^{2}}{2\left(\mu_{4}-3\right)}=0.
\]
Solving this quadratic equation yields 
\begin{equation}
b=\frac{-\left(1-\frac{\epsilon}{\sqrt{n}}\right)+\sqrt{\Delta}}{2};\quad c=\frac{-\left(1-\frac{\epsilon}{\sqrt{n}}\right)-\sqrt{\Delta}}{2},
\end{equation}
where 
\begin{align*}
\Delta & :=\left(1-\frac{\epsilon}{\sqrt{n}}\right)^{2}-4\left(\frac{1}{2}\left(1-\frac{\epsilon}{\sqrt{n}}\right)^{2}+\frac{1}{2}+\frac{\epsilon^{2}}{2\left(\mu_{4}-3\right)}\right)\\
 & =-\left(1-\frac{\epsilon}{n}\right)^{2}-2-\frac{2\epsilon^{2}}{\mu_{4}-3}.
\end{align*}
Note that $\Delta>0$ when $\epsilon^{2}>1.5\cdot\left|3-\mu_{4}\right|$
. Also, $b$ and $c$ satisfy 
\begin{equation}
1+b^{2}+c^{2}=\frac{\epsilon^{2}}{3-\mu_{4}}.
\end{equation}
By choosing $\alpha=\sqrt{\frac{3-\mu_{4}}{2\epsilon^{2}}}$, $\beta=b\alpha$,
and $\gamma=c\alpha$, we derive the form of $\boldsymbol{B}_{i}$
as introduced in (\ref{eq:alpha_beta_gamma}), which satisfies 
\[
\mathbb{E}\left[\boldsymbol{B}_{i}\left\langle \boldsymbol{B}_{i},\boldsymbol{X}\right\rangle \right]=\boldsymbol{X}.
\]

Finally, we remark that for \emph{any} norm $\left\Vert \cdot\right\Vert _{\text{ n}}$.
This can be easily bounded as follows 
\begin{align}
\left\Vert \boldsymbol{B}_{i}\right\Vert _{\text{n}} & \leq\sqrt{\frac{\left|3-\mu_{4}\right|}{2\epsilon^{2}}}\left(1+\left|b\right|+\left|c\right|\right)\max_{i:1\leq i\leq m}\left\Vert \boldsymbol{A}_{i}\right\Vert _{\text{n}}\nonumber \\
 & \leq\sqrt{3}\sqrt{\frac{\left|3-\mu_{4}\right|}{2\epsilon^{2}}\left(1+b^{2}+c^{2}\right)}\max_{i:1\leq i\leq m}\left\Vert \boldsymbol{A}_{i}\right\Vert _{\text{n}}\\
 & =\sqrt{3}\max_{i:1\leq i\leq m}\left\Vert \boldsymbol{A}_{i}\right\Vert _{\text{n}}.
\end{align}
This concludes the proof.

\section{Proof of Lemma \ref{lemma:UB-T(zz)}\label{sec:Proof-of-Lemma:UB-T(zz)}}

Let $\boldsymbol{M}$ represent the symmetric Toeplitz matrix as follows
\[
\boldsymbol{M}=\left[\boldsymbol{M}_{\left|i-l\right|}\right]_{1\leq i,l\leq n}:=\mathcal{T}\left(\boldsymbol{z}\boldsymbol{z}^{\top}\right),
\]
and since each descending diagonal of a Toeplitz matrix is constant,
the entry $\boldsymbol{M}_{k}$ is given by the average of the corresponding
diagonal, i.e. 
\[
\boldsymbol{M}_{k}:=\frac{1}{n-k}\sum_{l=k+1}^{n}\boldsymbol{z}_{l}\boldsymbol{z}_{l-k},\quad0\leq k<n.
\]
Apparently, one has $\mathbb{E}\left[\boldsymbol{M}_{0}\right]=1$
and $\mathbb{E}\left[\boldsymbol{M}_{k}\right]=0$ ($1\leq k<n$).

\begin{comment}
The variance satisfies 
\begin{align}
\forall k\text{ }\left(1\leq k\leq n-1\right):\quad{\bf Var}\left(\boldsymbol{M}_{k}\right) & \leq\mathbb{E}\left(\frac{1}{n-k}\sum_{j=k+1}^{n}\boldsymbol{z}_{j}\boldsymbol{z}_{j-k}\right)^{2}\nonumber \\
 & \leq\frac{3}{\left(n-k\right)^{2}}\sum_{j=k+1}^{n}\left(\mathbb{E}\boldsymbol{z}_{j}^{2}\right)\left(\mathbb{E}\boldsymbol{z}_{j-k}^{2}\right)\nonumber \\
 & =\frac{3}{n-k}.\label{eq:VarMk}
\end{align}
and 
\begin{equation}
{\bf Var}\left(\boldsymbol{M}_{0}\right)=\frac{1}{n}\mathbb{E}\left(\boldsymbol{z}_{1}^{4}\right)=\frac{3}{n}.\label{eq:VarM0}
\end{equation}
In addition, one can also verify that for any $0\leq k<l\leq n-1$,
the cross terms are uncorrelated, i.e. 
\begin{equation}
\mathbb{E}\boldsymbol{M}_{k}\boldsymbol{M}_{l}=\frac{1}{\left(n-k\right)\left(n-l\right)}\left(\sum_{j_{2}=l+1}^{n}\sum_{j_{1}=k+1}^{n}\boldsymbol{z}_{j_{1}}\boldsymbol{z}_{j_{1}-k}\boldsymbol{z}_{j_{2}}\boldsymbol{z}_{j_{2}-l}\right)=0.\label{eq:CorrelationM}
\end{equation}
\end{comment}

The harmonic structure of the Toeplitz matrix $\boldsymbol{M}$ motivates
us to embed it into a circulant matrix $\boldsymbol{C}_{\boldsymbol{M}}$.
Specifically, a $\left(2n-1\right)\times\left(2n-1\right)$ circulant
matrix 
\[
\boldsymbol{C}_{\boldsymbol{M}}:=\left[\begin{array}{cccc}
c_{0} & c_{1} & \cdots & c_{2n-2}\\
c_{2n-2} & c_{0} & c_{1} & c_{2}\\
\vdots & \vdots & \ddots & \vdots\\
c_{1} & c_{2} & \cdots & c_{0}
\end{array}\right]
\]
is constructed such that 
\[
c_{i}:=\begin{cases}
\boldsymbol{M}_{i},\quad & \text{if }0\leq i<n;\\
\boldsymbol{M}_{2n-i-1},\quad & \text{if }n\leq i\leq2n-2.
\end{cases}
\]
Since $\boldsymbol{M}$ is a submatrix of $\boldsymbol{C}_{\boldsymbol{M}}$,
it suffices to bound the spectral norm of $\boldsymbol{C}_{\boldsymbol{M}}$.
Define $\omega_{i}:=\exp\left(\frac{2\pi j}{2n-1}\cdot i\right),$
then the corresponding eigenvalues of $\boldsymbol{C}_{\boldsymbol{M}}$
are given by 
\begin{align*}
\lambda_{i}: & =\sum_{l}c_{l}\omega_{i}^{l}=\boldsymbol{M}_{0}+\sum_{l=1}^{n-1}\boldsymbol{M}_{l}\omega_{i}^{l}+\sum_{l=n}^{2n-2}\boldsymbol{M}_{2n-l-1}\omega_{i}^{l}\\
 & =\boldsymbol{M}_{0}+2\sum_{l=1}^{n-1}\boldsymbol{M}_{l}\cos\left(\frac{2\pi il}{2n-1}\right),
\end{align*}
for $i=0,1,\cdots,2n-2$, which satisfies $\mathbb{E}\lambda_{i}=\mathbb{E}\boldsymbol{M}_{0}=1$.
This leads to an upper bound as follows 
\begin{equation}
\left\Vert \boldsymbol{M}\right\Vert \leq\left\Vert \boldsymbol{C}_{\boldsymbol{M}}\right\Vert \leq\max_{0\leq i\leq2n-2}\left|\lambda_{i}\right|.\label{eq:BoundCirculantTrick}
\end{equation}

Note that $\lambda_{i}$ is a quadratic form in $\left\{ \boldsymbol{z}_{1},\boldsymbol{z}_{2},\cdots,\boldsymbol{z}_{n}\right\} $.
Define the symmetric coefficient matrix $\boldsymbol{G}^{\left(i\right)}$
such that for any $1\leq\alpha,\beta\leq n$, 
\[
\boldsymbol{G}_{\alpha,\beta}^{(i)}=\frac{1}{n-\left|l\right|}\cos\left(\frac{2\pi i\left|l\right|}{2n-1}\right),\quad\text{if }\alpha-\beta=l,
\]
which satisfies 
\begin{align*}
\lambda_{i} & =\mathbb{E}\left[\boldsymbol{M}_{0}\right]+\sum_{1\leq\alpha,\beta\leq n}\boldsymbol{G}_{\alpha,\beta}^{(i)}\left(\boldsymbol{z}_{\alpha}\boldsymbol{z}_{\beta}-\mathbb{E}\left[\boldsymbol{z}_{\alpha}\boldsymbol{z}_{\beta}\right]\right)\\
 & =1+\sum_{1\leq\alpha,\beta\leq n}\boldsymbol{G}_{\alpha,\beta}^{(i)}\left(\boldsymbol{z}_{\alpha}\boldsymbol{z}_{\beta}-\mathbb{E}\left[\boldsymbol{z}_{\alpha}\boldsymbol{z}_{\beta}\right]\right).
\end{align*}
When $\boldsymbol{z}$ are drawn from a sub-Gaussian measure, Lemma~\ref{lemma:hanson-wright}
asserts that there exists an absolute constant $c_{10}>0$ such that
\begin{equation}
\mathbb{P}\left\{ \left|\lambda_{i}-1\right|\geq t\right\} \leq\exp\left(-c_{10}\min\left\{ \frac{t}{\Vert\boldsymbol{G}^{(i)}\Vert},\frac{t^{2}}{\Vert\boldsymbol{G}^{(i)}\Vert_{\text{F}}^{2}}\right\} \right)\label{eq:UBquadraticForm}
\end{equation}
holds for any $t>0$.

It remains to compute $\Vert\boldsymbol{G}^{(i)}\Vert_{\text{F}}$
and $\Vert\boldsymbol{G}^{(i)}\Vert$. Since $\boldsymbol{G}^{(i)}$
is a symmetric Toeplitz matrix, we have 
\begin{equation}
\Vert\boldsymbol{G}^{(i)}\Vert_{\text{F}}^{2}=\sum_{\alpha,\beta=1}^{n}\left|\boldsymbol{G}_{\alpha,\beta}\right|^{2}\leq2\sum_{l=0}^{n-1}\frac{1}{n-l}\leq2\log n.\label{eq:FnormG}
\end{equation}
It then follows that 
\begin{equation}
\Vert\boldsymbol{G}^{(i)}\Vert\leq\Vert\boldsymbol{G}^{(i)}\Vert_{\text{F}}\leq\sqrt{2\log n}.\label{eq:NormG}
\end{equation}
Substituting these two bounds into (\ref{eq:UBquadraticForm}) immediately
yields that there exists a constant $c_{12}>0$ such that 
\begin{equation}
\lambda_{i}\leq c_{12}\log^{\frac{3}{2}}n,\quad1\leq i\leq2n-2
\end{equation}
holds with probability exceeding $1-\frac{1}{n^{10}}$. This taken
collectively with (\ref{eq:BoundCirculantTrick}) concludes the proof.

\section{Proof of Lemma \ref{lemma-Deviation-BB}\label{sec:Proof-of-Lemma-Deviation-BB}}

For technical convenience, we introduce another collection of events
\[
\forall1\leq i\leq m:\quad F_{i}:=\left\{ \left\Vert \boldsymbol{B}_{i}\right\Vert _{\text{F}}\leq20n\log n\right\} .
\]
Since the restriction of $\mathcal{B}_{i}$ to Toeplitz matrices is
isotropic and $\mathcal{T}\mathcal{B}_{i}^{*}\mathcal{B}_{i}\mathcal{T}\succeq0$,
we have $\mathcal{\mathcal{T}}=\mathbb{E}\left[\mathcal{T}\mathcal{B}_{i}^{*}\mathcal{B}_{i}\mathcal{T}\right]\succeq\mathbb{E}\left[\mathcal{T}\mathcal{B}_{i}^{*}\mathcal{B}_{i}\mathcal{T}{\bf 1}_{E}\right]\succeq\mathbb{E}\left[\mathcal{T}\mathcal{B}_{i}^{*}\mathcal{B}_{i}\mathcal{T}{\bf 1}_{E\cap F_{i}}\right]$,
which yields 
\begin{equation}
\left\Vert \mathbb{E}\left[\mathcal{T}\mathcal{B}_{i}^{*}\mathcal{B}_{i}\mathcal{T}{\bf 1}_{E}\right]-\mathcal{T}\right\Vert \leq\left\Vert \mathbb{E}\left[\mathcal{T}\mathcal{B}_{i}^{*}\mathcal{B}_{i}\mathcal{T}{\bf 1}_{E\cap F_{i}}\right]-\mathcal{T}\right\Vert .\label{eq:BoundEcapF}
\end{equation}
Thus, it is sufficient to evaluate $\left\Vert \mathbb{E}\left[\mathcal{T}\mathcal{B}_{i}^{*}\mathcal{B}_{i}{\bf 1}_{E\cap F_{i}}\right]-\mathcal{\mathcal{T}}\right\Vert $.
To this end, we adopt an argument of similar spirit as \cite[Appendix B]{CandesPlan2011RIPless}.
Write 
\begin{align*}
\mathcal{\mathcal{T}} & =\mathbb{E}\left[\mathcal{\mathcal{T}}\mathcal{B}_{i}^{*}\mathcal{B}_{i}\mathcal{\mathcal{T}}\right]\\
 & =\mathbb{E}\left[\mathcal{\mathcal{T}}\mathcal{B}_{i}^{*}\mathcal{B}_{i}\mathcal{\mathcal{T}}{\bf 1}_{E\cap F_{i}}\right]+\mathbb{E}\left[\mathcal{\mathcal{T}}\mathcal{B}_{i}^{*}\mathcal{B}_{i}\mathcal{\mathcal{T}}{\bf 1}_{E^{c}\cup F_{i}^{c}}\right],
\end{align*}
and, consequently, 
\begin{align}
 & \left\Vert \mathbb{E}\left[\mathcal{\mathcal{T}}\mathcal{B}_{i}^{*}\mathcal{B}_{i}\mathcal{\mathcal{T}}{\bf 1}_{E\cap F_{i}}\right]-\mathcal{\mathcal{T}}\right\Vert \nonumber \\
 & \quad=\left\Vert \mathbb{E}\left[\mathcal{\mathcal{T}}\mathcal{B}_{i}^{*}\mathcal{B}_{i}\mathcal{\mathcal{T}}{\bf 1}_{E^{c}\cup F_{i}^{c}}\right]\right\Vert \label{eq:SeparateBB}\\
 & \quad\leq\left\Vert \mathbb{E}\left[\mathcal{\mathcal{T}}\mathcal{B}_{i}^{*}\mathcal{B}_{i}\mathcal{\mathcal{T}}{\bf 1}_{F_{i}\cap E^{c}}\right]\right\Vert +\left\Vert \mathbb{E}\left[\mathcal{\mathcal{T}}\mathcal{B}_{i}^{*}\mathcal{B}_{i}\mathcal{\mathcal{T}}{\bf 1}_{F_{i}^{c}}\right]\right\Vert ,\nonumber 
\end{align}
which allows us to bound $\left\Vert \mathbb{E}\left[\mathcal{\mathcal{T}}\mathcal{B}_{i}^{*}\mathcal{B}_{i}\mathcal{\mathcal{T}}{\bf 1}_{F_{i}\cap E^{c}}\right]\right\Vert $
and $\left\Vert \mathbb{E}\left[\mathcal{\mathcal{T}}\mathcal{B}_{i}^{*}\mathcal{B}_{i}\mathcal{\mathcal{T}}{\bf 1}_{F_{i}^{c}}\right]\right\Vert $
separately.

First, it follows from the identity $\left\Vert \mathcal{\mathcal{T}}\mathcal{B}_{i}^{*}\mathcal{B}_{i}\mathcal{\mathcal{T}}\right\Vert =\left\Vert \mathcal{\mathcal{T}}\left(\boldsymbol{B}_{i}\right)\right\Vert _{\text{F}}^{2}$
and the definition of the event $F_{i}$ that 
\begin{equation}
\left\Vert \mathbb{E}\left[\mathcal{\mathcal{T}}\mathcal{B}_{i}^{*}\mathcal{B}_{i}\mathcal{\mathcal{T}}{\bf 1}_{F_{i}\cap E^{c}}\right]\right\Vert \leq\left(20n\log n\right)^{2}\mathbb{P}\left(E^{c}\right)<\frac{1}{n^{2}}.\label{eq:BoundBB_FandEc}
\end{equation}

Second, applying the tail inequality on the quadratic form (e.g. \cite[Proposition 1.1]{hsu2012tail})
yields 
\begin{equation}
\mathbb{P}\left(\left\Vert \boldsymbol{A}_{i}\right\Vert _{\text{F}}\geq c_{20}\left(n+2\sqrt{nt}+2t\right)\right)\leq e^{-t}.\label{eq:SubgaussianFnorm}
\end{equation}
Thus, for any $t>\left(20n\log n\right)^{2}$, one has 
\begin{equation}
\mathbb{P}\left(\left\Vert \boldsymbol{A}_{i}\right\Vert _{\text{F}}\geq\sqrt{\frac{t}{3}}\right)\leq e^{-c_{21}\sqrt{t}}
\end{equation}
for some absolute constant $c_{21}>0$. Recall that $\left\Vert \boldsymbol{B}_{i}\right\Vert _{\text{F}}\leq\sqrt{3}\max\left\{ \left\Vert \boldsymbol{A}_{3i-2}\right\Vert _{\text{F}},\left\Vert \boldsymbol{A}_{3i-1}\right\Vert _{\text{F}},\left\Vert \boldsymbol{A}_{3i}\right\Vert _{\text{F}}\right\} $,
which indicates 
\begin{align*}
\mathbb{P}\left(\left\Vert \boldsymbol{B}_{i}\right\Vert _{\text{F}}^{2}\geq t\right) & \leq\mathbb{P}\left(\left\Vert \boldsymbol{A}_{3i-1}\right\Vert _{\text{F}}^{2}\geq\frac{t}{3}\right)+\mathbb{P}\left(\left\Vert \boldsymbol{A}_{3i-2}\right\Vert _{\text{F}}^{2}\geq\frac{t}{3}\right)\\
 & \quad\quad\quad\quad+\mathbb{P}\left(\left\Vert \boldsymbol{A}_{3i}\right\Vert _{\text{F}}^{2}\geq\frac{t}{3}\right)\\
 & \leq3\mathbb{P}\left(\left\Vert \boldsymbol{A}_{i}\right\Vert _{\text{F}}\geq\sqrt{\frac{t}{3}}\right)\\
 & \leq3e^{-c_{21}\sqrt{t}}:=g(t).
\end{align*}
A similar approach as introduced in \cite[Appendix B]{CandesPlan2011RIPless}
gives 
\begin{align}
 & \left\Vert \mathbb{E}\left[\mathcal{\mathcal{T}}\mathcal{B}_{i}^{*}\mathcal{B}_{i}\mathcal{\mathcal{T}}{\bf 1}_{F_{i}^{c}}\right]\right\Vert \leq\mathbb{E}\left[\left\Vert \boldsymbol{B}_{i}\right\Vert _{\text{F}}^{2}{\bf 1}_{F_{i}^{c}}\right]\nonumber \\
 & \quad\leq\left(20n\log n\right)^{2}g\left(\left(20n\log n\right)^{2}\right)+{\displaystyle \int}_{\left(20n\log n\right)^{2}}^{\infty}g\left(t\right)\mathrm{d}t\nonumber \\
 & \quad<\left(20n\log n\right)^{2}g\left(\left(20n\log n\right)^{2}\right)+{\displaystyle \int}_{\left(20n\log n\right)^{2}}^{\infty}\frac{1}{t^{5}}\mathrm{d}t\nonumber \\
 & \quad<\frac{c_{15}}{n^{2}}\label{eq:BoundBB_Fc}
\end{align}
for some absolute constant $c_{15}>0$. This taken collectively with
(\ref{eq:BoundEcapF}), (\ref{eq:SeparateBB}) and (\ref{eq:BoundBB_FandEc})
yields 
\[
\left\Vert \mathbb{E}\left[\mathcal{\mathcal{T}}\mathcal{B}_{i}^{*}\mathcal{B}_{i}\mathcal{\mathcal{T}}{\bf 1}_{E}\right]-\mathcal{\mathcal{T}}\right\Vert \leq\left\Vert \mathbb{E}\left[\mathcal{\mathcal{T}}\mathcal{B}_{i}^{*}\mathcal{B}_{i}\mathcal{\mathcal{T}}{\bf 1}_{E\cap F_{i}}\right]-\mathcal{\mathcal{T}}\right\Vert \leq\frac{\tilde{c}_{15}}{n^{2}}
\]
for some absolute constant $\tilde{c}_{15}>0$.

\section{Proof of Lemma \ref{lemma:EstimatePAAP}\label{sec:Proof-of-Lemma-EstimatePAAP}}

Dudley's inequality \cite[Theorem 11.17]{ledoux1991probability} allows
us to bound the supremum of the Gaussian process as follows 
\begin{align}
 & \mathbb{E}\left[\left.\sup_{T\in\mathcal{M}_{r}^{\mathrm{t}},\boldsymbol{X}\in T,\left\Vert \boldsymbol{X}\right\Vert _{\text{F}}=1}\Bigg|\sum_{i=1}^{m}g_{i}\left|\mathcal{B}_{i}\left(\boldsymbol{X}\right)\right|^{2}\Bigg|\text{ }\right|\mathcal{B}_{i}\text{ }(1\leq i\leq m)\right]\nonumber \\
 & \quad\leq24{\displaystyle \int}_{0}^{\infty}\log^{\frac{1}{2}}N\left(\mathcal{D}_{2r}^{2},d\left(\cdot,\cdot\right),u\right)\mathrm{d}u,\label{eq:EntropyBoundGaussian}
\end{align}
where $\mathcal{D}_{r}^{2}:=\left\{ \boldsymbol{X}\mid\left\Vert \boldsymbol{X}\right\Vert _{\text{F}}=1,\mathrm{rank}\left(\boldsymbol{X}\right)\leq2r\right\} $.
Here, $N\left(\mathcal{Z},d\left(\cdot,\cdot\right),u\right)$ denotes
the smallest number of balls of radius $u$ centered in points of
$\mathcal{Z}$ needed to cover the set $\mathcal{Z}$, under the pseudo
metric $d\left(\cdot,\cdot\right)$ defined as follows 
\begin{align*}
d\left(\boldsymbol{X},\boldsymbol{Y}\right): & =\sqrt{\sum_{i=1}^{m}\left(\left|\mathcal{B}_{i}\left(\boldsymbol{X}\right)\right|^{2}-\left|\mathcal{B}_{i}\left(\boldsymbol{Y}\right)\right|^{2}\right)^{2}}.
\end{align*}
For any $\left(\boldsymbol{X},\boldsymbol{Y}\right)$ that satisfy
$\left\Vert \boldsymbol{X}\right\Vert _{\text{F}}=\left\Vert \boldsymbol{Y}\right\Vert _{\text{F}}=1$,
$\text{rank}\left(\boldsymbol{X}\right)\leq r$ and $\text{rank}\left(\boldsymbol{Y}\right)\leq r$,
the pseudo metric satisfies 
\begin{align*}
 & d\left(\boldsymbol{X},\boldsymbol{Y}\right)\leq\sqrt{\left(\max_{i:1\leq i\leq m}\left|\mathcal{B}_{i}\left(\boldsymbol{X}-\boldsymbol{Y}\right)\right|^{2}\right)\sum_{i=1}^{m}\left|\mathcal{B}_{i}\left(\boldsymbol{X}+\boldsymbol{Y}\right)\right|^{2}}\\
 & \quad\leq\sqrt{2}\sqrt{\sum_{i=1}^{m}\left|\mathcal{B}_{i}\left(\boldsymbol{X}\right)\right|^{2}+\left|\mathcal{B}_{i}\left(\boldsymbol{Y}\right)\right|^{2}}\max_{i:1\leq i\leq m}\left|\mathcal{B}_{i}\left(\boldsymbol{X}-\boldsymbol{Y}\right)\right|\\
 & \quad\leq\small\left\{ \sqrt{\left\langle \boldsymbol{X},\left(\sum_{i=1}^{m}\mathcal{B}_{i}^{*}\mathcal{B}_{i}\right)\left(\boldsymbol{X}\right)\right\rangle }+\sqrt{\left\langle \boldsymbol{Y},\left(\sum_{i=1}^{m}\mathcal{B}_{i}^{*}\mathcal{B}_{i}\right)\left(\boldsymbol{Y}\right)\right\rangle }\right\} \\
 & \quad\quad\quad\quad\cdot\sqrt{2}\max_{i:1\leq i\leq m}\left|\mathcal{B}_{i}\left(\boldsymbol{X}-\boldsymbol{Y}\right)\right|\\
 & \quad\leq2\sqrt{2}\sup_{T:T\in\mathcal{M}_{r}^{\mathrm{t}}}\sqrt{\left\Vert \sum_{i=1}^{m}\mathcal{P}_{T}\mathcal{B}_{i}^{*}\mathcal{B}_{i}\mathcal{P}_{T}\right\Vert }\max_{i:1\leq i\leq m}\left|\left\langle \boldsymbol{B}_{i},\boldsymbol{X}-\boldsymbol{Y}\right\rangle \right|,
\end{align*}
where the last inequality relies on the observation that $\left\Vert \boldsymbol{X}\right\Vert _{\text{F}}=\left\Vert \boldsymbol{Y}\right\Vert _{\text{F}}=1$.

If we introduce the quantity 
\begin{equation}
R:=\sup_{T:T\in\mathcal{M}_{r}^{\mathrm{t}}}\sqrt{\left\Vert \sum_{i=1}^{m}\mathcal{P}_{T}\mathcal{B}_{i}^{*}\mathcal{B}_{i}\mathcal{P}_{T}\right\Vert }\label{eq:DefnR}
\end{equation}
and define another pseudo metric $\left\Vert \cdot\right\Vert _{\mathcal{B}}$
as 
\begin{equation}
\left\Vert \boldsymbol{X}\right\Vert _{\mathcal{B}}:=\max_{i:1\leq i\leq m}\left|\left\langle \boldsymbol{B}_{i},\boldsymbol{X}\right\rangle \right|,\label{eq:MetricA}
\end{equation}
then $d\left(\boldsymbol{X},\boldsymbol{Y}\right)\leq2\sqrt{2}R\left\Vert \boldsymbol{X}-\boldsymbol{Y}\right\Vert _{\mathcal{B}}$,
which allows us to bound 
\begin{align}
 & {\displaystyle \int}_{0}^{\infty}\log^{\frac{1}{2}}N\left(\mathcal{D}_{2r}^{2},d\left(\cdot,\cdot\right),u\right)\mathrm{d}u\nonumber \\
 & \quad\leq{\displaystyle \int}_{0}^{\infty}\log^{\frac{1}{2}}N\left(\mathcal{D}_{2r}^{2},2\sqrt{2}R\left\Vert \cdot\right\Vert _{\mathcal{B}},u\right)\mathrm{d}u\nonumber \\
 & \quad={\displaystyle \int}_{0}^{\infty}\log^{\frac{1}{2}}N\left(\frac{1}{\sqrt{2r}}\mathcal{D}_{2r}^{2},\left\Vert \cdot\right\Vert _{\mathcal{B}},\frac{u}{4R\sqrt{r}}\right)\mathrm{d}u\nonumber \\
 & \quad\leq{\displaystyle \int}_{0}^{\infty}\log^{\frac{1}{2}}N\left(\mathcal{D}_{2r}^{1},\left\Vert \cdot\right\Vert _{\mathcal{B}},\frac{u}{4R\sqrt{r}}\right)\mathrm{d}u\nonumber \\
 & \quad\leq\text{ }4R\sqrt{r}{\displaystyle \int}_{0}^{\infty}\log^{\frac{1}{2}}N\left(\mathcal{D}^{1},\left\Vert \cdot\right\Vert _{\mathcal{B}},u\right)\mathrm{d}u.\label{eq:EntropyBound}
\end{align}
Here, $\mathcal{D}_{r}^{1}$ and $\mathcal{D}^{1}$ stand for 
\begin{align*}
\mathcal{D}_{r}^{1}: & =\left\{ \boldsymbol{X}\mid\left\Vert \boldsymbol{X}\right\Vert _{*}\leq1,\text{rank}\left(\boldsymbol{X}\right)\leq r\right\} ,\\
\mathcal{D}^{1}: & =\left\{ \boldsymbol{X}\mid\left\Vert \boldsymbol{X}\right\Vert _{*}\leq1\right\} ,
\end{align*}
and we have exploited the containment 
\[
\frac{1}{\sqrt{2r}}\mathcal{D}_{2r}^{2}\subseteq\mathcal{D}_{2r}^{1}\subseteq\mathcal{D}^{1}.
\]
Hence it suffices to bound 
\[
E_{2}:=4R\sqrt{r}{\displaystyle \int}_{0}^{\infty}\log^{\frac{1}{2}}N\left(\mathcal{D}^{1},\left\Vert \cdot\right\Vert _{\mathcal{B}},u\right)\mathrm{d}u.
\]

It remains to bound the covering number (or metric entropy) of the
nuclear-norm ball $\mathcal{D}^{1}$. Repeating the well-known procedure
as in \cite[Page 1113]{aubrun2009almost} yields 
\begin{align*}
{\displaystyle \int}_{0}^{\infty}\sqrt{\log N\left(\mathcal{D}^{1},\left\Vert \cdot\right\Vert _{\mathcal{B}},u\right)}\mathrm{d}u & \leq C_{10}K\left(\log n\right)^{5/2}\sqrt{\log m}\\
 & \leq C_{11}K\log^{3}n
\end{align*}
for some constants $C_{10},C_{11}>0$. This taken collectively with
(\ref{eq:EntropyBoundGaussian}) and (\ref{eq:EntropyBound}) gives
that conditioning on $\mathcal{B}_{i}$'s, one has 
\begin{align}
 & \mathbb{E}\left[\left.\sup_{T\in\mathcal{M}_{r}^{\mathrm{t}}}\Bigg\|\mathcal{P}_{T}\left(\sum_{i=1}^{m}g_{i}\mathcal{B}_{i}^{*}\mathcal{B}_{i}\right)\mathcal{P}_{T}\Bigg\|\text{ }\right|\mathcal{B}_{i}\text{ }(1\leq i\leq m)\right]\nonumber \\
 & \quad\leq C_{14}\sqrt{r}K\log^{3}n\sqrt{\sup_{T:T\in\mathcal{M}_{r}^{\mathrm{t}}}\left\Vert \sum_{i=1}^{m}\mathcal{P}_{T}\mathcal{B}_{i}^{*}\mathcal{B}_{i}\mathcal{P}_{T}\right\Vert }.\label{eq:EstimatePAAP}
\end{align}
for some absolute constant $C_{14}>0$. \bibliographystyle{IEEEtran}
\bibliography{bibfileToeplitzPR}

\begin{thebibliography}{10}
\providecommand{\url}[1]{#1}
\csname url@samestyle\endcsname
\providecommand{\newblock}{\relax}
\providecommand{\bibinfo}[2]{#2}
\providecommand{\BIBentrySTDinterwordspacing}{\spaceskip=0pt\relax}
\providecommand{\BIBentryALTinterwordstretchfactor}{4}
\providecommand{\BIBentryALTinterwordspacing}{\spaceskip=\fontdimen2\font plus
\BIBentryALTinterwordstretchfactor\fontdimen3\font minus
  \fontdimen4\font\relax}
\providecommand{\BIBforeignlanguage}[2]{{%
\expandafter\ifx\csname l@#1\endcsname\relax
\typeout{** WARNING: IEEEtran.bst: No hyphenation pattern has been}%
\typeout{** loaded for the language `#1'. Using the pattern for}%
\typeout{** the default language instead.}%
\else
\language=\csname l@#1\endcsname
\fi
#2}}
\providecommand{\BIBdecl}{\relax}
\BIBdecl

\bibitem{chen2014icassp}
Y.~Chen, Y.~Chi, and A.~Goldsmith, ``Estimation of simultaneously structured
  covariance matrices from quadratic measurements,'' in \emph{International
  Conference on Acoustics, Speech, and Signal Processing (ICASSP)}, Florence,
  Italy, May 2014, pp. 7669 -- 7673.

\bibitem{chen2014isit}
------, ``Robust and universal covariance estimation from quadratic
  measurements via convex programming,'' in \emph{International Symposium on
  Information Theory (ISIT)}, Honolulu, HI, June 2014, pp. 2017 -- 2021.

\bibitem{RoyKailathESPIRIT1989}
R.~Roy and T.~Kailath, ``{ESPRIT}-estimation of signal parameters via
  rotational invariance techniques,'' \emph{IEEE Transactions on Acoustics,
  Speech and Signal Processing}, vol.~37, no.~7, pp. 984 --995, Jul 1989.

\bibitem{karoui2008operator}
N.~Karoui, ``Operator norm consistent estimation of large-dimensional sparse
  covariance matrices,'' \emph{The Annals of Statistics}, pp. 2717--2756, 2008.

\bibitem{muthukrishnan2005data}
S.~Muthukrishnan, \emph{Data streams: Algorithms and applications}.\hskip 1em
  plus 0.5em minus 0.4em\relax Now Publishers Inc, 2005.

\bibitem{daniels200760}
R.~C. Daniels and R.~W. Heath, ``60 {GHz} wireless communications: emerging
  requirements and design recommendations,'' \emph{IEEE Vehicular Technology
  Magazine}, vol.~2, no.~3, pp. 41--50, 2007.

\bibitem{jung2013compressive}
A.~Jung, G.~Taubock, and F.~Hlawatsch, ``Compressive spectral estimation for
  nonstationary random processes,'' \emph{IEEE Transactions on Information
  Theory}, vol.~59, no.~5, May 2013.

\bibitem{Leus2011}
G.~Leus and Z.~Tian, ``Recovering second-order statistics from compressive
  measurements,'' in \emph{Computational Advances in Multi-Sensor Adaptive
  Processing (CAMSAP),}, 2011, pp. 337--340.

\bibitem{ariananda2012compressive}
D.~D. Ariananda and G.~Leus, ``Compressive wideband power spectrum
  estimation,'' \emph{IEEE Transactions on Signal Processing}, vol.~60, no.~9,
  pp. 4775--4789, 2012.

\bibitem{gallager1968information}
R.~G. Gallager, ``Information theory and reliable communication,'' 1968.

\bibitem{tropp2010beyond}
J.~A. Tropp, J.~N. Laska, M.~F. Duarte, J.~K. Romberg, and R.~G. Baraniuk,
  ``Beyond nyquist: Efficient sampling of sparse bandlimited signals,''
  \emph{Information Theory, IEEE Transactions on}, vol.~56, no.~1, pp.
  520--544, 2010.

\bibitem{scharf1994matched}
L.~L. Scharf and B.~Friedlander, ``Matched subspace detectors,'' \emph{IEEE
  Transactions on Signal Processing}, vol.~42, no.~8, pp. 2146--2157, 1994.

\bibitem{raymer1994complex}
M.~Raymer, M.~Beck, and D.~McAlister, ``Complex wave-field reconstruction using
  phase-space tomography,'' \emph{Physical review letters}, vol.~72, no.~8, p.
  1137, 1994.

\bibitem{tian2012experimental}
L.~Tian, J.~Lee, S.~B. Oh, and G.~Barbastathis, ``Experimental compressive
  phase space tomography,'' \emph{Optics Express}, vol.~20, no.~8, p. 8296,
  2012.

\bibitem{candes2012phaselift}
E.~J. Candes, T.~Strohmer, and V.~Voroninski, ``Phaselift: Exact and stable
  signal recovery from magnitude measurements via convex programming,''
  \emph{Communications on Pure and Applied Mathematics}, 2012.

\bibitem{shechtman2011sparsity}
Y.~Shechtman, Y.~C. Eldar, A.~Szameit, and M.~Segev, ``Sparsity based
  sub-wavelength imaging with partially incoherent light via quadratic
  compressed sensing,'' \emph{Optics express}, vol.~19, no.~16, pp.
  14\,807--14\,822, 2011.

\bibitem{waldspurger2012phase}
I.~Waldspurger, A.~d'Aspremont, and S.~Mallat, ``Phase recovery, maxcut and
  complex semidefinite programming,'' \emph{arXiv preprint arXiv:1206.0102},
  2012.

\bibitem{candes2014wirtinger}
E.~Candes, X.~Li, and M.~Soltanolkotabi, ``Phase retrieval via {Wirtinger}
  flow: Theory and algorithms,'' \emph{IEEE Transactions on Information
  Theory}, vol.~61, no.~4, pp. 1985--2007, April 2015.

\bibitem{shechtman2013gespar}
Y.~Shechtman, A.~Beck, and Y.~C. Eldar, ``{GESPAR}: Efficient phase retrieval
  of sparse signals,'' \emph{arXiv preprint arXiv:1301.1018}, 2013.

\bibitem{netrapalli2013phase}
P.~Netrapalli, P.~Jain, and S.~Sanghavi, ``Phase retrieval using alternating
  minimization,'' \emph{Advances in Neural Information Processing Systems
  (NIPS)}, 2013.

\bibitem{Schniter2015}
P.~Schniter and S.~Rangan, ``Compressive phase retrieval via generalized
  approximate message passing,'' \emph{Signal Processing, IEEE Transactions
  on}, vol.~63, no.~4, pp. 1043--1055, Feb 2015.

\bibitem{chen2014convex}
Y.~Chen, X.~Yi, and C.~Caramanis, ``A convex formulation for mixed regression
  with two components: Minimax optimal rates,'' in \emph{The Conference on
  Learning Theory (COLT)}, 2014.

\bibitem{candes2012solving}
E.~J. Candes and X.~Li, ``Solving quadratic equations via {PhaseLift} when
  there are about as many equations as unknowns,'' \emph{Foundations of
  Computational Mathematics}, 2013.

\bibitem{li2012sparse}
X.~Li and V.~Voroninski, ``Sparse signal recovery from quadratic measurements
  via convex programming,'' \emph{SIAM Journal on Mathematical Analysis}, 2013.

\bibitem{rudelson2008sparse}
M.~Rudelson and R.~Vershynin, ``On sparse reconstruction from {Fourier and
  Gaussian} measurements,'' \emph{Communications on Pure and Applied
  Mathematics}, vol.~61, no.~8, pp. 1025--1045, 2008.

\bibitem{Gross2011recovering}
D.~Gross, ``Recovering low-rank matrices from few coefficients in any basis,''
  \emph{IEEE Transactions on Information Theory}, vol.~57, no.~3, pp.
  1548--1566, March 2011.

\bibitem{liu2011universal}
Y.-K. Liu, ``Universal low-rank matrix recovery from pauli measurements,'' in
  \emph{Advances in Neural Information Processing Systems}, 2011, pp.
  1638--1646.

\bibitem{ExactMC09}
E.~J. Candes and B.~Recht, ``Exact matrix completion via convex optimization,''
  \emph{Foundations of Computational Mathematics}, vol.~9, no.~6, pp. 717--772,
  April 2009.

\bibitem{chen2013robust}
Y.~Chen and Y.~Chi, ``Robust spectral compressed sensing via structured matrix
  completion,'' \emph{IEEE Trans. on Information Theory}, vol.~60, no.~10, pp.
  6576 -- 6601, October 2014.

\bibitem{johnstone2006high}
I.~M. Johnstone, ``High dimensional statistical inference and random
  matrices,'' in \emph{Proceedings of the International Congress of
  Mathematicians: Madrid, August 22-30, 2006: invited lectures}, 2006, pp.
  307--333.

\bibitem{Don2006}
D.~Donoho, ``Compressed sensing,'' \emph{IEEE Transactions on Information
  Theory}, vol.~52, no.~4, pp. 1289 --1306, April 2006.

\bibitem{CandTao06}
E.~J. Candes and T.~Tao, ``Near-optimal signal recovery from random
  projections: Universal encoding strategies?'' \emph{IEEE Transactions on
  Information Theory}, vol.~52, no.~12, pp. 5406--5425, Dec. 2006.

\bibitem{tian2012cyclic}
Z.~Tian, Y.~Tafesse, and B.~M. Sadler, ``Cyclic feature detection with
  sub-nyquist sampling for wideband spectrum sensing,'' \emph{IEEE Journal of
  Selected Topics in Signal Processing}, vol.~6, no.~1, pp. 58--69, 2012.

\bibitem{romero2013wideband}
D.~Romero and G.~Leus, ``Wideband spectrum sensing from compressed measurements
  using spectral prior information,'' \emph{Signal Processing, IEEE
  Transactions on}, vol.~61, no.~24, pp. 6232--6246, 2013.

\bibitem{dasarathy2013sketching}
G.~Dasarathy, P.~Shah, B.~N. Bhaskar, and R.~Nowak, ``Sketching sparse
  matrices,'' \emph{arXiv preprint arXiv:1303.6544}, 2013.

\bibitem{beck2012sparsity}
A.~Beck and Y.~C. Eldar, ``Sparsity constrained nonlinear optimization:
  Optimality conditions and algorithms,'' \emph{to appear in SIAM
  Optimization}, 2013.

\bibitem{jaganathan2012recovery}
K.~Jaganathan, S.~Oymak, and B.~Hassibi, ``Recovery of sparse 1-{D} signals
  from the magnitudes of their {F}ourier transform,'' in \emph{IEEE ISIT},
  2012, pp. 1473--1477.

\bibitem{oja1983subspace}
E.~Oja, \emph{Subspace methods of pattern recognition}.\hskip 1em plus 0.5em
  minus 0.4em\relax Research Studies Press England, 1983, vol.~4.

\bibitem{Caramanis2013onlinePCA}
I.~Mitliagkas, C.~Caramanis, and P.~Jain, ``Memory limited, streaming {PCA},''
  \emph{Advances in Neural Information Processing Systems (NIPS)}, 2013.

\bibitem{Balzano-2010}
L.~Balzano, R.~Nowak, and B.~Recht, ``Online identification and tracking of
  subspaces from highly incomplete information,'' \emph{Proc. Allerton 2010},
  2010.

\bibitem{chi2012petrels}
Y.~Chi, Y.~Eldar, and R.~Calderbank, ``{PETRELS}: Parallel subspace estimation
  and tracking by recursive least squares from partial observations,''
  \emph{Signal Processing, IEEE Transactions on}, vol.~61, no.~23, pp.
  5947--5959, 2013.

\bibitem{brand2002incremental}
M.~Brand, ``Incremental singular value decomposition of uncertain data with
  missing values,'' in \emph{ECCV 2002}.\hskip 1em plus 0.5em minus 0.4em\relax
  Springer, 2002, pp. 707--720.

\bibitem{cai2013rop}
T.~T. Cai and A.~Zhang, ``Rop: Matrix recovery via rank-one projections,''
  \emph{arXiv preprint arXiv:1310.5791}, 2014.

\bibitem{RecFazPar07}
B.~Recht, M.~Fazel, and P.~A. Parrilo, ``Guaranteed minimum-rank solutions of
  linear matrix equations via nuclear norm minimization,'' \emph{SIAM Review},
  vol.~52, no.~3, pp. 471--501, 2010.

\bibitem{ohlsson2011compressive}
H.~Ohlsson, A.~Y. Yang, R.~Dong, and S.~S. Sastry, ``Compressive phase
  retrieval from squared output measurements via semidefinite programming,''
  \emph{arXiv preprint arXiv:1111.6323}, 2011.

\bibitem{candes2008restricted}
E.~J. Cand{\`e}s, ``The restricted isometry property and its implications for
  compressed sensing,'' \emph{Comptes Rendus Mathematique}, vol. 346, no.~9,
  pp. 589--592, 2008.

\bibitem{stewart1990matrix}
G.~W. Stewart and J.-g. Sun, ``Matrix perturbation theory,'' 1990.

\bibitem{jafarpour2009efficient}
S.~Jafarpour, W.~Xu, B.~Hassibi, and R.~Calderbank, ``Efficient and robust
  compressed sensing using optimized expander graphs,'' \emph{IEEE Transactions
  on Information Theory}, vol.~55, no.~9, pp. 4299--4308, 2009.

\bibitem{CandesPlan2011Tight}
E.~Candes and Y.~Plan, ``Tight oracle inequalities for low-rank matrix recovery
  from a minimal number of noisy random measurements,'' \emph{IEEE Transactions
  on Information Theory}, vol.~57, no.~4, pp. 2342--2359, 2011.

\bibitem{baraniuk2008simple}
R.~Baraniuk, M.~Davenport, R.~DeVore, and M.~Wakin, ``A simple proof of the
  restricted isometry property for random matrices,'' \emph{Constructive
  Approximation}, vol.~28, no.~3, pp. 253--263, 2008.

\bibitem{talagrand1996majorizing}
M.~Talagrand, ``Majorizing measures: the generic chaining,'' \emph{The Annals
  of Probability}, pp. 1049--1103, 1996.

\bibitem{CandesPlan2011RIPless}
E.~Candes and Y.~Plan, ``A probabilistic and {RIP}less theory of compressed
  sensing,'' \emph{IEEE Transactions on Information Theory}, vol.~57, no.~11,
  pp. 7235--7254, 2011.

\bibitem{Tao2012RMT}
T.~Tao, \emph{Topics in Random Matrix Theory}, ser. Graduate Studies in
  Mathematics.\hskip 1em plus 0.5em minus 0.4em\relax Providence, Rhode Island:
  American Mathematical Society, 2012.

\bibitem{demanet2012stable}
L.~Demanet and P.~Hand, ``Stable optimizationless recovery from phaseless
  linear measurements,'' \emph{arXiv preprint arXiv:1208.1803}, 2012.

\bibitem{grenander1958toeplitz}
U.~Grenander and G.~Szeg{\H{o}}, \emph{Toeplitz forms and their
  applications}.\hskip 1em plus 0.5em minus 0.4em\relax Univ of California
  Press, 1958.

\bibitem{Vershynin2012}
R.~Vershynin, ``Introduction to the non-asymptotic analysis of random
  matrices,'' \emph{Compressed Sensing, Theory and Applications}, pp. 210 --
  268, 2012.

\bibitem{hanson1971bound}
D.~L. Hanson and F.~T. Wright, ``A bound on tail probabilities for quadratic
  forms in independent random variables,'' \emph{The Annals of Mathematical
  Statistics}, vol.~42, no.~3, pp. 1079--1083, 1971.

\bibitem{rudelson2013hanson}
M.~Rudelson and R.~Vershynin, ``{H}anson-{W}right inequality and sub-{G}aussian
  concentration,'' \emph{arXiv preprint arXiv:1306.2872}, June 2013.

\bibitem{hsu2012tail}
D.~Hsu, S.~M. Kakade, and T.~Zhang, ``A tail inequality for quadratic forms of
  sub-{G}aussian random vectors,'' \emph{Electronic Communications in
  Probability}, vol.~17, no.~52, 2012.

\bibitem{ledoux1991probability}
M.~Ledoux and M.~Talagrand, \emph{Probability in {B}anach Spaces: isoperimetry
  and processes}.\hskip 1em plus 0.5em minus 0.4em\relax Springer, 1991,
  vol.~23.

\bibitem{aubrun2009almost}
G.~Aubrun, ``On almost randomizing channels with a short {K}raus
  decomposition,'' \emph{Communications in Mathematical Physics}, vol. 288,
  no.~3, pp. 1103--1116, 2009.

\end{thebibliography}

\begin{IEEEbiographynophoto}{Yuxin Chen} (S'09) received the B.S. in Microelectronics with High Distinction from Tsinghua University in 2008, the M.S. in Electrical and Computer Engineering from the University of Texas at Austin in 2010, the M.S. in Statistics from Stanford University in 2013, and the Ph.D. in Electrical Engineering from Stanford University in 2015. He is currently a postdoctoral scholar in the Department of Statistics at Stanford University. His research interests include high-dimensional structured estimation, information theory, compressed sensing, and network science. \end{IEEEbiographynophoto} 

\begin{IEEEbiographynophoto}{Yuejie Chi}
(S'09-M'12)
received the Ph.D. degree in Electrical Engineering from Princeton University in 2012, and the B.E. (Hon.) degree in Electrical Engineering from Tsinghua University, Beijing, China, in 2007. Since September 2012, she has been an assistant professor with the department of Electrical and Computer Engineering and the department of Biomedical Informatics at the Ohio State University.

She is the recipient of the IEEE Signal Processing Society Young Author Best Paper Award in 2013 and the Best Paper Award at the IEEE International Conference on Acoustics, Speech, and Signal Processing (ICASSP) in 2012. She received AFOSR Young Investigator Program Award in 2015, the Ralph E. Powe Junior Faculty Enhancement Award from Oak Ridge Associated Universities in 2014, a Google Faculty Research Award in 2013, the Roberto Padovani scholarship from Qualcomm Inc. in 2010, and an Engineering Fellowship from Princeton University in 2007. She has held visiting positions at Colorado State University, Stanford University and Duke University, and interned at Qualcomm Inc. and Mitsubishi Electric Research Lab. Her research interests include statistical signal processing, information theory, machine learning and their applications in high-dimensional data analysis, active sensing and bioinformatics.

\end{IEEEbiographynophoto}

\begin{IEEEbiographynophoto}{Andrea J. Goldsmith} (S'90-M'93-SM'99-F'05) is the Stephen Harris professor of Electrical Engineering at Stanford University. She was previously on the faculty of Electrical Engineering at Caltech. Her research interests are in information theory and communication theory, and their application to wireless communications and related fields. Dr. Goldsmith co-founded and served as CTO for two wireless companies: Accelera, Inc., which develops software-defined wireless network technology for cloud-based management of WiFi  access points, and  Quantenna Communications, Inc., which develops high-performance WiFi chipsets. She has also held industry positions at Maxim Technologies, Memorylink Corporation, and AT\&T Bell Laboratories. She is a Fellow of the IEEE and of Stanford, and  has received several awards for her work, including  the IEEE ComSoc Edwin H. Armstrong Achievement Award as well as Technical Achievement Awards in Communications Theory and in Wireless Communications,  the National Academy of Engineering Gilbreth Lecture Award, the IEEE ComSoc and Information Theory Society Joint Paper Award, the IEEE ComSoc Best Tutorial Paper Award, the Alfred P. Sloan Fellowship, and the Silicon Valley/San Jose Business Journal's Women of Influence Award. She is author of the book ``Wireless Communications'' and co-author of the books ``MIMO Wireless Communications'' and ``Principles of Cognitive Radio,'' all published by Cambridge University Press, as well as an inventor on 28 patents. She received the B.S., M.S. and Ph.D. degrees in Electrical Engineering from U.C. Berkeley.  

Dr. Goldsmith has served as editor for the IEEE Transactions on Information Theory, the Journal on Foundations and Trends in Communications and Information Theory and in Networks, the IEEE Transactions on Communications, and the IEEE Wireless Communications Magazine as well as on the Steering Committee for the IEEE Transactions on Wireless Communications. She participates actively in committees and conference organization for the IEEE Information Theory and Communications Societies and has served on the Board of Governors for both societies. She has also been a Distinguished Lecturer for both societies, served as President of the IEEE Information Theory Society in 2009, founded and chaired the student committee of the IEEE Information Theory society, and chaired the Emerging Technology Committee of the IEEE Communications Society. At Stanford she received the inaugural University Postdoc Mentoring Award, served as Chair of Stanfords Faculty Senate in 2009, and currently serves on its Faculty Senate, Budget Group, and Task Force on Women and Leadership. 
\end{IEEEbiographynophoto}
\end{document}